\documentclass[preprint,10pt,3p]{elsarticle}

\usepackage{amssymb}
\journal{}

\addtocounter{section}{0}

\usepackage{amsmath}

\usepackage{mathdots}
\usepackage{interval}
\usepackage{mathrsfs}

\usepackage{amsthm}
\usepackage{latexsym}
\usepackage{amsmath}
\usepackage{amssymb}
\usepackage{proof}
\usepackage{bbm}
\usepackage{bm}

\usepackage[all]{xy}

\usepackage{array,cmll}

\newtheorem{thm}{\bf Theorem}[section]
\newtheorem{exam}[thm]{\bf Example}
\newtheorem{prop}[thm]{\bf Proposition}
\newtheorem{lem}[thm]{\bf Lemma}
\newtheorem{cor}[thm]{Corollary}
\newtheorem{defn}[thm]{\bf Definition}

\newtheorem{rem}[thm]{\bf Remark}

\newcommand{\Comp}{\circ}%{\raisebox{.3ex}{\tiny o}} 
\newcommand{\Mes}{{\sf M}_{{\mc E}}}
\newcommand{\Ban}{{\sf Vec}_{{\mc E}^b}}
\newcommand{\Rel}{\sf{Rel}}
\newcommand{\GRel}{{\bf G}(\Rel)}

\newcommand{\SRel}{\sf{SRel}}

\newcommand{\TKer}{{\sf TKer}}
\newcommand{\TKerfin}{{\sf Tker}}
\newcommand{\TKeromg}{\TKer_{\omega}}  %@@

\newcommand{\Pcoh}{\sf{Pcoh}}
\newcommand{\Po}[1]{\sf{P} #1}

\newcommand{\ms}[1]{(#1, \mathcal{#1})}

\newcommand{\mf}[1]{\mc{E} (\mc{#1})}

\newcommand{\msnce}[1]{(#1_e^\bullet, \mathcal{#1}_e^\bullet)}

\newcommand{\mse}[1]{(#1_e, \mathcal{#1}_e)}

\newcommand{\msee}[1]{(#1_{ee}, \mathcal{#1}_{ee})}

\newcommand{\msed}[1]{(#1^\bullet_e, \mathcal{#1}^\bullet_e)}

\newcommand{\msi}[2]{(#1_{#2}, \mathcal{#1}_{#2})}

\newcommand{\zeroinf}{\overline{\mathbb{R}}_{+}}

\newcommand{\inj}{{\sf in}}
\newcommand{\pr}{{\sf pr}}

\newcommand{\mon}[2]{{\sf m}_{#1, #2}}
\newcommand{\di}[1]{{\sf d}_{#1}}

\newcommand{\dii}{{\sf d}_{\mathcal{I}}}
\newcommand{\moni}{{\sf m}_{\mathcal{I}}}
\newcommand{\stor}[1]{{\sf s}_{#1}}
\newcommand{\wk}[1]{{\sf w}_{#1}}
\newcommand{\con}[1]{{\sf c}_{#1}}

\newcommand{\msbf}[1]{\mathbbm{#1}}%{\bm{#1}}

\newcommand{\abs}[1]{\mid \!{#1} \! \mid}

\newcommand{\mc}[1]{\mathcal{#1}}

\newcommand{\ort}[3]{#1 \, \, \bot_{#2} \, \, \, #3}

\newcommand{\inpro}[3]{\langle #1  \! \mid  \! #3 \rangle_{#2}}

\newcommand{\inprocoh}[2]{\langle #1  , #2 \rangle}

\newcommand{\cunt}[1]{{\sf n}_{#1}}

\newcommand{\borel}{\mc{B}_{+}}

\newcommand{\ncdot}{\raisebox{.3mm}{\mbox{\tiny $\bullet$}}}

\newcommand{\FI}{F^{\mbox{-}1}}

\newcommand{\TKersfin}{{\sf TsKer}}

\newcommand{\TKersfinomg}{\TKersfin_{\omega}}

\newcommand{\opTKer}{{\sf TKer}^{\mbox{\tiny \sf op}}}
\newcommand{\opTKersfin}{{\sf TsKer}^{\mbox{\tiny \sf op}}
}
\newcommand{\GopTKersfin}{{\bf G} (\opTKersfin)}
\newcommand{\opTKersfinomg}{\opTKersfin_{\omega}}
\newcommand{\GopTKer}{{\bf G} (\opTKer)}

\newcommand{\GopTKersfinomg}{{\bf G} (\opTKersfinomg)}
\newcommand{\opTKeromg}{\opTKer_{\omega}}  %@@
\newcommand{\opsymbol}{\mbox{\tiny \sf op}}
\newcommand{\natkappa}[2]{\mathsf{k}_{#1}{(#2)}}

\newcommand{\Dd}[2]{\delta(#2, \,#1)}
\newcommand{\absI}[1]{\abs{#1}^{\mbox{-}1}}

\newcommand{\Int}[3]{\int_{#1} #2 \,  #3}

\newcommand{\Sla}[1]{{\bf S}(#1)}
\newcommand{\Ti}[1]{{\bf T}(#1)}

\newcommand{\crc}[1]{#1^{\circ}}

\newcommand{\ccrc}[1]{#1^{\circ \circ}}

\newcommand{\bang}[1]{! #1}
\newcommand{\nat}[1]{\natural #1}
\newcommand{\mn}{\mathfrak{m}}

\allowdisplaybreaks

\pdfoutput=1

\begin{document}

\begin{frontmatter}




\title{A Linear Exponential Comonad in s-finite Transition Kernels
and Probabilistic Coherent Spaces}


\author{Masahiro HAMANO}
\ead{hamano@gs.ncku.edu.tw}
\address{Miin Wu School of Computing, National Cheng Kung University, \\
No1, University Road, Tainan City, 70101, TAIWAN}

\begin{abstract}
This paper concerns a stochastic construction of probabilistic coherent spaces by
employing novel ingredients (i) linear exponential comonad arising
 properly in
the measure-theory (ii) continuous orthogonality between measures and measurable
functions. 

A linear exponential comonad is constructed over 
a symmetric monoidal category of transition kernels,
relaxing Markov kernels of Panangaden's stochastic relations
into s-finite kernels.
The model supports an orthogonality in terms of an integral
between measures and measurable functions, which can be seen as
a continuous extension of
%not only of Girard's linear duality of the coherent spaces, 
Girard-Danos-Ehrhard's linear duality for probabilistic coherent  spaces. 
The orthogonality is formulated by a Hyland-Schalk double
glueing construction, into which our measure theoretic
monoidal comonad structure is accommodated.
As an application to countable measurable spaces, 
a dagger compact closed category is obtained, whose double
glueing gives rise to the familiar category of
probabilistic coherent spaces.
\end{abstract}



\begin{keyword}
Stochastic Relations \sep 
Transition Kernels \sep 
Linear Exponential Comonad \sep 
Linear Logic \sep 
Orthogonality \sep 
Measure Theory \sep 
s-finite \sep
Exponential Measurable Space \sep
Double Glueing \sep 
Tight Orthogonality Category \sep 
Categorical Model \sep 
Probabilistic Denotational Semantics



\end{keyword}

\end{frontmatter}



\section*{Introduction}
Coherent spaces \cite{GirLL}, the original model in which Girard discovered
linear logic, provide a denotational semantics of
functional programming languages as well as logical systems.
Each space is a set endowed with a graph structure, called a {\em web}, in
which a proof (hence a program)
is interpreted by a certain subset, called a {\em clique}.
The  distinctive feature of this model is the linear duality,
stating that a clique $x \subseteq X$
and an anti-clique $x' \subseteq X'$
intersect in at most a singleton $\# (x
\cap x') \leq 1$.
The linear duality arising intrinsically to the coherent spaces
goes along with constructive modelling of logical
connectives. There arises a dual pair of multiplicative connectives
and of additive ones,  together with linear implication for 
multiplicative closed structure
(i.e., *-autonomy of denotational semantics). 

Category theoretically (freely from the web-based method),
the coherent spaces are
realised by Hyland-Schalk's double glueing construction \cite{HSha}
$\GRel$ over the category of relations $\Rel$, which
is the most primary self dual denotational  semantics
with the tensor (the cartesian product of sets)
and the biproduct (the disjoint union of sets).
The double glueing lifts the degenerate duality
of $\Rel$ into a nondegenarate one, called {\em orthogonality},
which in turn gives rise to the linear duality so that 
the coherent spaces
reside as an orthogonal subcategory.

Developing the web method,
Ehrhard investigates the linear duality in the mathematically richer structures
of K\"{o}the spaces \cite{EhrKothe} and finiteness spaces
\cite{EhrFinite}. His investigation of duality leads 
Danos-Ehrhard \cite{DE} to formulate a probabilistic (fuzzy) version of duality
in their probabilistic coherent spaces $\Pcoh$.
Their construction starts with
giving non-negative real valued functions on a
web $I$, reminiscent of probabilistic distributions (but
not necessarily to the interval $\interval{0}{1}$) on the web,
so that a clique becomes a subset of $\mathbb{R}_+^I$.
%In the probabilistic coherent spaces, 
Then, in their probabilistic setting, the linear duality becomes formulated
$\inprocoh{x}{x'}=\sum_{i \in I} x_i x'_i \leq 1 $ for $x, x' \in
\mathbb{R}_+^I$.
The precursor of the formulation is addressed earlier
in Girard \cite{Girquan}.
The probabilistic linear duality
is accommodated into the linear exponential $!$ by
generalising the original finite multiset functor construction
of Girard \cite{GirLL} with careful analysis of
permutations and combinations on enumerating members of multisets.
The canonicity of their exponential construction
is ensured in \cite{Crubille}.% in terms of free exponential.
% to be the free commutative comonad.

The recent trend of probabilistic semantics is more widely applied
to transition systems with continuous state spaces for
concurrent systems such as stochastic process calculi.
The stochastic relation $\SRel$, explored by Panangaden
\cite{Prabook,Prapaper},
provides a fundamental categorical ingredient to the study,
analogous to how the 
category $\Rel$ of the relations has been to deterministic
discrete systems.
Recalling that $\Rel$ is the Kleisli category of the powerset monad,
$\SRel$ is a probabilistic analog of the Giry monad \cite{Giry},
whereby powerset is replaced by a
probability measure on a set, giving
random choice of points, hence collections of fuzzy subsets are obtained.
$\SRel$ also provides coalgebraic reasoning for continuous time
branching logics \cite{Doberkat}.
Despite the lack of cartesian closed structure,
Markov kernels provide a measure theoretic foundation of recent
development of various denotational semantics for higher-order
probabilistic computations \cite{EPTpopl, HKSY},
theoretically with adequacy and practically with continuous
distributions for Monte Carlo simulation.
We also remark
 %that an intermediate approach but confining the discrete probability
 %%on relational matrices 
 %was the weighted relational model \cite{LMMP}
 %acquiring *-autonomy and exponential structure for Linear Logic.
 an intermediate approach on
 the weighted relational model \cite{LMMP}
 confining the discrete probability
 but acquiring *-autonomy and exponential structure for Linear Logic.

This paper intends to present a general machinery inspired by $\SRel$, amalgamating category and measure theories. 
This integration leads to the development of two fundamental constructions:
(i) linear exponential comonad tailored for stochastic processes
and (ii) linear duality and probabilistic orthogonality
in continuous spaces. The two parts involve comonad,
widely used in computer science,
but exploration in its continuous stochastic aspect
is initiated just recently by \cite{EhrCone, Geoff, Paquet}
in the higher order probabilistic programming.

Our results for each can be summarized as follows:
(i) The counting process \cite{CP} in the realm of stochastic processes introduces a novel categorical representation of linear exponential comonad, 
capturing the exponential modality of linear logic. Specifically,
for countable measurable spaces, this approach simplifies the understanding of the exponential structure within $\Pcoh$ by representing it as a discrete collapse of measure-based probability. 
Furthermore, our linear exponential comonad based on transition kernels can be viewed as a continuous version of the weighted relational model outlined in \cite{LMMP} utilising $\zeroinf$-weighted $\Rel$ for the analytic exponential.
(ii) Within a broader continuous framework, our study of transition kernels offers a new perspective on Hyland-Schalk orthogonality. This perspective provides insight into the continuous extension
of linear duality in terms of measures and measurable functions.

It is important to note that both (i) and (ii) do not incorporate any closed structure for monoidal products in the continuous framework, making them inconclusive as a complete model of linear logic.

\smallskip

The paper commences by introducing a stochastic framework involving transition kernels \cite{Bau} that establishes a category $\TKer$ with biproduct.
This framework elucidates the necessity for kernels to encompass
the infinite real values $\infty$,
especially in the context of transition kernels between measurable spaces and measurable functions, which serve as a relaxation of sub-Markov kernels within Panangaden's $\SRel$ of sub-probability measures. 

The monoidal product is straightforwardly derivable through 
the measure theoretic direct product, similar to $\SRel$.
However within the relaxation in our framework,
careful examination of the functoriality of the product 
becomes imperative.
In measure theory, the functorial monoidal product is guaranteed by the fundamental Fubini-Tonelli Theorem, where 
$\sigma$-finiteness \cite{Prabook}, including finiteness and subMarkov
properties, plays a crucial role.
However, category theory presents challenges
as $\sigma$-finiteness is not preserved 
under categorical composition. Although a smaller
class of finiteness \cite{Borg}
maintains both categorical composition and functorial monoidalness,
it proves insufficient for accommodating the exponential modality as
required in this paper. 
The concept of s-finiteness, recently explored in Staton's work \cite{Stat},
extends the traditional $\sigma$-finiteness to preserve composition while
upholding Fubini-Tonelli for functorial monoidal product. 
We demonstrate that s-finiteness also facilitates an exponential construction, both in terms of measure theory and category theory.
The exponential construction is characterised
through ``counting measures'' \cite{CP} 
in the realm of stochastic processes, where the counting 
function for multisets becomes measurable.
In our study, we establish an exponential endofunctor within the monoidal 
category $\TKersfin$, consisting of s-finite transition kernels. 
Furthermore, we devise a linear exponential comonad in
$\opTKersfin$, serving as a model for the exponential modality in linear logic
from a category-theoretical perspective \cite{BS,HSha, Mellies}.

Secondly, we utilise Hyland-Schalk general categorical construction of double
glueing \cite{HSha} to $\opTKersfin$. In the double glueing $\GopTKersfin$,
the coproduct and product operations, as well as the monoidal tensor and cotensor, exhibit distinctions. Our primary focus is how to lift the linear exponential 
comonad in $\opTKersfin$ to the double glueing. 
By applying the general methodology \cite{HSha}
to our specific exponential kernels, distinct linear comonad structures are established within $\GopTKersfin$.  Subsequently,
by observing a contravariant equivalence between $\TKer$ and $\Mes$ 
(representing measurable functions and linear positive maps
preserving monotone convergence), we formulate an orthogonality between
a measurable map $f$ and a measure $\mu$ so that $f \in \TKer(\mc{X},\mc{I})$
and $\mu \in \TKer(\mc{I},\mc{X})$ are orthogonal if  
$\int f d \mu \leq 1$. 
This orthogonality serves as a continuous version of Danos-Ehrhard's linear
duality for $\Pcoh$.
Our orthogonality includes an adjunction between operators
$\kappa^*$ and $\kappa_*$
associated with a kernel $\kappa$ respectively on measurable functions and on
measures.
This adjunction enforces a coherence condition for the orthogonality concerning the exponential comonad. It notably simplifies
the general construction by Hyland-Schalk.
The introduced orthogonality concept enables the construction of
certain double glueing subcategories, including tight orthogonality one,
on which our primary focus lies.
%Our primary focus lies on the tight category moving forward.

Finally, we delve into the full subcategory $\TKersfinomg$ of countable measurable spaces, where morphisms of transition kernels collapse into transition matrices. Within $\TKersfinomg$, there exists a dagger functor internalising the contravariant equivalence in the subcategory restriction. This results in a monoidal closed structure within $\TKersfinomg$, rendering the category dagger compact closed. Consequently the double gluing $\GopTKersfinomg$
becomes *-autonomous. Our goal is to establish an equivalence 
of the tight orthogonal subcategory $\Ti{\opTKersfinomg}$ to the category
$\Pcoh$ of probabilistic coherent spaces. Notably, this equivalence marks the first precise formulation
of the folklore among the linear logic community (cf. \cite{Pagslide}).
%In turn this particularly subsumes
%the coincidence of our $\TKersfinomg$ with
%$\zeroinf$-weighted $\Rel$ (\cite{LMMP}) known previously
%within the community. 

\smallskip
 
The paper is organised as follows:
Section \ref{firstsec} presents various categories of
transition kernels and  measurable spaces.
Section \ref{secems} starts with a measure theoretic study on exponential
measurable spaces and establishes exponential transition kernels.
Section \ref{secems} constructs a linear exponential comonad over
a monoidal category
$\opTKersfin$.
Section \ref{dgort} is an application of
Hyland-Schalk double glueing 
to our measure theoretic construction.
% and presents a new instance of the
%orthogonality.
Section \ref{closedTK} restricts to the countable measurable spaces
in particular for obtaining $\Pcoh$ as a double glueing.

\tableofcontents

\bigskip

\section{Preliminaries from Measure Theory}
This section recalls some basic definitions and a theorem
from measure theory, necessary in this paper.

\noindent ({\bf Terminology})
$\mathbb{N}$ denotes the set of non negative integers.
$\mathbb{R}_+$ denotes the set of non negative reals.
$\zeroinf$ denotes $\mathbb{R}_+ \cup \{ \infty \}$.
% i.e., the interval $\interval{0}{\infty}$ in $\mathbb{R}$.
$\mathfrak{S}_n$ denotes the symmetric group
over $\{1, \ldots ,n\}$.
$\delta_{x,y}$ is the Kronecker delta.
For a subset $A$, $\chi_A$ denotes the characteristic function of $A$.
The Dirac delta $\delta_x (A)$ is  
$\chi_A(x)$. $\uplus$ denotes the disjoint union of sets.

\smallskip

\begin{defn}[$\sigma$-field $\mc{X}$
and measurable space $\ms{X}$]~\\{\em
A {\em $\sigma$-field} over a set $X$
is a family $\mc{X}$ of subsets of $X$
containing $\emptyset$,
closed under the complement and countable union.
A pair $\ms{X}$ is called a {\em measurable space}.
The members of $\mc{X}$ are called {\em measurable sets}.
The measurable space is often written simply 
by $\mc{X}$, as $X$ is the largest element in $\mc{X}$.
For a measurable set $Y \in  \mc{X}$,
the measurable subspace $\mc{X} \cap Y$,
called the {\em  restriction on} $Y$,
is defined by $\mc{X} \cap Y := \{  A \cap Y \mid A \in \mc{X} \}$.
}
\end{defn}

%\smallskip

\begin{defn}[$\sigma(\mc{F})$ and Borel $\sigma$-field $\borel$]
{\em For a family $\mc{F}$ of subsets of $X$,
$\sigma(\mc{F})$ denotes the $\sigma$-field generated by $\mc{F}$,
i.e., the smallest $\sigma$-field containing $\mc{F}$.
When $X$ is $\zeroinf$ and $\mc{F}$ is the family $\mc{O}_{\zeroinf}$
of the open sets in $\zeroinf$ (with the topology whose basis consists of 
the open intervals in $\mathbb{R}_+$
together with $(a, \infty):=\{ x \mid a  < x  \}$
for all $a \in \mathbb{R}_+$), the $\sigma$-field
is denoted by $\borel$, whose members
are called Borel sets over $\zeroinf$.
}\end{defn}
%\smallskip
\begin{defn}[measurable function]{\em
For measurable spaces$\ms{X}$ and $\ms{Y}$,
a function $f: X \longrightarrow Y $ is {\em
 $(\mc{X},\mc{Y})$-measurable}
(often just {\em measurable})
if $f^{-1}(B) \in \mc{X}$ whenever $B \in \mc{Y}$.
In this paper,
a measurable function unless otherwise mentioned is  to the Borel set
$\borel$ over
$\zeroinf$ from some measurable
space $\ms{X}$. 
}\end{defn}
%\smallskip
\begin{defn}[measure]{\em
A {\em measure} $\mu$ on a measurable space $\ms{X}$
is a function from $\mc{X}$ to $\zeroinf$
satisfying ({\em $\sigma$-additivity}):
If $\{ A_i \in \mc{X} \mid i \in I \}$
is a countable family of pairwise disjoint sets,
then $\mu(\bigcup_{i \in I} A_i)= \sum_{i \in I} \mu(A_i)$.
 }\end{defn}
%\smallskip

\begin{defn}[integration]{\em
For a measure $\mu$ on $\ms{X}$, and a
$(\mc{X}, \borel)$ -measurable function $f$, 
the integral of $f$  over $X$ wrt the measure $\mu$ is defined
by $\int_X f(x) \mu(dx)$, which is simply written $\int_X f d \mu$.
It is also written $\int_X d \mu f$.
}\end{defn}

%\smallskip

\begin{thm}[monotone convergence] \label{MC}
Let $\mu$ be a measure on a measurable space $\ms{X}$.
For an monotonic sequence $\{ f_n \}$
of $(\mc{X}, \borel)$-measurable functions,
if $f = \sup_n {f_n}$, then $f$ is measurable and $\sup \int_X f_n d \mu =
\int_X f d \mu$.
\end{thm}

%\smallskip

\begin{defn}[push forward measure $\mu \Comp F^{-1}$
along a measurable function $F$]\label{pfm}{\em
For a measure $\mu$ on $\ms{Y}$ and a measurable function $F$ from
$\ms{Y}$ to $\ms{Y'}$, $\mu'(B'):= \mu (F^{-1}(B))$ becomes a measure on
$\ms{Y'}$, called {\em push forward measure} of $\mu$
along $F$. The push forward measure $\mu'$ has the following property, called
{\em ``variable change of integral along push forward $F$''}:
\begin{align*} %\label{defIWPFM}
& \textstyle \int_{Y'} g \, d \mu'= \int_{Y} (g \Comp F) \, d \mu. 
& \text{That is,} \quad \quad
\textstyle \int_{Y'} g (y') \,\mu' (d y')= \int_{Y} g (F(y)) \, \mu (dy)
\end{align*}
The push forward measure $\mu'$ is often denoted by $\mu \Comp F^{-1}$
by abuse of notation.
}
\end{defn}

\section{Category $\TKer$, its Dual $\Mes$ and Monoidal Subcategory
 $\TKersfin$
of s-finite Transition Kernels}
\label{firstsec}
This section starts with introducing a category $\TKer$
of transition kernels with convolution
(i.e., an integral transform on the product)
as categorical composition. 
Measures and measurable function on a measurable space both arise
as certain morphisms in the category.
A contravariant equivalence is shown to
a category $\Mes$ of measurable functions.
When imposing s-finiteness to kernels,
a monoidal subcategory $\TKersfin$ with (countable) biproducts 
is obtained.
\subsection{Transition kernels and Contravariant Equivalence}
\begin{defn}[transition kernel]\label{TKer}{\em 
For measurable spaces $(X, \mathcal{X})$ and $(Y, \mathcal{Y})$,
a {\em transition kernel} from
$(X, \mathcal{X})$ to $(Y, \mathcal{Y})$ is a function
$$\begin{aligned}
\kappa  :  X \times \mathcal{Y}
\longrightarrow \zeroinf
\end{aligned}$$ 
such that 
\begin{itemize}
 \item[(i)]
For each $x \in X$, 
the function $\kappa(x, -):
\mathcal{Y}
\longrightarrow \zeroinf$ is a measure on $\ms{Y}$.
\item[(ii)]
For each $B \in \mathcal{Y}$,
 the function $\kappa(-,B): X
\longrightarrow \zeroinf$ is measurable on $\ms{X}$.
\end{itemize}
}\end{defn}

\begin{defn}[Operations $\kappa_*$ and $\kappa^*$ of a kernel $\kappa$
on measures and measurable functions] \label{opekernel} ~\\
{\em Let $\kappa: \ms{X} \longrightarrow \ms{Y}$
be a transition kernel.
\begin{itemize}
 \item 
For a measure $\mu$ on $\mathcal{X}$, 
$$\begin{aligned}
(\kappa_* \mu)(B) := \int_{X} \kappa (x,B) \mu( dx)
\end{aligned}$$
is a measure on $\mathcal{Y}$, where $B \in \mathcal{Y}$.

In particular, for a Dirac measure $\delta_a$ with any $a \in X$,
\begin{align*}
(\kappa_* \delta_a)(B)=\int_X \kappa(x, B) \delta_a (dx)
= \kappa(a, B)
\end{align*}

\item
For a measurable function $f$ on $\mathcal{Y}$, 
$$\begin{aligned}
(\kappa^* f)(x) := \int_{Y} f(y) \kappa (x, dy)
\end{aligned}$$
is measurable on $\mathcal{X}$, where $x \in X$. \\
In particular, for a characteristic function $\chi_B$
for any $B \in \mathcal{Y}$,
\begin{align}
(\kappa^* \chi_B)(x) := \kappa (x, B) \label{kapchi}
\end{align}
\end{itemize}
}\end{defn}
\noindent It is direct to check, by the monotone
convergence theorem \ref{MC}, that $\kappa^* f$ is measurable.
%taken simple functions $s_n$s converging pointwise to $f$.
%The functoriality will be later seen when the convolution of kernels
%is defined.

A characterization is known in (\ref{kapchi})
for which general mappings $\alpha$ in place of $\kappa^*$
in turn define transition kernels as follows:
\begin{prop}[Lemma 36.2 \cite{Bau}] \label{lemchar}
Let $\mf{X}$ denote the set of all $\zeroinf$-valued
measurable functions on a
 measurable space $\mc{X}$.
If a function $\alpha:\mf{Y} \longrightarrow \mf{X}$
is linear (that is, $\alpha(0)=0$ and 
$\alpha(r f + sg) = r \alpha( f) + s \alpha(g)$ for $r,  s \in \mathbb{R}_+$),
positive (that is, $\alpha f \geq 0$ whenever $f \geq 0$)
and preserves monotone convergence (that is, $\sup \alpha(f_n) =
\alpha \sup (f_n)$
for any monotone sequence $\{ f_n \}$ in $\mf{Y}$),
then 
\begin{align*}
 \alpha(\chi_B)(x):= \kappa (x, B)
\end{align*}
becomes a transition kernel from $\ms{X}$ to $\ms{Y}$.
Moreover $\kappa$ is the unique transition kernel satisfying $\kappa^* f
 =\alpha(f)$ for all $f \in \mf{X}$.
\end{prop}

\bigskip

\begin{defn}[categories $\TKer$ and $\Mes$] ~
{\em 
\begin{enumerate}
 \item[-]
$\TKer$ denotes the category where each object is a measurable space $\ms{X}$
and a morphism is a transition kernel $\kappa$
from $\ms{X}$ to $\ms{Y}$.
The composition is the 
{\em convolution} of two kernels
$\kappa : \ms{X} \longrightarrow
\ms{Y}$ and $\iota : \ms{Y} \longrightarrow
\ms{Z}$:
\begin{align}
\iota \Comp \kappa (x,C) = &
\int_Y \kappa(x,dy) \iota(y,C) \label{compTKer}
\end{align}
$\operatorname{id}_{\ms{X}}$ is the {\em unit kernel} $\delta:
\ms{X} \longrightarrow \ms{X}$, defined 
for $x \in X$ and $A \in \mc{X}$ by;
$$\begin{aligned}
\text{if $x \in A$ then $\delta(x,A) = 1$, else $\delta(x,A) = 0$.}
\end{aligned}$$  
That is, for each $x \in X$, $\delta(x, -)$ is the Dirac measure on $\ms{X}$.

\item[-]
$\Mes$ denotes the category whose objects are measurable spaces, same as $\TKer$, but whose morphisms
$\Mes(\mc{X}, \mc{Y})$ consists of any linear positive map $\alpha:
\mf{X} \longrightarrow \mf{Y}$
preserving monotone convergence.
The composition is simply that of the functions.
\end{enumerate}
}\end{defn}
\noindent 
It is now well known that the composition (\ref{compTKer})
for sub-Markov kernels (cf. below Remark \ref{Srel})
comes from Giry's probabilistic monad,
resembling the power set monad of the relational composition
(cf. \cite{Giry, Prabook}).
Instead of using the monad applied to our general setting
for the transition kernels, 
we give a simpler intuition 
%Rather than rigidity, an intuition is given 
how the composition (\ref{compTKer}) arises 
via the simpler composition of
$\Mes$  by assuming the expected functoriality
$\begin{aligned}
 (\iota \Comp \kappa)^* f  = 
 (\iota^* \Comp \kappa^*)(f)
\end{aligned}$.
That is, for any $x \in X$, we see 
\begin{align*}
(\iota \Comp \kappa)^* f (x)  =
  ((\kappa^* \Comp \iota^*)(f))(x)
=  
(\kappa^* (\iota^* f))(x) = 
\int_Y \iota^* f (y) \kappa(x,dy) 
=
\int_Y (\int_Z f(z) \iota(y,dz)) \kappa(x,dy)
\end{align*}
In particular, taking $f=\chi_{C}$ yields 
\begin{align*}
(\iota \Comp \kappa)^* \chi_C (x)  =
\int_Y \iota(y,C) \kappa(x,dy),
\end{align*}
which by (\ref{kapchi}) imposes the definition of 
the composition of the two kernels.
%$(\iota \Comp \kappa)(x,C)$ 
%in $\TKer$.

\smallskip

\begin{rem}[measures and measurable functions as $\TKer$ morphisms.]\label{remmesfun}{\em
Measures and measurable functions
both reside as morphisms in $\TKer$: Let $\ms{I}$ be the singleton
 measurable space with $I=\{ * \}$, hence $\mc{I}=\{ \emptyset, \{ *\}
 \}$, then 
$$\begin{aligned}
\TKer(\mc{I}, \mc{X}) &  = \{ \kappa(*,-) : 
\mc{X} \longrightarrow \zeroinf \mid 
 \mbox{kernel $\kappa$ with domain $\mc{I}$}  \} 
=  \{ \mbox{ the measures $\mu$ on $\ms{X}$ }\} \\
\TKer(\mc{X}, \mc{I}) & =  \{ 
\kappa(-,\{ * \}) : X \longrightarrow \zeroinf
\mid \mbox{kernel $\kappa$ with codomain $\mc{I}$} \}  
   \cup  \{ \kappa(-, \emptyset) =0 \}  \\ &  =  
\{ \mbox{\rm the measurable functions $f$ on $\ms{X}$ to $\borel \!$ }\} 
 \end{aligned}$$
The operations $\kappa_*$ and $\kappa^*$ of Definition 
\ref{opekernel} are respectively categorical precomposition and composition
with $\kappa$ in $\TKer$ so that $\begin{aligned}
\kappa_* \mu = \kappa \Comp \mu
\quad \text{and} \quad \kappa^* f = f \Comp \kappa.
\end{aligned}$
}\end{rem}

%\smallskip

\begin{rem}[$\SRel$ \cite{Prabook,Prapaper}] \label{Srel}
{\em The category $\SRel$ of {\em stochastic relations}
is a wide subcategory of $\TKer$ strengthening
the conditions of Definition \ref{TKer}
into (i) $\kappa(x, -)$ is a {\em sub-probability} measure
(i.e., a measure from $\mc{Y}$ to $\left[0, 1\right]$)
 and (ii) $\kappa(-, B)$ is a {\em bounded} measurable function. \\
The morphisms of $\SRel$ are called {\em sub-Markov kernels}.
They are called Markov kernels 
when $\kappa(x, B)=1$ for any $x \in X$.
}\end{rem}
\noindent Note: The bounded condition of (ii) is derivable from (i),
thus the condition (ii) is redundant when defining $\SRel$
as a subcategory of $\TKer$.

\smallskip
We also remark here
a crucial reason seen immediately in the next Subsection \ref{cbTK}
why $\SRel$ needs to be extended to $\TKer$ in this paper: The coprpduct of $\SRel$
given in \cite{Prabook,Prapaper}
is not a biproduct in $\SRel$, but it is so in $\TKer$ 
(cf. Proposition \ref{TKerbipro}).
The biproduct existing in $\TKer$ will be crucial
to the main purpose of the paper in order to construct
an exponential structure in the s-finite subcategory introdued in
Section \ref{mpacb} below.

\smallskip
Proposition \ref{lemchar} says category theoretically;
\begin{prop} \label{contraeq}
$\TKer$ and $\Mes$ are contravariantly equivalent.
The equivalence is given by the contravariant functor 
$$\begin{aligned}
(~)^*: \TKer
 \cong (\Mes)^{\opsymbol}
\end{aligned}$$
 On the objects, $(~)^*$ acts as the identity. 
On the morphisms, the functoriality  $(\iota \Comp \kappa)^* = \kappa^* \Comp \iota^*$ is checked above. 
\end{prop}
The contravariant equivalence $(~)^*$ in particular gives a direct
account 
on the measurable functions as the homset in Remark \ref{remmesfun}
by $\TKer(\mc{X},\mc{I}) \cong \Mes(\mc{I}, \mc{X})$.

\smallskip

\begin{rem}{\em
The contravariant functor $(~)^*$ when restricted to the Markov kernels $\SRel$
gives a contravariant equivalence to 
$\Ban$, where each object
is the subspace $\mc{E}^b(\mc{X})$ of {\em bounded} measurable
 functions, which forms a vector space.
The boundedness makes the space not only a vector space but moreover
a Banach space with the uniform norm 
$\left\lVert f \right\rVert  =\sup_{x \in X} \abs{\! f(x) \!}$.
The opposite category is studied in \cite{Prabook}
as the category of the {\em predicate transformers},
stemming from Kozen's precursory work %\cite{Koz}
on probabilistic programming.
Taking measurable functions as predicates and measures as states,
the ordinary satisfaction relation, say $\mu \models f$,
is generalised into integrals, say
$\int f d \mu$ giving a value in 
the interval $\interval{0}{1}$.
In the present paper in Section \ref{ort}, this satisfaction relation
will be explored in terms of the orthogonality relation.
}\end{rem}

\subsection{Countable Biproducts in $\TKer$} \label{cbTK}
The transition kernels have an intrinsic category theoretical property.
\begin{prop}[biproduct $\coprod$] \label{TKerbipro}
$\TKer$ has countable biproducts.
\end{prop}
\begin{proof}
Given a countable family $ \{  \msi{X}{i} \}_i$ of measurable spaces,
we define 
\begin{align}
&  \textstyle \coprod \limits_{i} \msi{X}{i} :=(\bigcup \limits_i \{ i
 \} \! \times \! X_i, \,  \biguplus \limits_i
\mathcal{X}_i),  \label{biprod} 
\end{align}
where $ \biguplus_i \mathcal{X}_i := \{ \bigcup_i \{ i \} \! \times \! A_i \mid A_i \in
\mathcal{X}_i  \}$
is the $\sigma$-field generated by the measurable sets of each
 summands.

\noindent (Coproduct):
(\ref{biprod}) defines a coproduct for $\TKer$.
The injection
$\inj_j : \msi{X}{j} \longrightarrow (\bigcup \, \{ i \} \! \times \! X_i, \biguplus
\mathcal{X}_i)$ is defined by
$\inj_j (x_j, \sum_i A_i):= \delta(x_j,A_i) $.
The mediating morphism $\oplus_{i \in I} f_i
: (\bigcup \, \{i \} \times X_i, \biguplus
\mathcal{X}_i) \longrightarrow \ms{Y}$
for given morphisms
$f_i: \msi{X}{i} \longrightarrow \ms{Y}$
is defined by 
$(\oplus_{j \in I} \, f_j) ((i,x),B) = f_i (x,B)$.
Note (\ref{biprod}) is the same instance as the known 
coproduct in $\SRel$.
However, in the relaxed structure of $\TKer$, we have moreover;

\noindent (Product): (\ref{biprod}) becomes a product for $\TKer$.
The projection $\pr_i:  \coprod_{i} \msi{X}{i} \longrightarrow
\msi{X}{i}$ is given by
$\pr_i ((j,x),A_i):= \delta_{i,j} \cdot \delta(x,A_i)$.
The mediating morphism $\&_i g_i$ for given morphisms
$g_i : \ms{Y} \longrightarrow \msi{X}{i}$ is defined to be
$(\&_i g_i)(b, \,  \bigcup \, \{ i \} \! \times \! A_i
):=\sum_i  g_i (b,A_i)$. Note the construction for the meditating
morphism is not closed in Markov kernels, but is so 
in transition kernels.
This construction shows how values of measurable functions include the
infinite real $\infty$ when $I$ becomes infinite.
We check the uniqueness of the mediating morphism,
 say $m$: 
\begin{align}
& \textstyle (\pr_j \Comp m) (b,  \{ j \} \! \times \! A_j)
= \int_{\uplus_i X_i}  m(b,d(i,x)) \pr_j ((i,x), \{ j \} \! \times \!
 A_j) = \textstyle \int_{\uplus_i X_i}  m(b,d(i,x)) \, \delta_{j,i} \cdot
 \delta((i, x), \{ j \} \! \times \! A_j) \nonumber \\ 
&\textstyle  =
\int_{X_j}  m(b,d(j, x)) \delta((j, x),
\{ j \} \! \times \! A_j)
=
m(b, \{ j \} \! \times \! A_j)  \label{unimed}
\end{align}
The required commutativity for $m$
is $\pr_j \Comp m = g_j$, which holds  by (\ref{unimed}) if and only if 
$m(b,A_j)=g_j (b,A_j)$ for all $j$.
Since $m(b,-)$ is a measure and $\{ A_j \}_{j \in J}$
are disjoint, this yields the definition $\&_i g_i$
of the mediating morphism.
\end{proof}
\noindent The unit of the biproduct is
the null measurable space 
$\mc{T}= (\emptyset, \{ \emptyset \})$.

\smallskip

This subsection ends with the following remark,
which though is not required to comprehend the paper.
\begin{rem}[$\TKer$ is traced wrt the biproduct]\label{UDC}{\em 
$\TKer$ is a {\em unique decomposition category} \cite{Hag, HagScot},
which is a generalisation of Arbib-Manes
{\em partially additive category} studied in \cite{Prabook}
for $\SRel$.
A countable family of $\TKer$ morphisms $\{ \kappa_i : \ms{X} \longrightarrow
 \ms{Y} \mid i \in I  \}$ is summable so that 
$\sum_{i \in I} \kappa_i (x, B)$ is a transition kernel.
Then $(\TKer, \coprod )$ is traced so that
any $\kappa: \ms{X} \coprod \ms{Z} \longrightarrow \ms{Y} \coprod
 \ms{Z}$ yields
${\sf Tr}_{\mc{X},\mc{Y}}^{\mc{Z}}(\kappa):
\ms{X} \longrightarrow \ms{Y}$ which is the standard trace formula
corresponding to Girard's execution formula for Geometry of Interaction
\cite{GirGoII} and is defined further in \cite{HamGoI} using the execution concept. The trace operator enables modelling of both
{\em feed back} and {\em iteration}
 on a given morphism. 
Notably the Int construction 
by Joyal-Street-Verity \cite{JSV}
results in a compact closed
completion of $\TKer$ with $\coprod$ serving as tensor.
It is important to note that the monoidal product discussed in this paper is distinct from this one, but the measure theoretic direct product
as introduced in Definitions \ref{prodobj} and \ref{moncat} below.}
\end{rem}

\subsection{Monoidal Product and Countable Biproducts in $\TKersfin$ }
\label{mpacb}
This subsection introduces a subcategory $\TKersfin$
of s-finite transition kernels. 
The s-finiteness is a relaxation of
a standard measure theoretic class of the $\sigma$-finiteness so that 
the $\sigma$-finiteness resides intermediately
between finiteness and s-finiteness.
The relaxed class of the s-finite kernels is closed under
composing kernels, which is not the case
in the class of $\sigma$-finite kernels.
Inside the subcategory $\TKersfin$, the monoidal product
of morphisms is functorially defined to accommodate Fubini-Tonelli Theorem for
the unique integration over product measures. 
$\TKersfin$ is also shown to retain the countable biproducts in $\TKer$
of the previous subsection.

% Our restriction of the s-finiteness
% of the kernels is weaker than the standard measure-theoretic
% assumption of $\sigma$-finiteness, because the latter finiteness
% is not closed in general under composing kernels.

\begin{defn}[product of measurable spaces] \label{prodobj}{\em
The product of measurable spaces $(X_1, \mc{X}_1)$ and $(X_2, \mc{X}_2)$
is the measurable space $(X_1 \times X_2, \mc{X}_1 \otimes \mc{X}_2)$,
where $\mc{X}_1 \otimes \mc{X}_2$ denotes the $\sigma$-field
over the cartesian product $X_1 \times X_2$
generated by {\em measurable rectangles} $A_1 \times A_2$'s 
such that $A_i \in \mc{X}_i$.}
 \end{defn}
In order to accommodate 
measures into the product of measurable spaces,
each measures $\mu_i$ on $(X_i, \mc{X}_i)$ need to be extended uniquely to
that on the product.
The condition of {\em $\sigma$-finiteness} ensures this,
yielding the unique product measure over the product measurable space:
\begin{defn}[$\sigma$-finiteness]{\em
A measure $\mu$ on $\ms{X}$
is {\em $\sigma$-finite} when the set $X$ is written as a countable union of sets of
finite measures. That is, $\exists A_1, A_2, \ldots \in \mc{X}$ such that
$\mu(A_i) < \infty$ and $X = \cup_{i=1}^{\infty} A_i$.
}\end{defn} 
%\smallskip
\begin{defn}[product measure]\label{prodmeasure}{\em
For a $\sigma$-finite measures $\mu_i$ on $(X_i, \mc{X}_i)$ with $i=1,2$,
there exists a unique
measure $\mu$ on $(X_1 \times X_2, \mc{X}_1 \otimes \mc{X}_2)$
such that $\mu (A_1 \times A_2)=\mu_1 (A_1) \mu_2 (A_2)$.
$\mu$ is written $\mu_1 \otimes \mu_2$ and called the {\em product measure}
of $\mu_1$ and $\mu_2$.}\end{defn}
\smallskip
The product measure derived from $\sigma$-finite measures
guarantees a basic theorem in measure theory, stating
double integration is treated as iterated integration.
\begin{thm}[Fubini-Tonelli] \label{FT}
For $\sigma$-finite measures $\mu_i$ on $(X_i, \mc{X}_i)$ with $i=1,2$
and a 
$(\mc{X}_1 \otimes \mc{X}_2, \borel)$-measurable function $f$,
$$\begin{aligned}
\int_{X_1 \times X_2} \! \! \! \! \! \! f \, d (\mu_1 \otimes \mu_2) =
\int_{X_2} \! \!   d \mu_2  \int_{X_1} \!  f \, d \mu_1 =
\int_{X_1} \! \! d \mu_1  \int_{X_2} \! f \, d \mu_2
%\int_{X_2} (\int_{X_1} \! f \, d \mu_1) \, d \mu_2 =
%\int_{X_1} (\int_{X_2} \! f \, d \mu_2) \, d \mu_1
\end{aligned}$$
\end{thm}
% Note in this paper
% the second and the thirdvequations are written respectively by 
% $\int_{X_2} d \mu_2  \int_{X_1} \! f \, d \mu_1$
% and $\int_{X_1} d \mu_1  \int_{X_2} \! f \, d \mu_2$
The measure theoretical basic Fubini-Tonelli Theorem will become crucial
also to the categorical study of the present paper, not only dealing
with functoriality of morphisms
on the product measurable spaces (cf. Proposition
\ref{funcex} below), but also
giving a new instance of the orthogonality
using the measure theory in Section \ref{dgort}.

\smallskip

\smallskip

Although one can impose $\sigma$-finiteness
for the transition kernels $\kappa(x,-)$ (uniformly or non-uniformly
in $x$), this class of kernel is not closed in general
under the composition in the category $\TKer$.
For the sake of category theory, 
one remedy for ensuring the compositionality is
to tighten the class into the {\em finite} kernels.
This class confined to the measures is used in
finite measure transformer semantics \cite{Borg} for
probabilistic programs.
However the class of the finite kernels is not closed under
our exponential construction (Definition \ref{kerex}) later seen in
Section \ref{subsecexker}.
Thus, we need  another 
remedy to loosen  
the condition contrarily, which is how {\em s-finiteness} arises below.
While its notion was earlier established in \cite{Sharpe, Getoo},
the s-finiteness is recently studied by Staton \cite{Stat}
in modelling programming semantics.
%the s-finiteness seems to have hardly been studied in the measure theory.
In addition to the compositionality in our categorical setting,
the relaxed class of the s-finite kernels is shown to
retain the Fubini-Tonelli Theorem (Proposition \ref{FTsfin})
working with the uniquely defined product measure.
%we impose a relaxed condition {\em s-finiteness} on kernels,
%yielding a subcategory $\TKersfin$.

\smallskip

\begin{defn}[s-finite kernels \cite{Sharpe, Stat}] \label{sfinker}
{\em  
Let $\kappa$ be a transition kernel from $\ms{X}$ to $\ms{Y}$.
\begin{itemize}
 \item[-]
$\kappa$ %$\kappa(x,B): \ms{X} \longrightarrow \ms{Y}$
is called {\em finite} when 
$%\begin{aligned}
\sup_{x \in X} \kappa(x, Y) < \infty
%\end{aligned}
$; i.e., the condition says that up to the scalar $0 < a < \infty$ factor
determined by the sup, $\kappa$ is Markovian.
\item[-]
$\kappa$ %from $\ms{X}$ to $\ms{Y}$
is called {\em s-finite} when $\kappa = \sum_{i \in \mathbb{N}} \kappa_i$
where each $\kappa_i$ is a finite kernel from $\ms{X}$ to $\ms{Y}$
and the sum is defined by 
$(\sum_{i \in \mathbb{N}} \kappa_i )(x, B)
:=\sum_{i \in \mathbb{N}} \kappa_i (x, B)$.
This is well-defined because 
any countable sum of kernels from $\ms{X}$ to $\ms{Y}$
becomes a kernel of the same type.
\end{itemize}
}\end{defn}
\noindent In the definition of s-finiteness, 
note that
$(\sum_{i \in \mathbb{N}} \kappa_i)^*
= 
\sum_{i \in \mathbb{N}} \kappa_i^*$ and 
$(\sum_{i \in \mathbb{N}} \kappa_i)_*
= 
\sum_{i \in \mathbb{N}} (\kappa_i)_*$
for the operations of Definition \ref{opekernel}:
That is, the preservation of the operation $(~)^*$ (resp. of $(~)_*$) 
means the commutativity of integral
over countable sum of measures (resp. of measurable functions).

\begin{rem} \label{clopf}
{\em Both classes of the finite kernels and of the s-finite kernels are
closed under the categorical composition of $\TKer$.
This is directly calculated for the finite kernels, to which the
s-finite ones are reduced by the note in the above paragraph. We refer 
to the proof of Lemma 3 of \cite{Stat} for the calculation.
In particular,
the class of s-finite kernels is closed under push forwards along
 measurable functions.
The both classes form wide subcategories of $\TKer$ introduced below 
Definition \ref{moncat}.
}\end{rem}

The definition subsumes that of {\em s-finite measures}
when $\ms{X}$ is in particular taken the singleton
measurable space $\ms{I}$.
Note that every $\sigma$-finite measure is s-finite, but not vice versa:
E.g., the infinite measure
$\infty \cdot \delta_a$ for the Dirac $\delta_a$ with $a \in X$
is not $\sigma$-finite, but s-finite.

\smallskip

A characterization of s-finite kernels is directly derived:
\begin{prop}[Proposition 7 of \cite{Stat}]
A kernel is s-finite if and only if
it is a push forward of a $\sigma$-finite kernel.
\end{prop}
\begin{proof}
We prove ``only if'' part as ``if part'' is direct because
of the inclusion of $\sigma$-finiteness into s-finiteness
and of the closedness of s-finiteness under push forward.  \\
Given a s-finite kernel $\kappa=\sum_{i \in \mathbb{N}} \kappa_i$
with finite kernels $\kappa_i$s
from $\ms{X}$ to $\ms{Y}$, 
a $\sigma$-finite kernel $\nu$ from $\ms{X}$ to $(\mathbb{N}, 2^{\mathbb{N}})
\times \ms{Y}$ is defined by
$\nu (x,V):= \sum_{i \in \mathbb{N}} \kappa_i (x, \{ y \mid (i,y) \in V
 \})$. Then $\kappa$ is the push forward of $\nu$
along the projection $\mathbb{N} \times Y \longrightarrow Y$.
\end{proof}

\smallskip

The original  Fubini-Tonelli (Theorem \ref{FT})
for the $\sigma$-finite measures 
extends to
the s-finite measures:
\begin{prop}[Fubini-Tonelli extending for s-finite measures (cf. 
Proposition 5 of Staton \cite{Stat})]  \label{FTsfin}
For the same $f$ as Theorem \ref{FT} but
$\mu_1$ and $\mu_2$ are s-finite measures, it holds;
$$\begin{aligned}
%\int_{X_1 \times X_2} \! \! \! \! \! \! f \, d (\mu_1 \otimes \mu_2) =
\int_{X_2} \! \!   d \mu_2  \int_{X_1} \!  f \, d \mu_1 =
\int_{X_1} \! \! d \mu_1  \int_{X_2} \! f \, d \mu_2
\end{aligned}$$
\end{prop}
\begin{proof}
Write $\mu_1 = \sum_{i \in \mathbb{N}} \mu_1^i$ and 
$\mu_2 = \sum_{j \in \mathbb{N}} \mu_2^j$
with finite kernels $\mu_1^i$s and $\mu_2^j$s, then the following
is from (LHS) to (RHS): 

$\begin{aligned}
& \textstyle
\int_{X_2} 
\sum_{j \in \mathbb{N}} \mu_2^j (d x_2)
\int_{X_1} f \sum_{i \in \mathbb{N}} \mu_1^i (d x_1)
=
\sum_{j \in \mathbb{N}} 
\int_{X_2} 
\mu_2^j (d x_2)
\sum_{i \in \mathbb{N}}
\int_{X_1} f  \mu_1^i (d x_1)   \\
&  \textstyle =
\sum_{j \in \mathbb{N}} 
\sum_{i \in \mathbb{N}}
\int_{X_2} 
\mu_2^j (d x_2)
\int_{X_1} f  \mu_1^i (d x_1)
=
\sum_{j \in \mathbb{N}} 
\sum_{i \in \mathbb{N}}
\int_{X_1} 
\mu_1^i (d x_1)
\int_{X_2} f  \mu_2^j (d x_2) \\
& \textstyle
=
\sum_{i \in \mathbb{N}}
\int_{X_1} 
\mu_1^i (d x_1)
\sum_{j \in \mathbb{N}} 
\int_{X_2} f  \mu_2^j (d x_2)
=
\int_{X_1} \sum_{i \in \mathbb{N}}
\mu_1^i (d x_1)
\int_{X_2} \sum_{j \in \mathbb{N}}  f  \mu_2^j (d x_2)
\end{aligned}$ \\
The first and the last
(resp. the second and the second last)
 equations are
by the commutativity of integral over countable sum of
measurable functions (resp. of measures)
(cf. the note in Definition \ref{sfinker}).
The middle equation is the original Fubini-Tonelli for the
 $\sigma$-finite measures, hence here in particular for
the finite ones.
\end{proof}

\smallskip

Finally it is derived that the s-finite transition kernels form a
monoidal category.
\begin{defn}[monoidal subcategories $\TKersfin$ of s-finite kernels
and $\TKerfin$ of finite ones] \label{moncat}{\em
% A kernel $\kappa(x,B): \ms{X} \longrightarrow \ms{Y}$
% is called {\em finite} when 
% $%\begin{aligned}
% \sup_{x \in X} \kappa(x, Y) < \infty
% %\end{aligned}
% $. 
$\TKersfin$ is a wide subcategory of $\TKer$,
whose morphisms are the {\em s-finite} transition kernels.
The s-finiteness of kernels is preserved under the composition of $\TKer$.
$\TKersfin$ has a symmetric monoidal product $\otimes$: On objects 
is by Definition \ref{prodmeasure}.
Given morphisms $\kappa_1: \ms{X_{\mbox{\scriptsize $1$}}} \longrightarrow 
\ms{Y_{\mbox{\scriptsize $1$}}}$
and
$\kappa_2: \ms{X_{\mbox{\scriptsize $2$}}} \longrightarrow 
\ms{Y_{\mbox{\scriptsize $2$}}}$,
their product is defined explicitly: 
\begin{align*}
(\kappa_1 \otimes \kappa_2) ((x_1,x_2), C) :=
\int_{Y_1}  \kappa_1 (x, dy_1) 
\int_{Y_2}  \kappa_2 (x, dy_2) 
\, \, \chi_{C} ((y_1,y_2))
\end{align*}
Alternatively, thanks to Fubini-Tonelli (Proposition \ref{FTsfin}),
the monoidal product is implicitly defined as 
the unique transition kernel
$\kappa_2 \otimes \kappa_2: (X_1 \times X_1, \mc{X}_1 \otimes \mc{X}_2)
\longrightarrow (Y_1 \times Y_2, \mc{Y}_1 \otimes \mc{Y}_2)$
satisfying the following for any rectangle $B_1 \times B_2$
with $B_i \in \mc{X}_i$ for $i=1,2$:
\begin{align*}
(\kappa_1 \otimes \kappa_2) ((x_1,x_2), B_1 \otimes B_2)=
\kappa_1 (x_1,B_1) \kappa_2 (x_2,B_2) \quad
\end{align*}
% By 
% $(\kappa_1 \otimes \operatorname{Id}_2) \Comp 
% (\operatorname{Id}_1 \otimes \kappa_2)
% =
% (\operatorname{Id}_1 \otimes \kappa_2) \Comp 
% (\kappa_1 \otimes \operatorname{Id}_2) 
% $, where $\operatorname{Id}_i$ is the Dirac delta $\delta$ on
%  $\ms{X_{\mbox{\scriptsize $i$}}}$.
}\end{defn}

\noindent The unit of the monoidal product is the singleton measurable
space $\ms{I}$.
%where 
%$I= \{ *  \}$ and $\mathcal{I}=\{ \emptyset, I \}$.

$\TKerfin$ is a monoidal wide subcategory of 
$\TKersfin$ whose morphisms are finite transition kernels.

\smallskip

Employing s-finiteness in order to achieve the monoidal product
is due to Staton \cite{Stat}. Our continued focus on s-finiteness
is its capacity to retain the countable biproducts of $\TKer$
defined Proposition \ref{TKerbipro}.

\begin{prop}[The subcategory $\TKersfin$ retains the countable
biproducts of $\TKer$] \label{retbip}~\\
$\TKersfin$ has countable biproducts
which are those in $\TKer$ residing inside the subcategory.
\end{prop}
\begin{proof}
The coproduct construction of Proposition \ref{TKerbipro}
all works under the additional constraint of the s-finiteness of kernels. 
For the product construction, the only construction
necessary to be checked is that 
of the mediating morphism $\&_i g_i$, employing
the sum over $i \in I$ for a countable infinite $I$:
If given $g_i$'s of the product construction in
Proposition \ref{TKerbipro} are s-finite, then
each is written 
$g_i = \sum_{j \in \mathbb{N}} g_{(i,j)}$,
where each $g_{(i,j)}$ is a finite kernel from $\ms{Y}$ to $\msi{X}{i}$.
%thus so is $g_i$. 
In what follows,
the index set $I$ is identified with $\mathbb{N}$.
A transition kernel $h_n$ is defined for each $n \in \mathbb{N}$: 
\begin{align*}
\textstyle h_n := \sum \limits_{i+j = n} \inj_{i} \Comp g_{(i,j)} :
\ms{Y} \longrightarrow \coprod \limits_{j \in I} \msi{X}{j} 
\end{align*}
where $\inj_i: \msi{X}{i} \longrightarrow \coprod_{j \in I} \msi{X}{j}$
is the coproduct injection.
Note $h_n$ is a finite kernel, as the sum specified by
the subscript $i+j=n$ is finite.
Then in terms of the finite kernels,
the mediating morphism $\&_i g_i$ constructed in Proposition 
\ref{TKerbipro} is represented  as follows to be s-finite:
\begin{align*}
\textstyle
\raisebox{-.8ex}{
$\underset{i \in \mathbb{N}}{\text{\Large $\&$}}$} \, 
g_i
=
\sum \limits_{n \in \mathbb{N}} h_n
\end{align*}
\end{proof}
\begin{rem}[the infinite biproduct as colimit in $\TKersfin$]
\label{infbpcol}
{\em  
The countable infinite biproduct in $\TKersfin$
is characterised by the colimit inside
the subcategory: Given the {\em direct system}
$<  \, \coprod_{n=0}^{k} (X_{
 n},\mc{X}_{n} ) , \,  \inj_{k}^\ell >_{k \in \mathbb{N},
\ell \geq k}$ in $\TKersfin$,
the colimit $\varinjlim_k
\textstyle \coprod_{n=0}^{k} (X_{
 n},\mc{X}_{n} )$ coincides with the infinite biproduct
$\coprod_{n=0}^{\infty} (X_{
 n},\mc{X}_{n} )$ in
$\TKersfin$. Hence, the colimit is closed in the
subcategory $\TKersfin$, but not necessarily in $\TKerfin$.
}\end{rem}

\section{A Linear Exponential  Comonad over $\opTKersfin$} \label{secems}
\subsection{Exponential Measurable Space $\mse{X}$} \label{subsecems}
This subsection concerns a measure theoretic study on exponential
measurable spaces.  \cite{CP} is a good reference for the subsection.
%All the results in this section, unless otherwise mentioned,
%hold for $\sigma$-finite measurable spaces.

\begin{defn}[exponential monoid $X_e$] \label{expmonoid}
{\em $X_e$ denotes the {\em free abelian monoid} (the free semi group with identity)
generated by a set $X$: The members of $X_e$ are the formal products
$x_1 x_2 \cdots x_n$ where $x_i \in X$ and $n \in \msbf{N}$
so that order of factor is irrelevant.
The monoid operation for members of $X_e$
is obviously the free product.
When $n=0$, under the convention $x_1 x_2 \cdots x_n=0$, 
this is the monoid identity (in spite of the multiplicative notation),
which is equated with the empty sequence.
The monoid operator is written by a product $(x,y) \mapsto x y$.
Each member $\msbf{x} \in X_e$ is identified with a finite multiset of
 elements in $X$ and vice versa. Hence $\msbf{x}$ is 
seen as an integer valued function
on $X$, which vanishes to zero outside the finite sets;
$$\begin{aligned}
\msbf{x}(t)= \text{multiplicity of $t \in X$ as a factor of $\msbf{x}$}.
\end{aligned}$$
That is, $\msbf{x}$ represents the unique multiset of elements $X$,
and vice versa.

For $A \subseteq X$, we define 
$$\begin{aligned}
\textstyle \msbf{x}(A) :=  \sum \limits_{t \in A} \msbf{x}(t)
\end{aligned}$$
Then $\msbf{x}(A)$ represents the number of elements in $A$. \\
The {\em counting function} $\cunt{A}$ on $X_e$ is defined for each $A \subseteq
 X$,
$$\begin{aligned}
  \cunt{A} (\msbf{x}) = \msbf{x}(A)
 \end{aligned}$$
Note if $\msbf{x}$ is $x$ (i.e., the singleton sequence of $x \in X$),
then $ \msbf{x}(A)= \cunt{A}(x)= \delta(x,A)$ for any subset $A$ of $X$.
}\end{defn}

\smallskip

% \begin{rem}[$X_e$ as the set of all finite multisets of $X$]{\em 
% $X_e$ is identified with the set of all finite multisets of
% element of $X$: Given $\msbf{x} \in X_e$,
% the counting function
% $\msbf{x} : X
% \longrightarrow \mathbb{N}$ represents the unique multiset,
% and vice versa.
% % Each member $\msbf{x}=x_1 \cdots x_m \in X_e$ defines the unique
% % multiset of elements from the set $\{ x_i \mid i=1, \ldots m \}$
% % and vice versa. The corresponding multiset of $\msbf{x}$
% % is represented by the function
% % $\msbf{x} : \{ x_i \mid i=1, \ldots m \}
% % \longrightarrow \mathbb{N}$, $x_i \mapsto \msbf{x}(x_i)$.
% % Hence 
% }\end{rem}

\smallskip
The members of $X_e$ can be seen as equivalence classes of ordered sequences
in $X_e^\bullet$ defined below
under rearrangement (permutations of factors):
\begin{defn}[non-abelian monoid $X^\bullet_e$] \label{nam}{\em
Using $\ncdot$ for ordered sequences,
$X_e^\bullet$ denotes the nonabelian monoid generated by $X$,
consisting of ordered sequences
$x_1 \ncdot x_2 \ncdot \cdots \ncdot x_n$ where $x_i \in X$ and $n
\geq 0$.
The monoid operation for members of $X_e^\bullet$
is obviously the operation $\ncdot$ joining sequences in order.
Then the abelian monoid $X_e$ is the image of
the monoid homomorphism $F$ forgetting the order of the factors:
\begin{align*}
F :  X_e^\bullet \longrightarrow X_e &  \qquad x_1 \ncdot \cdots
\ncdot  x_n \mapsto x_1 \cdots x_n. 
\end{align*}
Obviously $F^{-1} (F(A))$ is 
the smallest symmetric set containing $A
\subset X_e^\bullet$. \\
The set $X_e^\bullet$ is a disjoint union
$$\begin{aligned}
X_e^\bullet = \biguplus_{n \geq 0} X^{\bullet n}
\end{aligned}$$
where the set $X^{\bullet n}$ denotes $
%\overbrace{X \ncdot \cdots \ncdot X}^{n}:=
\{ x_1 \ncdot \cdots
\ncdot  x_n  \mid x_i \in X \}$, which is
isomorphic to the $n$-ary cartesian product $X^n$ of $X$.
}\end{defn}
\smallskip

\noindent ({\bf Notation})
For any family $\mathcal{F}$ of subsets of $X$,
$\mathscr{P}^{\bullet}(\mathcal{F})$ denotes 
the class of all finite {\em ordered}
sequences $A_1 \ncdot \cdots \ncdot A_n :=\{
a_1 \ncdot \cdots \ncdot  a_n  \mid a_i \in A_i \}$
with $A_i \in \mathcal{F}$ and $n \in \mathbbm{N}$.
%For any subset $A$ of $X$, $\mc{F} \cap A$ denotes 
%$\{ Y \cap A \in Y \mid \mc{F}  \}$.
Similar notation for $\mathscr{P}(\mathcal{F})$
denoting the class of all {\em symmetric} formal product
$A_1 \cdots A_n
:=\{a_1 \cdots a_n  \mid a_i \in A_i \}$
so that the order of factors is irrelevant.

\smallskip

\smallskip

\begin{defn}[measurable space $\msed{X}$ induced by $\ms{X}$]
\label{fieldXnce}{\em 
Every measurable space $\mathcal{X}$ on $X$ induces a corresponding
measurable
space on the set $X_e^\bullet$ defined by:
$$\begin{aligned}
(X_e^\bullet, \mathcal{X}_e^\bullet)
& % := \coprod_{n \geq 0} (X^{\bullet n},\mc{X}^{\bullet n} ) 
 = (\biguplus_{n \geq 0} X^{\bullet n}, \biguplus_{n \geq 0} \mc{X}^{\bullet
   n} ),
\end{aligned}$$
\noindent whose $\sigma$-field $\mc{X}_e^\bullet$
is the disjoint union of the measure theoretic $n$-ary
direct product of $\mc{X}$, on the set $X^{\bullet n}$. 
That is,
$\biguplus_{n \geq 0} \mc{X}^{\bullet
   n} = \{ \biguplus_{n \geq 0} A_n \mid A_n \in \mc{X}^{\bullet n} \}$.
Note by this definition,
$\mc{X}_{e}^\bullet$
is the $\sigma$-field generated by $\mathscr{P}^\bullet(\mc{X})$
and the subspace of $\mc{X}_{e}^\bullet$ restricted to $X^{\bullet n}$
coincides with the $n$-ary direct product of the measurable space $\mc{X}$:
 I.e.,
$$\begin{aligned}
\mathcal{X}_{e}^\bullet= \sigma (\mathscr{P}^\bullet(\mathcal{X}))
\quad \text{and} \quad \mathcal{X}_{e}^\bullet
\cap X^{\bullet n}  =  \mathcal{X}^{\bullet n} \quad \text{for any $n
   \geq 0$} 
\end{aligned}$$
In particular when $n=0$, $\mathcal{X}^{\bullet 0}$ is the only
   $\sigma$-field $\{ \emptyset, \{ \emptyset \} \}$
over $X^{\bullet 0}=\emptyset$.
}\end{defn}
\smallskip

\bigskip

In terms of category theory,
Definition \ref{fieldXnce} says
\begin{prop} \label{catfieldXnce}
In $\TKer$,
the measurable space $\msnce{X}$ of Definition \ref{fieldXnce}
is the countable infinite coproduct 
$\textstyle \coprod \limits_{n=0}^{\infty}
(X^{\bullet
 n},\mc{X}^{\bullet n} )$, whose inclusion $\inj_{m}^\infty$ 
from the $m$-th component is given;  
\begin{align*}
& \textstyle \msnce{X} \,   \cong \, \,
\textstyle \coprod \limits_{n=0}^{\infty}
(X^{\bullet
 n},\mc{X}^{\bullet n} )
\stackrel{\inj_{m}^\infty}{\longleftarrow}
(X^{\bullet
 m},\mc{X}^{\bullet m} ), \quad \mbox{where for $x_1 \ncdot \cdots
\ncdot  x_m \in X^{\bullet m}$ and $\msbf{A} \in \mc{X}_e^{\bullet}$} \\
&  \inj_{m}^\infty (
x_1 \ncdot \cdots
\ncdot  x_m,   \msbf{A}) = \delta (x_1 \ncdot \cdots
\ncdot  x_m,   \msbf{A} \cap X^{\bullet m}) 
= \delta (x_1 \ncdot \cdots
\ncdot  x_m,   \msbf{A} ) 
\end{align*}
Note that the injection factors through the colimit inclusion 
$\inj_l$ of the direct system of Remark \ref{infbpcol} such that 
$\inj_m^\infty = \inj_l \Comp \inj_{m}^l$ for any $l \geq m$. 

The infinite coproduct simultaneously 
becomes infinite biprducts, whose $m$-th projection $\pr_m^\infty$ is 
\begin{align*}
\pr_{m}^\infty (
x_1 \ncdot \cdots
\ncdot  x_k,   \msbf{B})=
\delta_{k,m} \, \delta (
x_1 \ncdot \cdots
\ncdot  x_k,   \msbf{B}) \quad \mbox{for $x_1 \ncdot \cdots
\ncdot  x_k \in X_e^\bullet$ \mbox{with $k \in \mathbf{N}$} and $\msbf{B} \in \mc{X}^{\bullet m}$.} 
\end{align*}
\end{prop}

%\bigskip

Because of Remark \ref{infbpcol},
the construction of Proposition \ref{catfieldXnce}
is closed inside the subcategory $\TKersfin$
of s-finite kernels (but not in $\TKerfin$ of finite kernels).

\smallskip

Finally, the exponential measurable space $\mse{X}$
is obtained by the following equivalent
characterisations of a $\sigma$-field $\mc{X}_e$.
\begin{prop}[$\sigma$-field $\mc{X}_e$ over $X_e$
(cf. Theorem 4.1 \cite{CP})] \label{fieldXe}
For a measurable space $\ms{X}$,
the following families of subsets of $X_e$ all coincide
with the $\sigma$-field $\sigma(\mathscr{P}(\mc{X}))$,
which is denoted by $\mathcal{X}_e$.

%\begin{enumerate}
%\setlength{\itemsep}{,5pt} 
%  \setlength{\parskip}{1mm} 

\vspace{1ex}
{\em

\noindent (i) The quotient $\mathcal{X}_e^\bullet$
wrt rearrangement $\{ A \subset X_e \mid F^{\mbox{-}1} (A) \in
	   \mc{X}_e^\bullet  \} $.

\vspace{.5ex}

% \item[(ii)]
\noindent (ii)  The projection of $\mathcal{X}_e^\bullet$ by $F$: I.e.,
the $\sigma$-field 
$\{ F(\msbf{A}) \mid  \msbf{A} \in \mathcal{X}_e^\bullet  \}$.
%This is the push forward measure from $\mathcal{X}_e^\bullet$
% with respect to the map $F :  X_e^\bullet \longrightarrow 	    X_e$.

\vspace{.5ex}

\noindent (iii) The smallest $\sigma$-field wrt which the counting
 	     functions $\cunt{A}$
 are measurable for all $A \in \mathcal{X}$.

\vspace{.5ex}

 \noindent (iv) The largest $\sigma$-field $\mc{Y}$ for $X_e$
having $\mc{X}$ as a subspace and such that the monoid product is
 	    measurable from $\mc{Y} \times \mc{Y}$ to $\mc{Y}$.

\vspace{.5ex}

\noindent (v)   The smallest $\sigma$-field for $X_e$ containing $\mc{X}$
 and for which the monoid product preserves measurability. 
\vspace{.5ex}

 \noindent (vi)  The smallest $\sigma$-field for $X_e$
containing $\mc{X}$ and
 	     closed under the symmetric product.
% \end{enumerate}
}
\end{prop}

\smallskip

\begin{defn}[exponential measurable space $\mse{X}$]
{\em
The measurable space $\mse{X}$, whose 
$\mc{X}_e$ is defined by Proposition \ref{fieldXe}
such that $\mc{X}_e=\sigma(\mathscr{P}(\mc{X}))$,
is called the {\em exponential measurable
 space} of $\ms{X}$.
}
\end{defn}

\smallskip

This section ends with a measure theoretic proposition
on isomorphisms relating the biproduct and the tensor via the exponential:

\begin{prop} \label{seely}
The following holds for any measurable spaces $\ms{X}$ and $\ms{Y}$:
\begin{enumerate}
 \item[(i)] 
$$\begin{aligned}
\left( \ms{X} \coprod \ms{Y} \right)_e & \, \cong \, \, 
\ms{X}_e \otimes \ms{Y}_e 
\end{aligned}$$

\item[(ii)]
%For the unit $\mc{T}$ of biproducts,
$\mc{T}_e= ( \{ \emptyset \}, \{ \emptyset,  \{ \emptyset \} \})$, which is
	   isomorphic
to the monoid unit $(I,\mathcal{I})$.
\end{enumerate}
\end{prop}

\begin{proof}
We prove (i) since
(ii) is direct, as the monoid identity
of the exponential monoid of Definition \ref{expmonoid}
is given by the empty
 sequence. \\
First, the monoid isomorphism $(X \uplus Y)_e \cong X_e \times Y_e$  between the largest measurable sets
of each side is given as follows; 
For any $\msbf{z}= z_1 \cdots z_n 
\in (X \uplus Y)_e$, there exist a rearrangement $\sigma$
and $ 0 \leq \ell \leq n$ such that 
$\msbf{z'}:=z_{\sigma(1)} \cdots z_{\sigma(\ell)} \in X_e$ and 
$\msbf{z''}:=z_{\sigma(\ell +1)} \cdots z_{\sigma(n)} \in Y_e$.
Note $\msbf{z'}$ and $\msbf{z''}$ are unique independently of
the choice of the rearrangement, 
thus  mapping $\msbf{x}$ to $(\msbf{z'}, \msbf{z''})$
gives a monoid isomorphism. \\
Second, the monoid iso is shown to induce the set theoretical
 isomorphism of the $\sigma$-fields of both sides
$(\mathcal{X} \uplus \mathcal{Y})_e \cong \mathcal{X}_e \otimes \mathcal{Y}_e
 $.

By the definition of the product of two measurable spaces
and Proposition \ref{fieldXe} (iii),
$\mathcal{X}_e \otimes \mathcal{Y}_e$ is the smallest $\sigma$-field
in which the product of counting functions
$\cunt{A} \cunt{B}: (\msbf{x},\msbf{y}) \mapsto \cunt{A} (\msbf{x}) 
\cunt{B}  (\msbf{y})$
becomes measurable for all $A \in \mathcal{X}$ and $B \in \mathcal{Y}$.
As the exponential function is one to one,
Proposition \ref{fieldXe} (iii) holds with $n_A$ replaced by $2^{n(A)}$.
Since $\cunt{A \uplus B}(\msbf{x} \msbf{y})= \cunt{A}(\msbf{x}) +
 \cunt{B}(\msbf{y})$, the following commutes so that the isomorphism
becomes that between the two $\sigma$-fields.
$$\xymatrix@R=1pt{
 & \mathbb{N} & \\
(\mathcal{X}
\coprod \mathcal{Y} )_e \ar[ru]^{2^{\cunt{\mbox{\tiny $A \uplus B$}}}}  & \cong  &
\mathcal{X}_e \otimes \mathcal{Y}_e \ar[lu]_{2^{\cunt{A}}
 2^{\cunt{B}}}}
$$
That is, 
$(\mathcal{X}
\coprod \mathcal{Y} )_e$ and 
$\mathcal{X}_e \otimes \mathcal{Y}_e$
are the smallest $\sigma$-fields making the respective functions 
$2^{\cunt{A \uplus B}}$ and $2^{\cunt{A}} 2^{\cunt{B}}$
measurable for all $A \in \mathcal{X}$ and $B \in \mathcal{Y}$.
\end{proof}

\begin{rem}[Seely isomorphism]{\em 
The isomorphism (i)
of Proposition \ref{seely}
is a Seely isomorphism \cite{RS} in
an appropriate category theoretical model
of linear logic,
as our binary biproduct models the logical connective
$\&$. The Seely isomorphism
is known derivable \cite{BS, Mellies} using category theoretic abstraction
from any linear exponential comonad
structure with product, which structure
will be obtained for a certain class of transition kernels
in the next Section \ref{comonadTKer} (cf. Theorem \ref{linexcomo}).
}\end{rem}


\subsection{Exponential Kernel $\kappa_e$ in s-finiteness}
\label{subsecexker}
This subsection concerns a categorical investigation in $\TKersfin$
on the exponential measurable spaces of Section \ref{subsecems}. 
This section starts with seeing the exponential acts not only on objects as
defined in Section \ref{subsecems} but on the morphisms on $\TKersfin$,
hence becomes an endofunctor.

\noindent ({\bf Notation}) For a measurable space $\ms{X}$ and 
$m \in \mathbb{N}$,
$$\begin{aligned}
X^{(m)}:= \{ \msbf{x} \in X_e \mid \cunt{X}(\msbf{x})=m  \}
\subset X_e
\end{aligned}$$
This divides the set $X_e$ into the following disjoint union:
\begin{align}
X_e = \biguplus_{n \geq 0} X^{(n)} \label{edisuni}
\end{align}
For any $A \in \mc{X}$, $A^{(m)}$ is defined same for the subspace $\mc{X} \cap A$.

\smallskip

\begin{defn}[exponential kernel $\kappa_e$] \label{kerex}{\em
In $\TKersfin$,
every transition kernel $\kappa: \ms{X} \longrightarrow \ms{Y}$
induces a corresponding exponential kernel
$\kappa_e: \mse{X} \longrightarrow \mse{Y}$,
which we shall define in (\ref{exker}).

In what follows, $\kappa^n$ denotes the $n$-ary cartesian product
$\kappa^n: \ms{X}^n \longrightarrow \ms{Y}^n$,
for which the $n$-ary cartesian product of an object is given by
$\ms{X}^n = (X^{\bullet n}, \mathcal{X}^{\bullet n})$. \\
The characterisation of Proposition \ref{catfieldXnce}
ensures the unique morphism $\kappa_e^\bullet$
from $(X_e^\bullet, \mathcal{X}_e^\bullet)$
to $(Y_e^\bullet, \mathcal{Y}_e^\bullet)$
in $\TKersfin$;
\begin{eqnarray}
 \kappa_e^\bullet & := \varinjlim_l  
 \displaystyle \coprod_{n=0}^{l}  (\inj_n^l \Comp \kappa^n ) 
 \simeq \coprod_{n=0}^\infty (\inj_{n}^\infty \Comp
 \kappa^n) 
 \label{ebul} 
\end{eqnarray}
See the following commutative diagram for the definition
(\ref{ebul} ):

$$\xymatrix@R=14pt{
(X_e^\bullet, \mc{X}_e^\bullet)   \ar@{=}[r]  & 
 \coprod \limits_{n=0}^{\infty}
(X^{\bullet
 n},\mc{X}^{\bullet n} )
\ar[rr]^{ \kappa_e^\bullet := 
\coprod \limits_{n=0}^\infty 
(\inj_{n}^\infty \Comp \kappa^n)} 
\ar@{-}[d]_{\cong}
& & 
 \coprod \limits_{n=0}^{\infty}
(Y^{\bullet
 n},\mc{Y}^{\bullet n} ) 
\ar@{-}[d]^{\cong} &
(Y_e^\bullet, \mc{Y}_e^\bullet) \ar@{=}[l]
 \\ 
& \varinjlim
\coprod \limits_n \ms{X}^n  \ar[rr]^{\varinjlim_l
\coprod \limits_{n=0}^l 
(\inj_{n}^l \Comp \kappa^n)} 
&  & 
\varinjlim  \coprod \limits_n \ms{Y}^n  & 
 \\ 
& \ar[u]^{\inj_l} 
\coprod \limits_{n=0}^l \ms{X}^n  
\ar[rr]^{\coprod \limits_{n=0}^l 
(\inj_{n}^l \Comp \kappa^n)} 
       & &
\coprod \limits_{n=o}^l \ms{Y}^n \ar[u]_{\inj_l}  & 
\\ 
& \ms{X}^m    \ar[rr]^{\kappa^m} 
\ar[u]^{\inj_{m}^l} 
\ar@/^50pt/[uuu]^{\inj_{m}^\infty}
&  &
  \ms{Y}^m \ar[u]_{\inj_m^l} 
\ar@/_50pt/[uuu]_{\inj_{m}^\infty}
&
}$$

\smallskip
Explicitly,
$\kappa_e^\bullet$ is defined for any $(x_1, \ldots, x_n) \in X_e \cap X^{(n)}$
with any $n$ and $B \in \mc{Y}^\bullet_e$,
\begin{eqnarray}
\kappa_e^\bullet ((x_1, \ldots ,x_n), B) = 
\kappa_e^\bullet ( (x_1, \ldots , x_n), B \cap Y^{\bullet n}) 
 = \kappa^n ( (x_1, \ldots , x_n), B \cap Y^{\bullet n}) \label{expebot}
\end{eqnarray}
This is by virtue that $B= \sum_{i=0}^\infty (B \cap Y^{\bullet i})$
and $\inj_{n}^\infty ((y_1, \ldots , y_n), B \cap 
Y^{\bullet m})=0$ unless $m=n$. 

\bigskip

Since $\kappa_e^\bullet :X_e^\bullet \times \mathcal{Y}_e^\bullet
 \longrightarrow \zeroinf$ is a transition kernel and the forgetful $F:
\msed{Y} \longrightarrow  \mse{Y}$
is
$(\mathcal{Y}_e^\bullet, \mathcal{Y}_e)$-measurable,
the pushforward measure on $\mc{Y}_e$ along $F$ is defined
for each fixed $(x_1, \ldots , x_m) \in \mathcal{X}_e^\bullet$,
which we denote (under the convention of Definition \ref{pfm}) by
$$\begin{aligned}
\kappa_e^\bullet(
(x_1, \ldots , x_m) , -) \Comp \FI: \quad \mathcal{Y}_e  \longrightarrow
 \zeroinf
\end{aligned}$$
This determines the following transition kernel,
denoted by $\kappa_e^\bullet \Comp \FI$,
from $\msed{X}$ to $\mse{Y}$:
\begin{align} \label{exdotker}
& \kappa_e^\bullet \Comp \FI:
\quad  X_e^\bullet \times \mathcal{Y}_e  \longrightarrow \zeroinf 
& ((x_1, \ldots , x_m), \msbf{B}) \longmapsto 
\kappa_e^\bullet ((x_1, \ldots , x_m), \FI(\msbf{B}))
\end{align}
for any $(x_1, \ldots , x_m) \in X_e^\bullet$
and any $\msbf{B} \in \mc{Y}_e \cap Y^{(m)}$ 
with any $m \in \mathbb{N}$.

\bigskip

Directly from the definition,
for any permutation $\sigma \in \mathfrak{S}_m$,
\begin{align}
\kappa_e^\bullet ((x_1, \ldots , x_m) , \FI (\msbf{B})) & =
\kappa_e^\bullet ((x_{\sigma(1)}, \ldots , x_{\sigma(m)}) ,
   \FI(\msbf{B})). \label{kabperm}
\end{align}
(\ref{kabperm}) is implied using (\ref{expebot})
by the following (\ref{invperm}) for any $D \in
\mc{Y}_e^\bullet \cap Y^{\bullet m}$
and any permutation $\sigma \in \mathfrak{S}_m$,
\begin{align}
\kappa^m ((x_1, \ldots , x_m) , \sigma (D)) & =
\kappa^m ((x_{\sigma^{-1}(1)}, \ldots , x_{\sigma^{-1}(m)}) , D) \label{invperm}
\end{align}
It is sufficient to check (\ref{invperm}) for any rectangle $D := C_1 \times
\cdots \times C_m$ such that $1 \leq \forall i \leq m \,$ $C_i \in \mc{Y}$,
but for which (\ref{invperm}) is
$\kappa(x_{\sigma(1)}, C_1) \cdots \kappa(x_{\sigma(m)}, C_m)
= \kappa(x_1, C_{\sigma^{-1}(1)}) \cdots \kappa(x_m,
 C_{\sigma^{-1}(m)})$
by the definition of the product measurable space.

\bigskip
Observing 
$$\FI{(\msbf{B} \cap Y^{(n)})}=
\FI{(\msbf{B})} \cap Y^{\bullet n}$$
We thus finally define $\kappa_e:
X_e \times \mathcal{Y}_e  \longrightarrow \zeroinf$
for any $x_1 \cdots x_m \in X_e$ and 
any $\msbf{B} \in \mc{Y}_e \cap Y^{(m)}$
\begin{align}
\kappa_e (x_1 \cdots x_m , \msbf{B}) & :=
(\kappa_e^\bullet \Comp \FI) \, ((x_1, \ldots , x_m) , \msbf{B})
 \nonumber \\
& =
\kappa_e^\bullet ((x_1, \ldots , x_m) , \FI(\msbf{B})), \label{exker}
\end{align}
which definition does not depend on the ordering of $x_1 \cdots x_m$.

\smallskip

We need to check 
$\kappa_e$ defined above is 
a transition kernel:
The second argument of $\kappa_e$ giving a measure over
 $\mathcal{Y}_e$ is direct by the definition (\ref{exker})
because so does the second argument of $\kappa_e^\bullet$. \\
For measurability in $\mathcal{X}_e$ for the first argument
of $\kappa_e$, by virtue of Proposition \ref{fieldXe} (ii), it suffices to show that
 $(\kappa_e)^{\mbox{-}1}
(\FI(-), \msbf{B})$ is measurable in $\mathcal{X}_e^\bullet$.
But this is derived from the measurability of
the first argument $\kappa_e^\bullet$ in $\mathcal{X}_e^\bullet$
because by the commutative diagram below, yielding \\ 
$
(\kappa_e)^{\mbox{-}1} (\FI(-), \msbf{B})
=
(F \times \operatorname{Id})^{\mbox{-}1} \Comp (\kappa_e)^{\mbox{-}1}
 (-,  \msbf{B})
=
(\kappa_e \Comp (F \times \operatorname{Id}))^{\mbox{-}1} (-, \msbf{B})
=
(\kappa_e^\bullet)^{\mbox{-}1} (-, \FI (\msbf{B}))
$.
$$\xymatrix@C=50pt
{
 X_e \times \mc{Y}_e \ar[dr]^{\kappa_e} & \\
X_e^\bullet  \times \mc{Y}_e \ar[u]^{F \times \operatorname{Id}}
\ar[r]_{\kappa_e^\bullet \Comp \FI}
& \zeroinf
}
$$
}\end{defn}

\bigskip

The so constructed $\kappa_e$ is s-finite, as
$\kappa_e^\bullet$ resides in $\TKersfin$ (cf. Proposition \ref{catfieldXnce})
and s-finiteness is closed under the push forward along $F$
(cf. Remark \ref{clopf}).

\bigskip
In order to show the functoriality of the exponential over kernels
of Definition \ref{kerex}, we prepare the following lemma
on $\pi$-system and Dynkin system.

\smallskip
\begin{lem}[for Proposition \ref{funcex}:
a $\pi$-system for $\mc{X}^{\bullet }_e$] \label{piDynkin} ~~
Let $\ms{X}$ be a measurable space.
\begin{enumerate}
 \item[(a)] 
For any $n \geq 0$,
the following family consisting of subsets of $X^{\bullet n}$
\begin{align}
\mathcal{D}_n := 
\left\{
\biguplus_{k=0}^\infty C_{k,1} \times
 \cdots \times C_{k,n} \mid C_{k,j} \in
\mc{X}
\right\} \label{famPi}
\end{align}
is both (i) a $\pi$-system and (ii) a Dynkin system. \\
{\em Recall that a nonempty family of subsets of a universal set
is a $\pi$-system if the family is closed under finite intersections.
It is a Dynkin system if the family contains $\emptyset$
and is closed both under complements and
under countable disjoint unions.
}
\item[(b)]
$\mathcal{D}_n$ becomes a $\sigma$-field,
hence coincides with the measurable space $\mc{X}^{\bullet n}$
for any $n \geq 0$. Thus we characterise
$$\mc{X}_e^\bullet = \biguplus_{n \geq 0} \mathcal{D}_n$$
\end{enumerate}

\end{lem}
\begin{proof}
As (b) is a consequence of (a) by Dynkin Theorem (cf. \cite{Bau}
for the theorem)
stating that any Dynkin system which is also a $\pi$-system
is a $\sigma$-field, we prove (a): \\
(a) (i) Direct by $(A_1 \times
 \cdots \times A_n) \cap
(B_{1} \times
 \cdots \times  B_{n})
= 
(A_{1} \cap  B_{1})
 \times
\cdots 
\times 
(A_{n} \cap B_{n})
$. \\
(ii) As the empty set is contained in (\ref{famPi}), we check the
other two conditions:

\noindent (Closedness under countable disjoint unions)
Immediate from the definition (\ref{famPi}), by observing \\
$\biguplus_{i=1}^\infty \biguplus_{k=0}^\infty C_{k,1}^i \times
 \cdots \times C_{k,n}^i
= 
\biguplus_{m=0}^\infty
\biguplus_{k+i=m}
C_{k,1}^i \times
 \cdots \times C_{k,n}^i$.

\noindent
(Closedness under the complement) First we observe 
$
(A_1 \times
 \cdots \times A_n)^\complement
= 
\biguplus_{ s \in \{1, \complement \}^{[n]}}
A_1^{s(1)} \times
 \cdots \times A_n^{s(n)},$
where $[n] := \{1,\ldots, n\}$ and $X^1=X$.
Then using De Morgan and distribution
of intersection over union:
$(\biguplus_{k=0}^\infty C_{k,1} \times
 \cdots \times C_{k,n} )^\complement
=
\bigcap_{k=0}^\infty (C_{k,1} \times
 \cdots \times C_{k,n})^\complement
=
\bigcap_{k=0}^\infty 
\biguplus_{ s \in \{1, \complement \}^{[n]}}
C_{k,1}^{s(1)} \times
 \cdots \times C_{k,n}^{s(n)}
=
\biguplus_{ s \in \{1, \complement \}^{[n]}}
 (\bigcap_{k=0}^\infty C_{k,1}^{s(1)})
\times
 \cdots \times 
(\bigcap_{k=0}^\infty C_{k,n}^{s(n)}$),
which belongs to (\ref{famPi})
as $\forall j \in \{ 1, \ldots , n \}$
$\bigcap_{k=0}^\infty C_{k,j}^{s(j)} \in \mc{X}$.
\end{proof}

\smallskip

\begin{prop}[functoriality of $(~)_e$ ]\label{funcex}~\\
$(~)_e$ of Definition \ref{kerex} becomes an endofunctor on the category $\TKersfin$.
\end{prop}
\begin{proof}
The condition $(\iota \Comp \kappa)_e = \iota_e \Comp \kappa_e$
%$\mse{X} \longrightarrow \mse{Z}$ 
for $\kappa: \ms{X} \longrightarrow \ms{Y}$
and $\iota: \ms{Y} \longrightarrow \ms{Z}$
is proved. For this, the following variable change (cf. Definition \ref{pfm})
plays a crucial role: 

\smallskip

\noindent(Variable change of integral along
$F: \msed{Y} \longrightarrow \mse{Y}$)
\begin{align} \label{IWPFM}
 \int_{Y_e}
\kappa_e^\bullet (-, \FI(d \msbf{y})) \,
\iota_e (\msbf{y}, \sim)
= 
\int_{Y_e^\bullet}
\kappa_e^\bullet (-, d \vec{\msbf{y}} ) \,
\iota_e (F(\vec{\msbf{y}}), \sim)
\end{align}
where $\vec{\msbf{y}}=(y_1, \ldots , y_m) \in Y_e^\bullet$ so that
$F(\vec{\msbf{y}}) =y_1 \cdots y_m = \msbf{y} \in \mc{Y}_e $ for any $m$.

\noindent The equation (\ref{IWPFM}) is that
of Definition \ref{pfm}
when the push forward measure $\mu'=\mu \Comp \FI$
is defined for $\mu(\msbf{B}):=\kappa_e^\bullet(-, \msbf{B})$
with any fixed $-$ (cf. (\ref{exdotker})), 
and the measurable function $g$ on $\mc{Y}_e$
is given by $\iota_e (\msbf{y}, \sim)$ 
with any fixed $\sim$.

\smallskip

\noindent
For any $x_1 \cdots x_n \in X_e$ and any $\msbf{C} \in
 \mc{Z}_e \cap Z^{(n)}$ such that, by Lemma \ref{piDynkin} (b), 
$$\FI(\msbf{C})=\biguplus_{k=0}^\infty C_{k,1} \times
 \cdots \times C_{k,n} \in \mc{Z}_e^\bullet \cap Z^{\bullet n},$$
\begin{align*}
& (\iota_e \Comp \kappa_e) (x_1 \cdots x_n, \msbf{C}) \\
&   \textstyle = 
\int_{Y_e} \kappa_e (x_1 \cdots x_n, d \msbf{y})
\, \iota_e(\msbf{y}, \msbf{C}) \\
& \textstyle =
\int_{Y_e} \kappa_e^\bullet ( (x_1, \ldots ,x_n), \FI(d \msbf{y}))
\, \iota_e(\msbf{y}, \msbf{C})  \tag*{by the def of
 $\kappa_e$ in (\ref{exker})}
 \\
& \textstyle = 
\int_{Y_e^\bullet} \kappa_e^\bullet ( (x_1, \ldots, x_n), d \vec{\msbf{y}})
\, \iota_e(F(\vec{\msbf{y}}), \msbf{C})  \tag*{by (\ref{IWPFM}) of variable change} \\
& 
\textstyle = \int_{Y^{\bullet n}} \kappa_e^\bullet ((x_1, \ldots x_n), d (y_1, \ldots , y_n))
\, \iota_e(F((y_1, \ldots , y_n)), \msbf{C})
 \tag*{by $\vec{\msbf{y}}=(y_1, \ldots, y_n) \in Y_e^\bullet \cap Y^{\bullet n}$} \\
& \textstyle = \int_{Y^{\bullet n}} \kappa_e^\bullet ( (x_1, \ldots, x_n), d (y_1, \ldots , y_n))
\, \iota_e^\bullet ((y_1 \cdots y_n), \FI(\msbf{C}))
\tag*{by the def of $\iota_e$}
\\
& \textstyle = \int_{Y^{\bullet n}} \kappa^n ( (x_1, \ldots ,x_n), d (y_1, \ldots , y_n))
\, \iota^n ((y_1, \ldots , y_n), \FI(\msbf{C}))
\tag*{by (\ref{expebot}) with the choice $\msbf{C}$} \\
& 
\textstyle = \int_{Y^{\bullet n}} \kappa^n ( (x_1, \ldots ,x_n), d (y_1, \ldots , y_n))
\, \iota^n ((y_1, \ldots , y_n), \biguplus \limits_{k=0}^\infty C_{k,1} \times
 \cdots \times C_{k,n}) \\
& 
\textstyle =  \int_{Y^{\bullet n}} \sum \limits_{k=0}^\infty \kappa^n ( (x_1, \ldots ,x_n), d (y_1, \ldots , y_n))
\, \iota^n ((y_1, \ldots , y_n),  C_{k,1} \times
 \cdots \times C_{k,n})  \tag*{by $\sigma$-additivity} \\
& 
\textstyle = \sum \limits_{k=0}^\infty \int_{Y^{\bullet n}} \kappa^n ( (x_1, \ldots ,x_n), d (y_1, \ldots , y_n))
\, \iota^n ((y_1, \ldots , y_n),  C_{k,1} \times
 \cdots \times C_{k,n}) \\  %\tag*{exchanging the countable sum and the integral}
& 
\tag*{commuting integral over countable sum of {\em non-negative} measurable functions \footnotemark} \\
& \textstyle = \sum \limits_{k=0}^\infty \int_{Y^{\bullet n}} 
\prod \limits_{i=1}^n \kappa( x_i, d y_i)
\prod \limits_{j=1}^n \iota (y_j,  C_{k, j})  \\
%\end{align*}
%\begin{align*}
& \textstyle = \sum \limits_{k=0}^\infty \prod \limits_{i=1}^n  \int_{Y} 
\kappa( x_i, d y)  \, \iota (y,  C_{k, i}) \tag*{by Fubini-Tonelli} \\
& \textstyle = \sum \limits_{k=0}^\infty \prod \limits_{i=1}^n  
(\iota \Comp \kappa) (x_i, C_{k, i}) \tag*{by the def of $\iota \Comp
 \kappa$}  \\
& \textstyle = \sum \limits_{k=0}^\infty 
(\iota \Comp \kappa)^n ((x_1, \ldots ,x_n),
C_{k,1} \times
 \cdots \times C_{k,n}) \tag*{by the product measure} \\
& =
(\iota \Comp 
\, \kappa )^n ((x_1, \ldots , x_n), \FI(\msbf{C})) 
\tag*{by $\sigma$-additivity}
\\
& 
 = 
(\iota \Comp 
\, \kappa )_e^\bullet ((x_1, \ldots , x_n), \FI(\msbf{C})) 
\tag*{by the def of $(\iota \Comp 
\, \kappa )_e^\bullet$}  \\
& = 
(\iota \Comp 
\, \kappa )_e (x_1 \cdots x_n, \msbf{C})
\tag*{by the def of $(\iota \Comp 
\, \kappa )_e$}
\end{align*}
\footnotetext{By the monotone convergence theorem together with the commuting
integral with finite sum}
% \noindent (Fubini-Tonelli Theorem \ref{FT})
% The integration of
% $\kappa^m ((x_1, \ldots ,x_m), d \msbf{y}) \, \iota^m (\msbf{y}, -)$
%  over $Y^m$ is equal to the $m$-iteration of that over $Y$.
\end{proof}

\begin{rem}[The exponential construction $(-)_e$ preserves
s-finiteness, but not finiteness.]{\em
In addition that the class retains Fubini-Tonelli
for the functorial monoidal product in Section \ref{mpacb},  
the class of s-finiteness is employed in this paper
because it makes $(-)_e$ an endofunctor as shown above.
E.g., its restriction 
on $\TKerfin$ of the finite kernels
is no more an endofunctor but from $\TKerfin$ to $\TKersfin$.
%   as in the following calculation searching for
%  an upper bound of the measure $\kappa_e (\msbf{x}, -):
% \mc{X}_e \longrightarrow \zeroinf$ for a fixed $\msbf{x} \in X_e$:
% %For any $m$, take any $\msbf{x}= x_1 \cdots x_m \in X_e \cap X^{(m)}$, 
% Let $\msbf{x}=x_1 \cdots x_m$ and $\msbf{B} \in \mc{Y}_e \cap Y^{(m)}$. \\
% $\begin{aligned}
% \textstyle \kappa_e (x_1 \cdots x_m , \,\, \msbf{B}) =
%  \kappa^m(
% (x_1, \ldots , x_m),  \, \, \FI(\msbf{B}) ) \quad 
%      \leq 
% %\sum \limits_{\sigma \in \mathfrak{S}_m}
% \kappa^m(
% (x_1, \ldots , x_m) , \, \,
% Y^{\bullet m} ) \\ \textstyle 
%  \textstyle =  
% \prod \limits_{i=1}^m \kappa(x_i , \, \, Y)
% \leq r^m
% \leq \sum \limits_{m=0}^{\infty} r^m
% \quad \text{where} \, \, r=\sup_{x \in X} \kappa (x,Y)
% \end{aligned}$ \\
% Thi s shows that when $r < 1$ the infinite geometric series giving
%an upperbound converges finite,
%though otherwise in general the upper bound diverges to $\infty$.
}\end{rem}

\subsection{A Linear Exponential  Comonad over $\opTKersfin$}
\label{comonadTKer}
The exponential presented in Section \ref{subsecems} and Section
 \ref{subsecexker}
is shown to provide a linear exponential comonad over 
the monoidal category $\opTKersfin$ with countable biproducts,
hence a categorical model of the exponential modality of linear logic
\cite{BS,HSha, Mellies}.

Due to the asymmetry between the first (measures) and the second (measurable
functions) arguments of transition kernels in continuous measurable spaces,
the exponential comonad considered 
in Subsection \ref{comonadTKer} is for the opposite category
$\opTKersfin$ \footnote{Our choice of opposite later in Section \ref{closedTK}
turns out to coincide with Danos-Ehrhard's (left) choice of permutation
of their formulation of exponential in $\Pcoh$. See Remark \ref{whyop}.}.

\smallskip

\noindent
{\bf Notation for morphisms in the opposite $\opTKersfin$}:
The category considered in this section is the opposite category
$\opTKersfin$ so that the composition is
converse to $\TKersfin$: In $\opTKersfin$, a morphism
$\kappa : \ms{X} \longrightarrow
\ms{Y}$ is a transition kernel from $\ms{Y}$ to $\ms{X}$.
Accordingly a morphism $\kappa$ is denoted by $\kappa(A, y)$ meaning that
its left (resp. right) argument determines a measure (resp. a
measurable function).
In particular, the Dirac delta
measure which is the identity morphism on $\ms{X}$
is written by $\Dd{x}{A}$.  Accordingly, 
the composition of two morphisms 
$\kappa(A, y) : \ms{X} \longrightarrow
\ms{Y}$ and $\iota(B, z) : \ms{Y} \longrightarrow
\ms{Z}$ in $\opTKersfin$ is
\begin{align*}
\iota \Comp \kappa (A,z) = &
\Int{Y}{\kappa(A,y)}{\iota(dy,z)} 
\end{align*}

\bigskip

\noindent {\bf Typographic Convention}:
In what follows, the following typography is used 
to discriminate levels of the exponential
measurable spaces:
$x,y , z, \ldots \in X$ and  
$A, B, C, \ldots \in \mc{X}$ for $\ms{X}$.
$\msbf{x}, \msbf{y} , \msbf{z}, \ldots \in X_e$ and 
$\msbf{A}, \msbf{B}, \msbf{C} \ldots \in \mc{X}_e$ for $\mse{X}$.
$\mathfrak{x}, \mathfrak{y}, \mathfrak{z}, \ldots \in X_{ee}$ and 
$\mathfrak{A}, \mathfrak{B}, \mathfrak{C}, \ldots \in \mc{X}_{ee}$
for $\msee{X}$.
%Finally, $\mathscr{A} \in \mc{X}_{eee}$.

\bigskip

We recall the definition of linear exponential comonad.
%Let us recall the definition of linear exponential comonad.

\begin{defn}[linear exponential comonad \cite{HSha, Mellies}]{\em
Let $({\cal C}, \otimes, I)$ be a symmetric monoidal category.
A linear exponential comonad on ${\cal C}$ is a monoidal comonad
$$\left(!: {\cal C} \longrightarrow {\cal C}, \, 
\stor{}: \, ! \longrightarrow \, !!, \,  \di{}: \, ! \longrightarrow
 \operatorname{Id}_{\cal C}, \,  \mon{X}{Y}: \, ! X \otimes ! Y \longrightarrow
 !(X \otimes Y), \,  {\sf m}_I : I \longrightarrow \, !I \, \right)$$
equipped with two monoidal natural transformations
$\con{}: \, ! \longrightarrow \Delta \, \Comp \, !$
(with $\Delta$ denoting the diagonal functor for the tensor) and
$\wk{}: \, ! \longrightarrow I$ such that the following holds for each $X$: 
\begin{itemize}
\item
 $(!X, \con{X}, \wk{X})$ forms a commutative comonoid.
\item $\con{X}$ is a coalgebra morphism from $(!X,\stor{X})$ to
           $(!X \otimes \, !X, \, m_{!X,!X} \Comp (\stor{X} \otimes \stor{X}))$.
  \item 
$\wk{X}$ is a coalgebra morphism from $(!X,\stor{X})$ to
           $(I, \, {\sf m}_I)$.
\item 
$\stor{X}$ is a comonoid morphism from 
$(!X, \con{X}, \wk{X})$ to 
$(!!X, \con{!X}, \wk{!X})$.
 \end{itemize}
}\end{defn}

\bigskip
We start to construct the structure maps in $\opTKersfin$
for the linear exponential comonad.
\begin{prop}[Dereliction]\label{natder}~\\{\em
$\begin{aligned}
 \di{\mathcal{X}} : \mse{X} \longrightarrow \ms{X}
\end{aligned}$
is defined for $\msbf{A} \in \mathcal{X}_e$ and $x \in X$
\begin{align*}
  \di{\mathcal{X}}(\msbf{A},x) := \Dd{x}{\msbf{A} \cap X^{(1)}}
%= 
%\begin{cases}
%   1 & \text{$n=1$ and $\msbf{x}=x_1 \in A$} \\
%    0 & \text{otherwise}
%  \end{cases}
\end{align*}
Recall that $\msbf{A} \cap X^{(1)} \subset X$.} \\
Then, this gives a natural transformation ${\sf d}: (~)_e
\longrightarrow \operatorname{Id}_{\opTKersfin}$.
\end{prop}
\begin{proof} ~\\
Let $\msbf{A} \in \mc{X}_e$ and $y \in Y$.
Given $\kappa : \mc{X} \rightarrow \mc{Y}$,
$\di{\mc{Y}}  \Comp \kappa_e (\msbf{A}, y) =
 \Int{Y_e}{\kappa_e(\msbf{A}, \msbf{z})}{\di{\mc{Y}} (d \msbf{z}, y)}
= \Int{Y_e \cap Y^{(1)}}{\kappa_e(\msbf{A}, \msbf{z})}{\Dd{y}{d \msbf{z} \cap Y^{(1)}}}
= \kappa_e(\msbf{A}, y)$.
While, $\kappa \Comp \di{\mc{X}} (\msbf{A}, y) =
\Int{X}{\di{\mc{X}} (\msbf{A}, x)}{\kappa(dx, y)}
= 
\Int{X}{\Dd{x}{\msbf{A} \cap X^{(1)}}}{\kappa(dx, y)}
= 
\kappa(\msbf{A} \cap X^{(1)}, y).$
The both HSs coincide because 
$\kappa_e(\msbf{A}, y)
=
\kappa(\msbf{A} \cap X^{(1)}, y)$ for $y \in Y_{e} \cap Y^{(1)}$.
\end{proof}

\bigskip

In order to introduce the storage in Proposition \ref{stor}, we prepare;
\begin{defn}[$\abs{~}:  X_{ee} \longrightarrow X_e$] \label{mapabs}{\em 
 For a set $X$, the mapping
$\abs{~ }$ 
is defined by 
\begin{align*}
 \abs{~}:  X_{ee} \longrightarrow X_e & & 
\msbf{a}_1 \cdots
\msbf{a}_n  \longmapsto
\abs{   \msbf{a}_1 \cdots
\msbf{a}_n  }:=a_{11} \cdots a_{1 k_1} \cdots a_{n1} \cdots 
a_{n k_n}
\end{align*}
where every $\msbf{a}_i \in X_e$
is $a_{i1} \cdots a_{i k_i}$ with each $a_{ij} \in X$.
}\end{defn}

\noindent Note:
\begin{itemize}
 \item[-] $\abs{~}$
on $X_{ee} \cap (X_e)^{(1)}$ is the identity.
That is, when $n=1$ %in Definition \ref{mapabs} 
so that
$\msbf{a}_1 \in X_{ee}$, 
it holds $\abs{  \msbf{a}_1 } = \msbf{a}_1$.

\item[-] $\abs{~}$ on $X_{ee} \cap (X_e \cap X^{(1)})^{(n)}$ is the identity. That is, when
	$k_i=1$ for all $i=1, \ldots, n$ %in Definition	\ref{mapabs} 
so that $a_{11} \cdots a_{n1} \in X_{ee}$,
it holds $\abs{a_{11} \cdots a_{n1}} = a_{11} \cdots a_{n1}$.
\end{itemize}

\smallskip

For any $\msbf{A} \in \mc{X}_e$, its inverse image along $\abs{-}$
is defined by
\begin{align*}
 \absI{\msbf{A}}= \{ \mathfrak{y} \in X_{ee} \quad \mbox{such that} 
\quad \abs{\mathfrak{y}} \in \msbf{A} \}
\end{align*}
The following lemma \ref{ee-emeasurable} ensures that the inverse
image $\absI{\msbf{A}}$ belongs to $\mc{X}_{ee}$.

\smallskip

\smallskip

\begin{lem} \label{ee-emeasurable}
The function $\abs{~}$ of Definition \ref{mapabs}
is $(\mc{X}_{ee}, \mc{X}_e)$-measurable.
\end{lem}
\begin{proof}{}
For any $\msbf{A} \in \mc{X}_e \cap X^{(n)}$ with an arbitrary $n$,
we show that  $\absI{\msbf{A}}$ (i.e., the the inverse image of $\msbf{A}$
along $\abs{-}$) belongs to $\mc{X}_{ee}$.
But this is equivalent to show that the inverse image of $\msbf{A}$ along the composition  
$\xymatrix{ 
\abs{-} \Comp \, \mathbf{F}:
(X_e)_e^\bullet 
\ar[r]^(.6){\mathbf{F}} & (X_e)_e
\ar[r]^{|~|} & X_e}$, 
\begin{align} \label{invFabs}
(\abs{-} \Comp \, \mathbf{F} )^{\mbox{-}1}(\msbf{A})
= \mathbf{F}^{\mbox{-}1} (\absI{\msbf{A}})
\end{align}
belongs to $(\mc{X}_e)_e^\bullet$,
for which $\mathbf{F}$ denotes the forgetful map
in Definition \ref{nam} for the adequate type.
Observe that (\ref{invFabs}) becomes a subset of $\textstyle \biguplus
X^{(n_1)} \times \cdots \times X^{(n_k)}$, whose union ranges over
$(n_1, \ldots , n_k)$s such that  $\textstyle \sum_{i=1}^k n_i = n$.
In what follows in the proof all the $\textstyle \biguplus$ is the same as this.

Consider the $k$-folding cartesian product $F^k$ of
the forgetful $F: X_e^\bullet \longrightarrow X_{e}$,
whose restriction on the domain yields
$X^{\bullet n_1} \times \cdots \times X^{\bullet n_k}
\longrightarrow
X^{(n_1)} \times \cdots \times X^{(n_k)}$.
The inverse image of (\ref{invFabs}) along the union
$\biguplus$ of the foldings
coincides with $\FI (\msbf{A})$. That is, 
\begin{align} \label{invFkFabs}
\textstyle 
(\biguplus F^k)^{\mbox{-}1}
(\ref{invFabs}):=
\biguplus
(F^k)^{\mbox{-} 1} ( \mathbf{F}^{\mbox{-}1}(\absI{\msbf{A}}))
=
\FI (\msbf{A}),
\end{align}
whose RHS belongs to
$\mc{X}_e^{\bullet} \cap X^{\bullet n}$ by the choice $\msbf{A}$.
This means (\ref{invFabs}) belongs to
\begin{align} \label{setinvFkFabs}
\textstyle
\biguplus 
(\mc{X}_e \cap X^{(n_1)}) \times \cdots \times 
(\mc{X}_e \cap X^{(n_k)}) 
\end{align}
as $\mc{X}_e^{\bullet} \cap X^{\bullet n}
= \biguplus 
(F^k)^{\mbox{-} 1}(
(\mc{X}_e \cap X^{(n_1)}) \times \cdots \times 
(\mc{X}_e \cap X^{(n_k)}))$.
The assertion has been proven as $(\ref{setinvFkFabs}) \subset 
\biguplus (\mc{X}_e)^{\bullet k}
\subset (\mc{X}_e)_e^\bullet$.
\end{proof}

\smallskip

\begin{prop}[Storage]\label{stor}{\em
Storage, also called digging, $\stor{\mathcal{X}}: \mse{X} \longrightarrow \msee{X}$ is defined
for $\msbf{A} \in \mc{X}_e$ and $\mathfrak{y} \in X_{ee}$
\begin{align*}
\stor{\mathcal{X}}(\msbf{A}, \mathfrak{y}) &  :=
\Dd{\abs{\mathfrak{y}}}{\msbf{A}} =
\Dd{\mathfrak{y}}{\absI{\msbf{A}}}
\end{align*}
}
Then, this gives a natural transformation ${\sf s}: (~)_e
\longrightarrow (~)_{ee}$.
\end{prop}
\begin{proof}
We show that $\kappa_{ee} \Comp \stor{\mc{X}}
= \stor{\mc{Y}} \Comp \kappa_e$
for any $\kappa : \mc{X} \longrightarrow \mc{Y}$.
Let $\msbf{A} \in \mc{X}_e$ and 
$\msbf{y_1} \cdots \msbf{y}_k \in Y_{ee}$.
\begin{align*}
LHS(\msbf{A}, \msbf{y_1} \cdots \msbf{y}_k)
&  \textstyle = \Int{X_{ee}}{\stor{\mc{X}} (\msbf{A},
\msbf{x}_1 \cdots \msbf{x}_k)}{
\kappa_{ee} ( d \msbf{x}_1 \cdots \msbf{x}_k, 
\msbf{y}_1 \cdots \msbf{y}_k )}\\
&  \textstyle 
 = \Int{X_{ee}}{\Dd{\msbf{x_1} \cdots \msbf{x}_k}{\absI{A}}}
{\kappa_{ee} ( d \msbf{x_1} \cdots \msbf{x}_k, 
\msbf{y_1} \cdots \msbf{y}_k)} =
\kappa_{ee} ( \absI{\msbf{A}}, 
\msbf{y_1} \cdots \msbf{y}_k) 
%\end{align*}
%\begin{align*}
\\ RHS(\msbf{A}, \msbf{y_1} \cdots \msbf{y}_k) =
& 
\textstyle \Int{Y_e}{\kappa_e (\msbf{A}, \msbf{z})}{\stor{\mc{Y}}
(d \msbf{z}, \msbf{y_1} \cdots \msbf{y}_k )} \\
& \textstyle = 
\Int{Y_e}{\kappa_e (\msbf{A}, \msbf{z})}{
\Dd{\abs{\msbf{y_1} \cdots \msbf{y}_k}}{d \msbf{z}}} =
\kappa_e (\msbf{A}, \abs{\msbf{y_1} \cdots \msbf{y}_k})
\end{align*}
The both HSs coincide by the following Lemma \ref{invee}.
\end{proof}

\begin{lem} \label{invee}
For any $\kappa: \mc{X} \longrightarrow \mc{Y}$, the following holds for
any $\msbf{A} \in \mc{X}_e$ and $\mathfrak{y} \in \mc{X}_{ee}$:
\begin{align*}
\kappa_{ee} (\absI{\msbf{A}}, \mathfrak{y})
=
\kappa_{e} (\msbf{A}, \abs{\mathfrak{y}})  
\end{align*}
\end{lem}

\begin{proof} 
Let
$\mathfrak{y}$ be $\msbf{y}_1 \cdots \msbf{y}_r \in \mc{X}_{ee}$
so that $\msbf{y}_i = y_{i1} \cdots y_{in_i} \in \mc{X}_e$ with $i=1,
 \ldots ,r$.
Then it suffices to consider
$\msbf{A} \in \mc{X}_e \cap X^{(n)}$
with $n=\sum_{i=1}^r n_i$. \\
In the following $\mathbf{F}$ and $F$ are the same as
in the proof of
Lemma \ref{ee-emeasurable}. \\

$\begin{aligned}
& RHS (\msbf{A}, \abs{\msbf{y}_1 \cdots
  \msbf{y}_r}) 
= \kappa_e (\msbf{A}, y_{11} \cdots y_{1n_1} \cdots y_{r1} \cdots
  y_{rn_r}) \\ &
= 
\kappa_e^\bullet (\FI (\msbf{A}), 
(y_{11},  \ldots , y_{1n_1},  \ldots ,  y_{r1},  \ldots
  y_{rn_r})) &  \\ & 
= 
\kappa^{n} (\FI (\msbf{A}) \cap X^{\bullet n}, 
(y_{11},  \ldots , y_{1n_1},  \ldots ,  y_{r1},  \ldots
  y_{rn_r})) & \mbox{by the choice $\msbf{A}$ with
$n=\sum_{i=1}^r n_i$.}
 \end{aligned}$

\bigskip

\noindent In the following $\diamondsuit$ is a short for
$\mathbf{F}^{\mbox{-}1}(\absI{\msbf{A}})$
and $\biguplus$ with the omitted subscript 
 is the same as in the proof of
Lemma \ref{ee-emeasurable}. \\

$\begin{aligned}
& LHS(\absI{\msbf{A}}, \msbf{y}_1 \cdots \msbf{y}_r )
= 
(\kappa_e)_e^\bullet 
(\diamondsuit, (\msbf{y}_1, \ldots
  ,\msbf{y}_r) ) \\
& 
= 
(\kappa_e)^r
(\diamondsuit \cap (X_e)^{\bullet r}
, (\msbf{y}_1, \ldots
  ,\msbf{y}_r) ) \\
& 
= 
(\kappa_e^\bullet)^r
((F^r)^{\mbox{-}1} \bigl(\diamondsuit \cap (X_e)^{\bullet r} \bigr)
, (\vec{\msbf{y}}_1, \ldots
  ,\vec{\msbf{y}}_r) )
\quad \quad \mbox{with $\vec{\msbf{y}}_j = (y_{j1}, \ldots , y_{jn_j})$ for each $j=1, \ldots r$}
 \\
& \textstyle = (\kappa_e^\bullet)^r
   ( (\biguplus F^k)^{\mbox{-}1}  \bigl(\diamondsuit \cap (X_e)^{\bullet r} \bigr)
, (\vec{\msbf{y}}_1, \ldots
  ,\vec{\msbf{y}}_r) )  \quad \mbox{as $
(\biguplus F^k)^{\mbox{-}1}
\bigl( \sim  \cap (X_e)^{\bullet r} \bigr)
= 
(F^r)^{\mbox{-}1}
\bigl( \sim  \cap (X_e)^{\bullet r}
 \bigr)$ for any subset $\sim$} \\
&  \textstyle 
= (\kappa^{n_1} \times \cdots \times \kappa^{n_r})
(  (\biguplus F^k )^{\mbox{-}1}
 \bigl( \diamondsuit
\bigr) \cap (X^{\bullet n_1} \times \cdots \times X^{\bullet n_r} )
  ,  
 (\vec{\msbf{y}}_1, \ldots,
\vec{\msbf{y}}_r)  ) \\
& \textstyle = \kappa^{\sum_{i=1}^r n_i} 
( 
(\biguplus F^k )^{\mbox{-}1}
\bigl( \diamondsuit \bigr) \cap X^{\bullet (\sum_{i=1}^r n_i)}
  , (\vec{\msbf{y}}_1, \ldots,
\vec{\msbf{y}}_r) ) 
\end{aligned}$ 
\\

\noindent 
The both HSs coincide by (\ref{invFkFabs}) of
Lemma \ref{ee-emeasurable}. The coincidence presumes the associativity of 
of the cartesian product so that 
$(\vec{\msbf{y}}_1, \ldots,
\vec{\msbf{y}}_r) = 
(y_{11}, \ldots , y_{1n_1}, \ldots
  , y_{r1}, \ldots , y_{rn_r})$.
\end{proof}

%\smallskip

\begin{defn}[Monoidalness]\label{monness}{\em
In the definition, $*^n$ denotes $\overbrace{* \cdots *}^{n} \in I_e$
for $n \in \mathbb{N}$.
\begin{itemize}
 \item 
$\moni : \ms{I} \longrightarrow \mse{I}$ is defined;
%for any  $I^{(n)} \in \mathcal{I}_e$ with $n \geq 1$;
\begin{align*}
 \moni(I, *^n):= 
\begin{cases}
1 & \text{for $n \geq 1$} \\
0 & \text{for $n=0$}
\end{cases}
& &  \text{To be short, $\moni(I, *^n)= \min(n,1)$. }
\end{align*}

\item $\mon{\mc{X}}{\mc{Y}}: \mse{X} \otimes \mse{Y} \longrightarrow
( (X \times Y)_e, (\mc{X} \otimes \mc{Y})_e)
$ is defined for every rectangle $\msbf{A} \times \msbf{B}$
with $\msbf{A} \in \mc{X}_e$,
$\msbf{B} \in \mc{Y}_e$ and every $(x_1, y_1) \cdots (x_n,y_n)
\in (X \times Y)_e$ for any $n \in \mathbb{N}$.
\begin{align*}
\mon{\mc{X}}{\mc{Y}}( \msbf{A} \times \msbf{B}, (x_1, y_1) \cdots (x_n,y_n))
& := \Dd{x_1 \cdots x_n}{\msbf{A}}
      \, \Dd{y_1 \cdots y_n}{\msbf{B}}
\end{align*}
\end{itemize}
}\end{defn}
Note for the definition, the finite rectangles with same dimensions suffice as
the following holds
$$\mon{\mc{X}}{\mc{Y}}( \msbf{A} \times \msbf{B}, (x_1, y_1) \cdots
(x_n,y_n))
=\mon{\mc{X}}{\mc{Y}}( (\msbf{A} \times \msbf{B}) \cap (X^{(n)} \times
Y^{(n)}), (x_1, y_1) \cdots (x_n,y_n))$$

\smallskip

\begin{prop}
The dereliction ${\sf d}$ is  a monoidal natural transformation
with respect to the monoidalness ${\sf m}$ of Definition \ref{monness}.
\end{prop}

\begin{proof}
The two conditions (i) and (ii) are checked:

\noindent (i) The composition $\xymatrix{I \ar[r]^{\moni} & I_e
 \ar[r]^{\dii} & I }$ is the identity.
\begin{align*}
\textstyle \dii \Comp \moni (-, *)
 = \textstyle 
\Int{I_e}{\moni (-, \msbf{z})}{\dii (d \msbf{z}, *)}
= 
\Int{I_e}{\moni (-, \msbf{z})}{\Dd{*}{d \msbf{z} \cap I^{(1)}}}
 \textstyle
=
\Int{I_e \cap I^{(1)}}{\moni (-, z)}{\Dd {*}{d z}}
= \moni(-, *)
\end{align*}
% $\begin{aligned}
% \int_{I_e} \moni (*, d \msbf{z}) \dii (\msbf{z}, -)
% = 
% \int_{I_e} \delta (*, d \msbf{z}) \dii (\msbf{z}, -)
% = \dii (*,-) = \delta(*,-)
% \end{aligned}$

\smallskip

\noindent (ii) $\begin{aligned}
		 \di{\mathcal{X} \otimes 
\mathcal{Y}} \Comp \sf{m}_{\mathcal{X},\mathcal{Y}}= \di{\mathcal{X}}
 \otimes \di{\mathcal{Y}}\end{aligned}$ 
\begin{align*}
LHS ( \msbf{A} \times \msbf{B}, (x,y) )
& = \textstyle \Int{(X \times Y)_e}{
\mon{\mc{X}}{\mc{Y}} ( \msbf{A} \times \msbf{B}, \msbf{z} )}{
\di{\mc{X} \otimes \mc{Y}} (d \msbf{z}, (x,y))} \\
& \textstyle =  \Int{(X \times Y)_e}{
\mon{\mc{X}}{\mc{Y}} ( \msbf{A} \times \msbf{B}, \msbf{z} )}{
\Dd{(x,y)}{d \msbf{z} \cap (X \times Y)^{(1)}}}   \\
& \textstyle
 = \Int{(X \times Y)_e \cap (X \times Y)^{(1)}}{
\mon{\mc{X}}{\mc{Y}} ( \msbf{A} \times \msbf{B}, z )}{
\Dd{(x,y)}{d z}} \\
& = 
\mon{\mc{X}}{\mc{Y}} ( \msbf{A} \times \msbf{B}, (x,y))
= 
\Dd{x}{\msbf{A}} \, \Dd{y}{\msbf{B}} 
= 
\Dd{x}{\msbf{A} \cap X^{(1)}} \, \Dd{y}{\msbf{B} \cap Y^{(1)}}
\\
& \textstyle
= \di{\mc{X}}(\msbf{A}, x) \, \di{\mc{Y}}(\msbf{B}, y)
= RHS( \msbf{A} \times \msbf{B}, (x,y) )
\end{align*}
\end{proof}

\begin{prop}
$((~)_e,  \di{\mc{X}}, \stor{\mc{X}})$ is a comonad on $\opTKersfin$.
\end{prop}

\begin{proof} The two conditions
(i) and (ii) are checked. In the proof $\msbf{x}
= x_1 \cdots x_n \in X_e \cap X^{(n)}$ for any $n$
and $\msbf{A} \in 
\mc{X}_e$. 
\bigskip

\noindent (i) 
$\begin{aligned}
\di{\mathcal{X}_e} \Comp \stor{\mathcal{X}}
= \operatorname{Id}_{\mathcal{X}_e}
=  (\di{\mathcal{X}})_e \Comp \stor{\mathcal{X}}
\end{aligned}$
\begin{align*}
& \textstyle LHS (
\msbf{A}, \msbf{x}) = 
\Int{X_{ee}}{
\stor{\mathcal{X}}(\msbf{A}, \mathfrak{y})}{
\di{\mathcal{X}_e} (d \mathfrak{y}, \msbf{x})} =
\Int{X_{ee} \cap (X_e)^{(1)}}{
\stor{\mathcal{X}}(\msbf{A}, \mathfrak{y})}{
\Dd{\msbf{x}}{d \mathfrak{y}}}
= 
\stor{\mathcal{X}}(\msbf{A}, \msbf{x})
= 
\Dd{\msbf{x}}{\absI{\msbf{A}}} 
\end{align*}
\begin{align*}
& \textstyle RHS (\msbf{A}, \msbf{x}) = 
\Int{X_{ee}}{\stor{\mc{X}}(\msbf{A}, \mathfrak{y})}{ (\di{\mc{X}})_e 
(d \mathfrak{y}, \msbf{x})}
       = 
\Int{X_{ee}}{\Dd{\mathfrak{y}}{\absI{\msbf{A}}}}{(\di{\mc{X}})_e 
(d \mathfrak{y}, \msbf{x})}
= 
(\di{\mc{X}})_e 
(\absI{\msbf{A}}, \msbf{x}) \\
& = 
(\di{\mc{X}})^{\bullet n}
(F^{-1} (\absI{\msbf{A}}) \cap \,\, (X_e)^{\bullet n},  (x_1, \ldots ,  x_n))
= \delta^{\bullet n}
(F^{-1} (\absI{\msbf{A}}) \cap \,\, (X_e \cap X^{(1)})^{\bullet n}, 
(x_1, \ldots ,  x_n)) \\
& 
= 
\delta_e
(\absI{\msbf{A}},  x_1 \cdots x_n)
= 
\Dd{x_1 \cdots x_n}{\absI{\msbf{A}}}
\end{align*}
The both HS's coincide with 
$\operatorname{Id}_{\mathcal{X}_e}$ as 
$\Dd{x_1 \cdots x_n}{\absI{\msbf{A}}} =
\Dd{\abs{x_1 \cdots x_n}}{\msbf{A}}$
and $\abs{x_1 \cdots x_n} = x_1 \cdots x_n$.

\bigskip

\noindent (ii) $\begin{aligned}
\stor{\mathcal{X}_e} \Comp \stor{\mathcal{X}}
= 
(\stor{\mathcal{X}})_e \Comp \stor{\mathcal{X}}
\end{aligned}$
\begin{align*} 
LHS (\msbf{A}, -) & = \textstyle
\Int{X_{ee}}{
\stor{\mc{X}} (\msbf{A}, \mathfrak{y})}{
\stor{\mc{X}_e} (d \mathfrak{y}, -)}
= 
\Int{X_{ee}}{
\Dd{\mathfrak{y}}{\absI{\msbf{A}}}}{
\stor{\mc{X}_e} (d \mathfrak{y}, -)}
\\
 & 
= \stor{\mc{X}_e} (\absI{\msbf{A}}
, -) = 
\Dd{\abs{-}}{\absI{\msbf{A}}}
= 
\Dd{\abs{\abs{-}}}{\msbf{A}}
\end{align*}

For RHS, let $-$ be instantiated with
$\mathfrak{y}_1 \cdots \mathfrak{y}_n \in X_{eee}$
such that $\mathfrak{y}_i \in X_{ee}$ with
$i=1, \ldots ,n$.
\begin{align*}  \textstyle
& RHS (\msbf{A}, \mathfrak{y}_1 \cdots \mathfrak{y}_n) = \textstyle  
\Int{X_{ee}}{\stor{\mc{X}} (\msbf{A}, \mathfrak{z})}{
(\stor{\mathcal{X}})_e ( d \mathfrak{z}, 
\mathfrak{y}_1 \cdots \mathfrak{y}_n)}
= 
\Int{X_{ee}}{\Dd{\mathfrak{z}}{\absI{\msbf{A}}}}{ 
(\stor{\mathcal{X}})_e ( d \mathfrak{z}, 
\mathfrak{y}_1 \cdots \mathfrak{y}_n)} \\
&  \textstyle
= 
(\stor{\mathcal{X}})_e ( \absI{\msbf{A}}, 
\mathfrak{y}_1 \cdots \mathfrak{y}_n)
= 
(\stor{\mathcal{X}})^n ( \FI{(\absI{\msbf{A}})} \cap X^{\bullet n}, 
(\mathfrak{y}_1, \ldots , \mathfrak{y}_n))
=
\delta^n ( \FI{(\absI{\msbf{A}})} \cap X^{\bullet n}, 
(\abs{\mathfrak{y}_1}, \ldots, \abs{\mathfrak{y}_n})) \\ \textstyle
& = 
\delta^n ( \FI{(\absI{\msbf{A}})}, 
(\abs{\mathfrak{y}_1}, \ldots, \abs{\mathfrak{y}_n}))  = 
\delta_e ( \absI{\msbf{A}}, 
\abs{\mathfrak{y}_1} \cdots \abs{\mathfrak{y}_n})
= 
\Dd{\abs{\mathfrak{y}_1} \cdots \abs{\mathfrak{y}_n}}{\absI{\msbf{A}}}
=
\Dd{\abs{\abs{\mathfrak{y}_1} \cdots \abs{\mathfrak{y}_n}}}{\msbf{A}}
\end{align*}
% The last equation is by the functoriality of $(-)_e$ preserving the
% identity.

%\smallskip

The both HSs coincide because of the following equality in $X_e$:
$$ \abs{\abs{\mathfrak{y}_1 \cdots \mathfrak{y}_n}}
= \abs{\abs{\mathfrak{y}_1} \cdots \abs{\mathfrak{y}_n}}
$$
\end{proof}

%\smallskip

In terms of the monoidalness,
Proposition \ref{stor} is strengthened into 
\begin{prop}[Monoidality of ${\sf s}$]
The natural transformation storage ${\sf s}$ is monoidal.
That is 
\begin{align*}
 \stor{\mc{X} \otimes \mc{Y}} \Comp \mon{\mc{X}}{\mc{Y}}
= (\mon{\mc{X}}{\mc{Y}})_e \Comp \mon{\mc{X}_e}{\mc{Y}_e} \Comp
(\stor{\mc{X}} \otimes \stor{\mc{Y}})
\end{align*}
Note that the monoidality on the functor $(~)_{ee}$
is given by 
$(\mon{\mc{X}}{\mc{Y}})_e \Comp \mon{\mc{X}_e}{\mc{Y}_e}$.
\end{prop}
\begin{proof}
In the proof, it is sufficient to consider
an instantiation at
any rectangle $\msbf{A} \times \msbf{B} \in \mc{X}_e \otimes \mc{Y}_e$
such that
$\msbf{A} \in \mc{X}_e \cap X^{(n)}$ and 
$\msbf{B} \in \mc{Y}_e \cap Y^{(n')}$ for any $n, n' \geq 0$.

%\smallskip

For LHS, by virtue of the note on $\mon{\mc{X}}{\mc{Y}}$
below Definition \ref{monness},
we calculate the case $n=n'$, as the other case $n \not = n'$
directly makes LHS zero.
\begin{align*}
 LHS (\msbf{A} \times \msbf{B}, -) 
&
\textstyle
= \Int{(X \times Y)_e}{\mon{\mc{X}}{\mc{Y}}(\msbf{A} \times \msbf{B}, 
 \msbf{z}
  )}{\stor{\mc{X} \otimes \mc{Y}}(d \msbf{z}, - )} \\
& \textstyle
= \Int{(X \times Y)^{(n)}}{\mon{\mc{X}}{\mc{Y}}(\msbf{A} \times \msbf{B}, 
 (x_1, y_1) \cdots (x_n, y_n)
  )}{\stor{\mc{X} \otimes \mc{Y}}(d (x_1, y_1) \cdots (x_n,
  y_n), - )} \\
& \textstyle 
= \Int{(X \times Y)^{(n)}}{
\Dd{x_1 \cdots x_n}{\msbf{A}}
\, \Dd{y_1 \cdots y_n}{\msbf{B}}
}{\stor{\mc{X} \otimes \mc{Y}}( d (x_1, y_1) \cdots (x_n,
  y_n), - )}  \\
& \textstyle 
=
\stor{\mc{X} \otimes \mc{Y}}( [\msbf{A},\msbf{B} ]  , - ), 
\\
\tag*{in which $[\msbf{A},\msbf{B} ]
:=  \{ (x_1, y_{1}) \cdots (x_n, y_{n}) \mid
x_1 \cdots x_n \in \msbf{A} \quad 
y_1 \cdots y_n \in \msbf{B} \quad n \in \mathbb{N}  \}.$} 
\end{align*} 
Note $(x_1, y_{\sigma(1)}) \cdots (x_n, y_{\sigma(n)}) \in [\msbf{A},\msbf{B}
 ]$ for any $\sigma \in \mathfrak{S}_n$ whenever
$(x_1, y_{1}) \cdots (x_n, y_{n}) \in [\msbf{A},\msbf{B} ]$.
(Symmetrically 
$(x_{\sigma(1)}, y_1) \cdots (x_{\sigma(n)}, y_n) \in [\msbf{A},\msbf{B}
 ]$ under the same condition.) 

\smallskip
\noindent Let $-$ be instantiated with any element 
 $\msbf{c}_1 \cdots \msbf{c}_m
\in 
(\mc{X} \otimes \mc{Y})_{ee} \cap ((X \times Y)_e)^{(m)}$ for any $m$
so that each $\msbf{c}_i = (x_{i1}, y_{i1}) \cdots (x_{in_i}, y_{in_i}) \in
(X \times Y)_e$.
Then
\begin{align*}
 LHS (\msbf{A} \times \msbf{B},  \msbf{c}_1 \cdots \msbf{c}_m ) & =
\stor{\mc{X} \otimes \mc{Y}}( [\msbf{A},\msbf{B} ], 
\msbf{c}_1 \cdots \msbf{c}_m  )
= \Dd{\abs{\msbf{c}_1 \cdots \msbf{c}_m}}{[\msbf{A},\msbf{B}]} \\
& =  \Dd{x_{11} \cdots x_{1n_1} \cdots x_{m1} \cdots x_{mn_m}
}{\msbf{A}} \, \Dd{
y_{11} \cdots y_{1n_1} \cdots y_{m1} \cdots y_{mn_m}}{\msbf{B}}
\end{align*}
\noindent For (RHS), first observe, \\
$\begin{aligned}
& \textstyle \mon{\mc{X}_e}{\mc{Y}_e} \Comp (\stor{\mc{X}} \otimes \stor{\mc{Y}})
(\msbf{A} \times \msbf{B}, -)
 = \int_{X_{ee}} \Int{Y_{ee}}
{
(\stor{\mc{X}} \otimes \stor{\mc{Y}})(\msbf{A} \times \msbf{B},
  (\mathfrak{y}_1,
\mathfrak{y}_2)) 
}{\mon{\mc{X}_e}{\mc{Y}_e}(d \mathfrak{y}_1 \times d \mathfrak{y}_2, -)}
\\ & 
\textstyle 
= 
\int_{X_{ee}}\Int{Y_{ee}}
{
\Dd{\mathfrak{y}_1}{\absI{\msbf{A}}}
\, \Dd{\mathfrak{y}_2}{\absI{\msbf{B}}}
}{\mon{\mc{X}_e}{\mc{Y}_e}(d \mathfrak{y}_1 \times d \mathfrak{y}_2, -)}
= \mon{\mc{X}_e}{\mc{Y}_e}(\absI{\msbf{A}} \times \absI{\msbf{B}} , -)
\end{aligned}$  \\

\smallskip

\noindent Using the observation, the following is calculated in which
$-$ denotes an arbitrary instantiation $\msbf{c}_1 \cdots \msbf{c}_m \in 
(X \times Y)_{ee} \cap ((X \times Y)_{e})^{(m)}$  with any $m$ and 
$\msbf{c}_i \in (X \times Y)_e$.
\begin{align*}
& RHS(\msbf{A} \times \msbf{B}, -) \\
& \textstyle
= 
\Int{(X_e \times Y_e)_e}
{
\mon{\mc{X}}{\mc{Y}} \Comp (\stor{\mc{X}} \otimes \stor{\mc{Y}})
(\msbf{A} \times \msbf{B}, 
\msbf{z})}{
(\mon{\mc{X}}{\mc{Y}})_e (d  \msbf{z}, -)} \\ 
& \textstyle 
= 
\sum\limits_{k=0}^{\infty}
\Int{(X_e \times Y_e)^{(k)}}
{
\mon{\mc{X}}{\mc{Y}} \Comp (\stor{\mc{X}} \otimes \stor{\mc{Y}})
(\msbf{A} \times \msbf{B}, (\msbf{x}_1, \msbf{y}_1) \cdots
(\msbf{x}_k, \msbf{y}_k))}{
(\mon{\mc{X}}{\mc{Y}})_e (d  (\msbf{x}_1, \msbf{y}_1) \cdots
(\msbf{x}_k, \msbf{y}_k) , -)} \\ 
& \tag*{by (\ref{edisuni}) and commuting integral over countable sum} \\ 
& \textstyle 
= 
\sum\limits_{k=0}^{\infty}
\Int{(X_e \times Y_e)^{(k)}}
{
\mon{\mc{X}_e}{\mc{Y}_e}(\absI{\msbf{A}} \times \absI{\msbf{B}} ,
(\msbf{x}_1, \msbf{y}_1) \cdots
(\msbf{x}_k, \msbf{y}_k)
)}{
(\mon{\mc{X}}{\mc{Y}})_e (d  (\msbf{x}_1, \msbf{y}_1) \cdots
(\msbf{x}_k, \msbf{y}_k) , -)}  \\
& \textstyle = 
\sum\limits_{k=0}^{\infty}
\Int{(X_e \times Y_e)^{(k)}}
{
\Dd{\msbf{x}_1  \cdots \msbf{x}_k }{\absI{\msbf{A}}}
\, 
\Dd{\msbf{y}_1  \cdots \msbf{y}_k }{\absI{\msbf{B}}}
}{
(\mon{\mc{X}}{\mc{Y}})_e (d  (\msbf{x}_1, \msbf{y}_1) \cdots
(\msbf{x}_k, \msbf{y}_k) , -)}  \\
& \textstyle
= \Int{(X_e \times Y_e)^{(m)}}
{\Dd{\msbf{x}_1  \cdots \msbf{x}_m}{\absI{\msbf{A}}}
 \, \Dd{\msbf{y}_1  \cdots \msbf{y}_m }{\absI{\msbf{B}}}}{
(\mon{\mc{X}}{\mc{Y}})_e (d  (\msbf{x}_1, \msbf{y}_1) \cdots
(\msbf{x}_k, \msbf{y}_m) , \msbf{c}_1 \cdots \msbf{c}_m)}  \\
&    
\tag*{as the sum solely contributes when $k=m$ (i.e., zero if $k \not =m$)
by (\ref{expebot})} \\
& %\textstyle
= \int_{ (X_e \times Y_e)^{(m)}}
\left[ \begin{aligned}
 \Dd{\msbf{x}_1  \cdots \msbf{x}_m}{\absI{\msbf{A}}}
 \, \Dd{\msbf{y}_1  \cdots \msbf{y}_m}{\absI{\msbf{B}}} 
\quad \quad \quad \quad \quad \quad \quad \quad      \\
(\mon{\mc{X}}{\mc{Y}})_e^\bullet (
\FI(d  (\msbf{x}_1, \msbf{y}_1) \cdots
(\msbf{x}_m, \msbf{y}_m)),
(\msbf{c}_1, \ldots \msbf{c}_m))
\end{aligned} \right]
\tag*{ by the def of $(\mon{\mc{X}}{\mc{Y}})_e$}
\\
& %\textstyle
= \int_{ (X_e \times Y_e)^{\bullet m}}
\left[ \begin{aligned}
\Dd{(\msbf{x}_1, \ldots , \msbf{x}_m) }{\FI(\absI{\msbf{A}})}
 \, \Dd{(\msbf{y}_1,  \ldots , \msbf{y}_m) }{\FI(\absI{\msbf{B}})} 
\quad \quad \quad \quad \quad 
 \\
(\mon{\mc{X}}{\mc{Y}})_e^\bullet (d   (\msbf{x}_1, \msbf{y}_1) \times \cdots 
\times
d (\msbf{x}_m, \msbf{y}_m) , (\msbf{c}_1, \ldots \msbf{c}_m))
\end{aligned} \right]
\\ & 
\tag*{by the variable change (\ref{IWPFM}) along $F:
(X_e \times Y_e)_e^\bullet %\cap (X_e \times Y_e)^{\bullet m}
 \rightarrow 
(X_e \times Y_e)_e 
%\cap (X_e \times Y_e)^{(m)}$
$}
\\
& \textstyle
= \Int{ (X_e \times Y_e)^{\bullet m}}
{\Dd{(\msbf{x}_1, \ldots , \msbf{x}_m) }{\FI(\absI{\msbf{A}})}
 \, \Dd{(\msbf{y}_1,  \ldots , \msbf{y}_m) }{\FI(\absI{\msbf{B}})}}{
\prod\limits_{i=1}^m \mon{\mc{X}}{\mc{Y}} (d  (\msbf{x}_i, \msbf{y}_i), 
\msbf{c}_i )} \\
& \tag*{by the def of $(\mon{\mc{X}}{\mc{Y}})_e^\bullet$ using the product measure} \\
& %\textstyle
= \int_{(X_e \times Y_e)^{\bullet m}}
\left[ \begin{aligned}
 \Dd{(\msbf{x}_1, \ldots , \msbf{x}_m) }{\FI(\absI{\msbf{A}})}
 \, \Dd{(\msbf{y}_1,  \ldots , \msbf{y}_m) }{\FI(\absI{\msbf{B}})}  \quad \quad \quad \\
\textstyle \prod\limits_{i=1}^m 
\Dd{a_{i1} \cdots a_{in_i}}{d \msbf{x}_i} \, 
\, \Dd{b_{i1} \cdots b_{in_i}}{d \msbf{y}_i}
\end{aligned} \right]
\\ & 
\tag*{by putting explicitly  
$\msbf{c}_i= (x_{i1}, y_{i1}) \cdots (x_{in_i}, y_{in_i}) \in (X \times Y)_e$
with $i=1, \ldots ,m$
}
\\ \\
& \textstyle 
= \Dd{((x_{11} \cdots x_{1n_1}), \ldots , (x_{m1} \cdots x_{mn_m}))
}{\FI(\absI{\msbf{A}})}
\, \Dd{((y_{11} \cdots y_{1n_1}), \ldots , (y_{m1} \cdots y_{mn_m}))
}{\FI(\absI{\msbf{B}})} \\
%\end{align*}
%\begin{align*}
& \textstyle 
= 
\Dd{ \msbf{x}_{1} \cdots  \msbf{x}_{m}
}{\absI{\msbf{A}}} \, 
\Dd{ \msbf{y}_1 \cdots \msbf{y}_m
}{\absI{\msbf{B}}} \tag*{with 
$\msbf{x}_i = x_{i1} \cdots x_{in_i}$ and
$\msbf{y}_i = y_{i1} \cdots y_{in_i}$ ($1 \leq i \leq m$).}  \\
& \textstyle 
= 
\Dd{ \abs{\msbf{x}_{1} \cdots  \msbf{x}_{m}}
}{\msbf{A}}
\, \Dd{ \abs{\msbf{y}_{1} \cdots  \msbf{y}_{m}}
}{\msbf{B}}
\\ & \textstyle 
= 
\Dd{ x_{11} \cdots x_{1n_1} \cdots  x_{m1} \cdots x_{mn_m}
}{\msbf{A}}
\, \Dd{ y_{11} \cdots y_{1n_1} \cdots  y_{m1} \cdots y_{mn_m}
}{\msbf{B}}
 \end{align*}
Both HSs coincide.
% \smallskip
% RHS coincides with LHS on any instance $\msbf{A} \times \msbf{B}$
% when $n$ and $n'$ are given both equal to $\textstyle \sum_{i=1}^m n_i$,
% while the other instance when $n \not = n'$ directly makes RHS zero.
% Thus both HSs coincide. 
\end{proof}

\smallskip

\begin{prop}[weakening and contraction]~\\{\em
Monoidal natural transformations $\wk{\mc{X}}$ and $\con{\mc{X}}$
are defined: 

\begin{itemize}
 \item[-](Weakening)
$\wk{\mathcal{X}} : \mse{X} \longrightarrow \ms{I}$ is defined
for $\msbf{A} \in \mathcal{X}_e$:
$$\begin{aligned}
%%\wk{\mathcal{X}}(\msbf{x}, I) =1 & \quad \text{and} & 
%%\wk{\mathcal{X}}(\msbf{x}, \emptyset) =0  \\
%\wk{\mathcal{X}}(\msbf{x}, -) =0
\wk{\mathcal{X}}(\msbf{A}, *) := \Dd{0}{\msbf{A}}, \quad
\mbox{where $0$ is the monoid identity in $X_e$.}
\end{aligned}$$

\item[-] (Contraction) 
$\con{\mathcal{X}} : \mse{X} \longrightarrow \mse{X} \otimes \mse{X}$
is defined for $\msbf{x}^1, \msbf{x}^2 \in X_e$ and $\msbf{A} \in \mc{X}_e$
\begin{align*}
\con{\mathcal{X}}(\msbf{A}, (\msbf{x}^1, \msbf{x}^2) ) :=
\Dd{\abs{\msbf{x}^1 \msbf{x}^2}}{\msbf{A}}
%\delta ((\msbf{x},\msbf{x}), \msbf{A} \otimes \msbf{B}) 
\end{align*}
Note that $\abs{\msbf{x}^1 \msbf{x}^2}$ is the image of
$(\msbf{x}^1, \msbf{x}^2)$ by the following composition: \\
$X_e \times X_e \longrightarrow X_{ee} \longrightarrow X_e 
\hspace{1cm}
(\msbf{x}^1, \msbf{x}^2) \longmapsto \msbf{x}^1 \msbf{x}^2 \longmapsto
\abs{\msbf{x}^1 \msbf{x}^2}$  
\end{itemize}
}
\smallskip

Then $(\mc{X}_e, \con{\mc{X}},\wk{\mc{X}})$ forms a commutative comonoid.
Moreover  $\wk{\mc{X}}$ is a coalgebra morphism from
$(\mc{X}_e, \stor{\mc{X}})$ to $(\mc{I},\moni )$
and $\con{\mc{X}}$ is a coalgebra morphism from
$(\mc{X}_e, \stor{\mc{X}})$
to 
$(\mc{X}_e \otimes \mc{X}_e, \mon{\mc{X}_e}{\mc{X}_e} \Comp
 (\stor{\mc{X}} \otimes  \stor{\mc{X}}) )$.
\end{prop}

\begin{proof}
The commutative comonoid conditions are the following (a), (b) and
 (c): \\
\noindent (a) 
${\sf sy}_{\mc{X}_e,\mc{X}_e} \Comp
\con{\mc{X}}= 
\con{\mc{X}}$, where ${\sf sy}$ is the symmetry of
monoidal product. 
This is by \\ $\con{\mc{X}}( \msbf{A}, (\msbf{x}^1, \msbf{x}^2))
= 
\con{\mc{X}}( \msbf{A}, (\msbf{x}^2, \msbf{x}^1))$
for $\msbf{x}^1, \msbf{x}^2 \in X_e$.

\noindent (b)
${\sf ac}_{\mc{X}_e,\mc{X}_e,\mc{X}_e} \Comp
(\con{\mc{X}_e} \otimes \operatorname{Id}_{\mc{X}_e})
\Comp \con{\mc{X}_e}
= 
(\operatorname{Id}_{\mc{X}_e} \otimes \con{\mathcal{X}})
\Comp \con{\mathcal{X}}$,
 where ${\sf ac}$ is an associativity of $\otimes$. 
By Fubini-Tonelli, the condition amounts to the equality
$\abs{\abs{\msbf{x}^1 \msbf{x}^2} \msbf{x}^3}
= \abs{\msbf{x}^1 \abs{\msbf{x}^2 \msbf{x}^3} }$ in $X_e$
for $\msbf{x}^1, \msbf{x}^2, \msbf{x}^3 \in X_e$.

\smallskip

\noindent
(c)
$(\wk{\mc{X}} \otimes \operatorname{Id}_{\mc{X}_e}) \Comp \con{\mc{X}}$
coincides with the canonical morphism $\mse{X} \longrightarrow  \ms{I} \otimes
\mse{X}$ for the monoidal unit. The condition is checked as follows: 
\begin{align*}
& \textstyle  ((\wk{\mc{X}} \otimes \operatorname{Id}_{\mc{X}_e}) \Comp \con{\mc{X}})
(\msbf{A}, (*, \msbf{x}) ) 
 = \int_{X_e} \Int{X_e}{
\con{\mc{X}}(\msbf{A}, (\msbf{y}_1, \msbf{y}_2))}{
\wk{\mc{X}} (d \msbf{y}_1, *) \Dd{\msbf{x}}{d \msbf{y}_2}}
\\
& \textstyle =
\Int{X_e}{
\con{\mc{X}}(\msbf{A}, (\msbf{y}_1, \msbf{x}))}{
\wk{\mc{X}} (d \msbf{y}_1, *)}
= 
\Int{ X_e  \cap X^{(0)}}{
\con{\mc{X}}(\msbf{A}, (y_1, \msbf{x}))}{
\wk{\mc{X}} (d y_1, *)}
\\
& \textstyle =
\con{\mc{X}}(\msbf{A}, (0, \msbf{x}))
\, \Dd{0}{\{ 0\}}
=
\con{\mc{X}}(\msbf{A}, (0, \msbf{x}))
=
\Dd{\abs{0 \msbf{x}}}{\msbf{A}} 
= \Dd{\msbf{x}}{\msbf{A}}  
\end{align*}
The last equation holds because $\abs{0 \msbf{x}}= \msbf{x}$.

\bigskip

The conditions for the coalgebra morphisms are
the following (i) and (ii) respectively for the weakening and for the
 contraction:

\smallskip
 
\noindent (i) $\moni \Comp \wk{\mc{X}} = (\wk{\mc{X}})_e \Comp
 \stor{\mc{X}}$
\begin{align*}
\textstyle LHS(\msbf{A}, *^n)
& = \textstyle 
\Int{I}{\wk{\mc{X}}(\msbf{A}, x)}{\moni(dx, *^n)}
= 
\wk{\mc{X}}(\msbf{A}, *) \,  \moni( \{ * \}, *^n )
= 
\Dd{0}{\msbf{A}} \times \min (n, 1)
%\end{align*}
%\begin{align*}
\\
\textstyle RHS(\msbf{A}, *^n)
& \textstyle = \Int{X_{ee}}{\stor{\mc{X}}(\msbf{A}, \mathfrak{y})}{
(\wk{\mc{X}})_e (d \mathfrak{y}, *^n)}
= \Int{X_{ee}}{
\Dd{\mathfrak{y}}{\absI{\msbf{A}}}}{
(\wk{\mc{X}})_e (d \mathfrak{y}, *^n)} \\
 & \textstyle  
= (\wk{\mc{X}})_e ( \absI{\msbf{A}},   *^n)
= \begin{cases}
\delta_e (\absI{\msbf{A}}, 0^n)
= \Dd{0}{\absI{\msbf{A}}}
=  \Dd{\abs{0}}{\msbf{A}}
= \Dd{0}{\msbf{A}}
 & \text{if $n \geq 1$}  %\hspace{3ex}  \text{Note:
%$\{ 0 \} \cdots \{ 0 \}$ denotes $\underbrace{\{ 0 \} \cdots \{ 0
   %\}}_n$.} 
\\
0  & \text{if $n=0$}
  \end{cases}
\end{align*}
The first case uses $\abs{0}=0$.
The second one uses $*^0=0$ and that 
$\kappa^0: \mc{X}^{\bullet 0} \rightarrow \mc{Y}^{\bullet 0}$
is zero for any kernel $\kappa: \mc{X} \longrightarrow \mc{Y}$ as 
$\mc{Z}^{\bullet 0}$ is the only $\sigma$-field over the $\emptyset$
for any $\mc{Z}$.

\bigskip

\noindent (ii) $(\con{\mc{X}})_e \Comp \stor{\mc{X}} =
\mon{\mc{X}_e}{\mc{X}_e} \Comp (\stor{\mc{X}} \otimes  \stor{\mc{X}})
\Comp \con{\mc{X}}$ 

\bigskip

\noindent For RHS, first we calculate: \\
$\begin{aligned}
& \textstyle
(\stor{\mc{X}} \otimes \stor{\mc{X}}) \Comp \con{\mc{X}}
(\msbf{A}, (\mathfrak{y}_1, \mathfrak{y}_2))
= \Int{X_e \times X_e}{\con{\mc{X}}(\msbf{A}, (\msbf{y}, \msbf{z}))
}{\stor{\mc{X}} \otimes \stor{\mc{X}} (d (\msbf{y}, \msbf{z}),
 (\mathfrak{y_1, \mathfrak{y}_2}))} 
\\
& \textstyle 
\stackrel{FT}{=} 
\int_{X_e} \Int{X_e}{\con{\mc{X}}(\msbf{A}, (\msbf{y}, \msbf{z}))
}{
\Dd{\abs{\mathfrak{y}_1}}{d \msbf{y}}
\Dd{\abs{\mathfrak{y}_2}}{d \msbf{z}}
= \con{\mc{X}}(\msbf{A}, (\abs{\mathfrak{y}_1}, \abs{\mathfrak{y}_2}))
= 
\Dd{\abs{\abs{\mathfrak{y}_1} \abs{\mathfrak{y}_2}}}{\msbf{A}}
}
\end{aligned}$ 
\smallskip

Thus, using this at the following final line,  
\begin{align*}
& \textstyle RHS (\msbf{A},
(\msbf{x}_1, \msbf{y}_1) \cdots (\msbf{x}_n, \msbf{y}_m)) \\
& \textstyle
= 
\int_{X_{ee}} \Int{X_{ee}}{
(\stor{\mc{X}} \otimes \stor{\mc{X}}) \Comp \con{\mc{X}} (\msbf{A},
  (\mathfrak{y}_1, \mathfrak{y}_2))}{\mon{\mc{X}_e}{\mc{Y}_e}
(d \mathfrak{y}_1 \times d \mathfrak{y}_2, 
(\msbf{x}_1, \msbf{y}_1) \cdots (\msbf{x}_n, \msbf{y}_n))
} \\
& \textstyle = 
\int_{X_{ee}}\Int{X_{ee}}{
(\stor{\mc{X}} \otimes \stor{\mc{X}}) \Comp \con{\mc{X} } (\msbf{A},
  (\mathfrak{y}_1, \mathfrak{y}_2))}{
\Dd{\msbf{x}_1  \cdots \msbf{x}_n}{d \mathfrak{y}_1}
\, \, \Dd{\msbf{y}_1  \cdots \msbf{y}_n}{d \mathfrak{y}_2}}
\\ & 
\textstyle
= \Dd{\abs{\abs{\, \msbf{x}_1  \cdots \msbf{x}_n} \, \abs{\msbf{y}_1
  \cdots \msbf{y}_n}} \,}{\msbf{A}}
\end{align*}

\smallskip

On the other hand,
\begin{align*}
& \textstyle
LHS (\msbf{A},
(\msbf{x}_1, \msbf{y}_1) \cdots (\msbf{x}_n, \msbf{y}_n)) 
= \Int{X_{ee}}{\stor{\mc{X}}(\msbf{A}, \mathfrak{y})}
{(\con{\mc{X}})_e  (d \mathfrak{y}, \,
(\msbf{x}_1, \msbf{y}_1) \cdots (\msbf{x}_n, \msbf{y}_n))} \\
& \textstyle 
= \Int{X_{ee}}{
\Dd{\mathfrak{y}}{\absI{\msbf{A}}}}
{(\con{\mc{X}})_e (d \mathfrak{y}, \,
(\msbf{x}_1, \msbf{y}_1) \cdots (\msbf{x}_n, \msbf{y}_n))}
= 
(\con{\mc{X}})_e (\absI{\msbf{A}} , \,
(\msbf{x}_1, \msbf{y}_1) \cdots (\msbf{x}_n, \msbf{y}_n)) \\
& \textstyle
= (\con{\mc{X}})_e^\bullet (\FI(\absI{\msbf{A}}) , \,
((\msbf{x}_1, \msbf{y}_1), \ldots , (\msbf{x}_n, \msbf{y}_n)))
= 
(\con{\mc{X}})^n (\FI(\absI{\msbf{A}}) \cap X^{\bullet n} , \,
((\msbf{x}_1, \msbf{y}_1), \ldots , (\msbf{x}_n, \msbf{y}_n)))
\\ & \textstyle
=
\delta^n (\FI(\absI{\msbf{A}}) \cap X^{\bullet n} , \,
( \abs{\msbf{x}_1 \msbf{y}_1}, \ldots , \abs{\msbf{x}_n \msbf{y}_n}))
= 
\delta^\bullet_e (\FI(\absI{\msbf{A}}) , \,
( \abs{\msbf{x}_1 \msbf{y}_1}, \ldots , \abs{\msbf{x}_n \msbf{y}_n})) 
\\ & \textstyle
= \delta_e (\absI{\msbf{A}} , \,
\abs{\msbf{x}_1 \msbf{y}_1}  \cdots  \abs{\msbf{x}_n \msbf{y}_n} )
= \delta (\absI{\msbf{A}} , \,
\abs{\msbf{x}_1 \msbf{y}_1} \cdots  \abs{\msbf{x}_n \msbf{y}_n} )
= \delta (\msbf{A} , \,
\abs{\abs{\msbf{x}_1 \msbf{y}_1} \cdots  \abs{\msbf{x}_n \msbf{y}_n}} )
\end{align*}

The both HSs coincide because of the following equality in $X_e$.
\begin{eqnarray*}
\abs{\abs{\, \msbf{x}_1  \cdots \msbf{x}_n} \, \abs{\msbf{y}_1
  \cdots \msbf{y}_n}}
=
\abs{\abs{\msbf{x}_1 \msbf{y}_1} \cdots  \abs{\msbf{x}_n \msbf{y}_n}}
\end{eqnarray*}
\end{proof}

\smallskip

We end this section with Theorem \ref{linexcomo} summarising this section
after the lemma below:

\smallskip

\begin{lem}[comonoidality of $\stor{}$] \label{monos}
$\stor{\mc{X}}$ is a comonoid morphism from
$(\mc{X}_e, \con{\mc{X}}, \wk{\mc{X}})$
to $(\mc{X}_{ee}, \con{\mc{X}_e}, \wk{\mc{X}_e})$.
\end{lem}
\begin{proof}
The two conditions need to be checked:
\begin{itemize}
 \item
$\begin{aligned}
(\stor{\mc{X}} \otimes \stor{\mc{X}} ) \Comp \con{\mc{X}} =
\con{\mc{X}_e} \Comp \stor{\mc{X}}
\end{aligned}$

$
%(\con{\mc{X}_e} \Comp \stor{\mc{X}})
LHS (\msbf{A}, (\mathfrak{y}_1, \mathfrak{y}_2))
=
\int_{\mc{X}_{ee}}
\stor{\mc{X}} (\msbf{A}, \tau)
\, \con{\mc{X}_e} (d \tau, (\mathfrak{y}_1, \mathfrak{y}_2))
=
\int_{\mc{X}_{ee}}
\stor{\mc{X}} (\msbf{A}, \tau)
\, \delta (d \tau, \abs{\mathfrak{y}_1 \mathfrak{y}_2}) \\
=
\stor{\mc{X}} (\msbf{A}, \abs{\mathfrak{y}_1 \mathfrak{y}_2})
= 
\delta (\msbf{A}, \abs{\abs{\mathfrak{y}_1 \mathfrak{y}_2}}).
$

$
%((\stor{\mc{X}} \otimes \stor{\mc{X}} ) \Comp \con{\mc{X}})
RHS (\msbf{A},
(\mathfrak{y}_1, \mathfrak{y}_2)) =
\int_{X_e \times X_e} \con{\mc{X}} (\msbf{A}, (\msbf{x}_1, \msbf{x}_2))
\, 
(\stor{\mc{X}} \otimes \stor{\mc{X}}) (d (\msbf{x}_1, \msbf{x}_2),
(\mathfrak{y}_1, \mathfrak{y}_2))
\\
= \int_{X_e \times X_e} \con{\mc{X}} (\msbf{A}, (\msbf{x}_1, \msbf{x}_2))
\, 
\stor{\mc{X}} (d \msbf{x}_1, \mathfrak{y}_1)
\, \stor{\mc{X}} (d \msbf{x}_2, \mathfrak{y}_2)
\\ = 
\int_{X_e \times X_e} \con{\mc{X}} (\msbf{A}, (\msbf{x}_1, \msbf{x}_2))
\, 
\delta (d \msbf{x}_1, \abs{\mathfrak{y}_1})
\, \delta (d \msbf{x}_2, \abs{\mathfrak{y}_2})
=
\con{\mc{X}} (\msbf{A}, (\abs{\mathfrak{y}_1}, \abs{\mathfrak{y}_2}))
=
\delta (\msbf{A}, \abs{\abs{\mathfrak{y}_1}\abs{\mathfrak{y}_2}}).
$

The both HSs coincide because
$\abs{\abs{\mathfrak{y}_1 \mathfrak{y}_2}}=
\abs{\abs{\mathfrak{y}_1}\abs{\mathfrak{y}_2}}$.

\item $\begin{aligned}
\wk{\mc{X}_e} \Comp  \stor{\mc{X}} 
= \wk{\mc{X}} 
\end{aligned}$

$
%\wk{\mc{X}_e} \Comp \, \stor{\mc{X}} 
LHS ( \msbf{A}, *)
= 
\int_{X_{ee}} \! \! \stor{\mc{X}} (\msbf{A}, \mathfrak{y})
\,  \wk{\mc{X}_e} (d \mathfrak{y}, *)
= 
\int_{X_{ee}} \!  \! \stor{\mc{X}} (\msbf{A}, \mathfrak{y})
\,  \delta (d \mathfrak{y}, 0 )
= 
\stor{\mc{X}} (\msbf{A}, 0 )
= \delta (\msbf{A}, \abs{ 0 })$, which equates to RHS as 
$\abs{ 0 } = 0 $.

\end{itemize}
\end{proof}

\smallskip

\begin{thm} \label{linexcomo}
$((~)_e, \stor{\mc{X}}, \di{\mc{X}},
\mon{\mc{X}}{\mc{Y}}, \moni)$ equipped with $\con{\mc{X}}$ and $\wk{\mc{X}}$
is a linear exponential comonad in $\opTKersfin$.
\end{thm}

\section{Double Glueing and Orthogonality over $\opTKersfin$} \label{dgort} 
This section constructs the double glueing over $\opTKersfin$
in accordance with Hyland-Schalk's general categorical framework
\cite{HSha} for constructing the structure of linear logic,
but without the assumption of
any closed structure of the base category.
In Section \ref{dg}, a crude but non degenerate opposite pair
is obtained between
product and coproduct as well as between tensor and cotensor, lifting
those but collapsed in the monoidal category $\opTKersfin$.
Furthermore an exponential comonad is constructed
for the glueing $\GopTKersfin$ over $\opTKersfin$. 
In Section \ref{ort},
a new instance of Hyland-Schalk orthogonality is given 
in terms of Lebesgue integral between measures and measurable functions,
owing to the measure theoretic study in the preceding
sections.
The instance in $\TKersfin$ has an adjunction property,
called reciprocal, in terms of an inner product using the integral.
The reciprocal orthogonality enables us to retain the exponential comonad
to the slack subcategory $\Sla{\opTKersfin}$.
Following the framework \cite{HSha},
the double gluing considered in this paper 
is along hom-functors to the category of sets.

\subsection{Double Glueing $\GopTKersfin$
with Exponential Comonad} \label{dg}

\begin{defn}[The category $\GopTKer$]~ \label{CatGop}\\{\em 
\noindent An object is a tuple $(\mc{X},U,R)$ such that 
$\mc{X}$ is an object of $\opTKer$, and 
$U$ and $R$ are sets $U \subseteq \opTKer(\mathcal{I},\mathcal{X})$
and $R \subseteq \opTKer(\mathcal{X}, \mathcal{I})$.
That is, $U$ and $R$ comprise specific classes
of measurable functions and of measures respectively.

\noindent Each map from $\bm{\mc{X}}=(\mathcal{X},U,R)$
to $\bm{\mc{Y}}=(\mathcal{Y},V,S)$  is any $\opTKer$
map $\kappa: \mathcal{X} \longrightarrow \mathcal{Y}$ satisfying: \\
\noindent-
$\forall g  : \mathcal{I} \longrightarrow \mathcal{X}$ in $U$,
the composition \xymatrix{ \kappa^* g : \mathcal{I} \ar[r]^{g} & \mathcal{X} \ar[r]^\kappa & \mathcal{Y}  
} belongs to $V$. \\
\noindent -
$\forall \mu : \mathcal{Y} \longrightarrow \mathcal{I}$ in $S$,
the composition \xymatrix{ \kappa_* \mu  : \mathcal{X} \ar[r]^{\kappa} & 
\mathcal{Y} \ar[r]^\mu & \mathcal{I}  
} belongs to $R$.
}\end{defn}
The forgetful functor exists $\GopTKer \longrightarrow
\opTKer$ forgetting the second and the third components of the objects.

\bigskip

The double glueing category $\GopTKersfin$ is defined the same 
over $\opTKersfin$ and becomes a subcategory of $\GopTKer$.
The general result of Hyland-Schalk \cite{HSha}
applies to the subcategory.
\begin{prop} \label{catGTKer}
$\GopTKersfin$
is a monoidal category with product and coproduct,
which is collapsed to the corresponding structures of
$\opTKersfin$
by the forgetful functor.
\end{prop}
Given objects
$\bm{\mc{X}}=(\mathcal{X},U,R)$
and $\bm{\mc{Y}}=(\mathcal{Y},V,S)$ of $\GopTKersfin$, \\
\noindent Tensor product \\
\vspace{-3ex}
$$\begin{aligned}
& \bm{\mc{X}} \otimes \bm{\mc{Y}} = 
(\mc{X} \otimes \mc{Y}, U \otimes V, T),
\text{where} \\
& U \otimes V = \{ 
 f \otimes
g : \mc{I} \cong \mc{I} \otimes \mc{I}  \longrightarrow
 \mc{X} \otimes \mc{Y}  \mid f \in U \, \,  g \in V \} 
\\ & 
%T&%= \GTKersfin(\mc{X},\mc{Y}^\perp) \\ & 
T=
\{ \nu: \mc{X} \otimes \mc{Y} \longrightarrow \mc{I} \mid
\xymatrix{\mc{I} \ar[r]^(.4){\forall f \in U} & 
\mc{X}},  \, \,  
\nu^* (f \otimes \delta_{\mc{Y}}) \in S 
\quad \text{and} \quad 
\xymatrix{\mc{I} \ar[r]^(.4){\forall g \in V} & \mc{Y}}, \, \, 
\nu^* (\delta_{\mc{X}} \otimes g) \in R
\}
\end{aligned}$$

\vspace{-2ex}
\noindent Note $\nu^* (f \otimes \delta_{\mc{Y}}):$ 
\xymatrix{ 
 \mc{Y} \cong \mc{I} \otimes
 \mc{Y} \ar[r]^{f \otimes \delta_{\mc{Y}}} & \mc{X} \otimes \mc{Y}
 \ar[r]^\nu  &  \mc{I}} and 
$\nu^* (\delta_{\mc{X}} \otimes g):$ 
\xymatrix{ 
\mc{X} \cong
\mc{X} \otimes \mc{I} \ar[r]^(.55){\delta_{\mc{X}} \otimes g} 
&   \mc{X} \otimes \mc{Y}  \ar[r]^\nu  &  \mc{I} }.

\smallskip

The tensor unit is given  $\mc{I}=(\mc{I}, \{\operatorname{Id}_{\mc{I}}\},
\opTKersfin(\mc{I},\mc{I}))
%= 
%(\mc{I}, \{\delta_{\mc{I}}\},   \{ r \, \delta_{\mc{I}} \mid r \in \mathbb{R}_+
%\})
$.

\smallskip

\noindent For a subset $U$ of a homset and a morphism $f$
of appropriate type, $U \Comp f$ and $f \Comp U$ denote the respective
subsets composed and precomposed with $f$ element-wisely to $U$.

\noindent Product 
\vspace{-.7ex}
$$\begin{aligned}
& \bm{\mc{X}} \& \bm{\mc{Y}} = 
(\mc{X} \amalg \mc{Y}, U \& V, (R \Comp \pr_{\mc{X}} ) \cup (S \Comp 
\pr_{\mc{X}}) ), \text{where}  
\\
& \quad   U \& V : =\{ 
u \& v : \mc{I} \longrightarrow \mc{X} \amalg \mc{Y} \, \,
\mid
u \in U \, v \in V \}.
 \end{aligned}$$

\vspace{-1.5ex}
\noindent Note $u \& v$ denotes the mediating morphism
for $\amalg$ as the product in $\opTKersfin$.

\noindent Coproduct
\vspace{-.8ex}
$$\begin{aligned}
& \bm{\mc{X}} \oplus \bm{\mc{Y}} = 
(\mc{X} \amalg \mc{Y}, (\inj_{\mc{X}} \Comp U)  \cup
(\inj_{\mc{Y}} \Comp V), R \oplus S ), \quad \text{where}
\\ & \quad  R \oplus S : =\{   r
 \oplus s : \mc{X} \amalg \mc{Y} \longrightarrow  \mc{I} \mid
r \in R \, \, s \in S \}
 \end{aligned}$$

\vspace{-1.5ex}
\noindent Note $r \oplus s$ denotes the mediating morphism
for $\amalg$ as the coproduct in $\opTKersfin$.

The unit for the coproduct is $(\mc{T},\emptyset, \{ \emptyset \})$.

\begin{rem}[product/coproduct and tensor/cotensor]{\em 
The product and the coproduct of $\GopTKersfin$ do not coincide,
despite that the
forgetful functor makes them collapse into
the biproduct in $\opTKersfin$. Similarly,
another tensor product is defined, say the cotensor
$\parr$, owing to the nonsymmetricity of the second and the third
 components for the tensor object:
}\end{rem}
\noindent Cotensor product
\vspace{-1ex}
\begin{align*}
& \bm{\mc{X}} \parr \bm{\mc{Y}} = 
(\mc{X} \otimes \mc{Y}, W, R \otimes S),
\text{where} \\
& R \otimes S = \{ 
 \kappa \otimes
\tau : \mc{X} \otimes \mc{Y} \longrightarrow
\mc{I} \otimes \mc{I}   \cong  \mc{I}
   \mid \kappa \in R \, \,  \tau \in S \} \\
& 
%T&%= \GTKersfin(\mc{X},\mc{Y}^\perp) \\ & 
W=
\{ h : \mc{I} \longrightarrow \mc{X} \otimes \mc{Y}   \mid
\xymatrix{\mc{X} \ar[r]^(.4){\forall  \kappa \in R} &  \mc{I}}, 
\quad
 (\kappa \otimes \delta_{\mc{Y}})^* h \in V
\quad \text{and} \quad 
\xymatrix{\mc{Y} \ar[r]^(.4){\forall \tau \in S} & \mc{I}}, \, \,
\quad 
(\delta_{\mc{X}} \otimes \tau)^* h \in U
\}
\end{align*}
Note
$\xymatrix{
  (\kappa \otimes \delta_{\mc{Y}})^* h: 
\mc{I} \ar[r]^(.6){h} &  \mc{X} \otimes \mc{Y}
 \ar[r]^(.45){\kappa \otimes \delta_{\mc{Y}}} &
\mc{I} \otimes \mc{Y} \cong \mc{Y}} \text{and} 
\xymatrix{
 (\delta_{\mc{X}} \otimes \tau)^* h:
\mc{I} \ar[r]^(.6){h}  &  \mc{X} \otimes \mc{Y}  
\ar[r]^(.45){\delta_{\mc{X}} \otimes \tau}  & 
\mc{X} \otimes \mc{I} 
\cong \mc{X}}.$

\bigskip

Remind  for what follows that $*$ denotes the unique element of the
singleton measurable space $I$ of $\mc{I}$.

\smallskip

Our linear exponential comonad over $\opTKersfin$
in Section \ref{comonadTKer} lifts to that for $\GopTKersfin$,
directly along with Hyland-Schalk exponential construction in the double
glueing (cf. Section 4.2.2 of \cite{HSha}).
In \cite{HSha} an exponential structure in a double gluing category $ {\bf G}(C)$ is given through a natural transformation $\mathsf{k}:
\mc{C}(I, -) \longrightarrow \mc{C}(I, (-)_e)$,
which transformation makes $\mc{C}(I, -)$ linear distributive.
In our concrete framework $\opTKersfin$,
the natural transformation $\mathsf{k}$ is characterised concretely
as follows: 
\begin{defn}\label{defnatkappa}
{\em A natural transformation
$\mathsf{k} : \opTKersfin (\mc{I}, -) \longrightarrow \opTKersfin (\mc{I},
(-)_e)$ is defined by the following instance
$\natkappa{\mc{X}}{u}: \mc{I} \longrightarrow \mc{X}_e$
for every $u : \mc{I} \longrightarrow \mc{X}$:
\begin{align*}
\natkappa{\mc{X}}{u} (I, x_1 \cdots x_n)  & :=
u_e (I^{(n)}, x_1 \cdots x_n) \\
& = u_e^\bullet (\FI(I^{(n)}), (x_1, \ldots ,x_n)) \\
& =  u^n (I^n, (x_1, \ldots ,x_n) ) &  \text{by $\FI(I^{(n)})=I^n$} \\
&\textstyle = \prod\limits_{i=1}^n u(I, x_i)
\end{align*}
}
\end{defn}

\bigskip

\noindent Definition \ref{defnatkappa} is well defined so that 
the naturality of $\mathsf{k}$
$$\begin{aligned}
\iota_e \Comp \mathsf{k}_{\mc{X}} (u) =
\mathsf{k}_{\mc{Y}} (\iota \Comp u
): \mc{I} \longrightarrow \mc{Y}_e & \quad \mbox{for any $\iota : \mc{X} \rightarrow \mc{Y}$
and $u : \mc{I} \rightarrow \mc{X}$ in $\opTKersfin$}
\end{aligned}$$
is shown as follows with any $y_1 \cdots y_n \in \mc{Y}_e$;
\begin{align*}
RHS(I,y_1 \cdots y_n)  & =
\textstyle 
\prod_{i=1}^n (\iota \Comp u) (I, y_i)
 = 
\prod_{i=1}^n 
\int_X u(I, x) \, \iota (dx, y_i) \\
LHS(I, y_1 \cdots y_n) &  =
\textstyle
\int_{X_e} \mathsf{k}_{\mc{X}} (I, \msbf{x}) \, \iota_e (d \msbf{x},
y_1 \cdots y_n )  \\
&  =  \textstyle
\int_{X_e}
\mathsf{k}_{\mc{X}} (I, \msbf{x}) \, \iota_e^{\bullet} (\FI (d \msbf{x}),
(y_1, \ldots ,y_n) )
\tag*{Def of $\iota_e$ in (\ref{exker})} \\
& =
\textstyle \int_{X_e^\bullet}
\mathsf{k}_{\mc{X}} (I, F(\vec{\msbf{x}})) \, \iota_e^{\bullet}
( d \vec{\msbf{x}},
(y_1, \ldots,  y_n) ) 
\tag*{by variable change}
\\
& \textstyle = 
\int_{X^n}
\mathsf{k}_{\mc{X}} (I, x_1 \cdots x_n ) \, \iota_e^{\bullet}
( d (x_1, \ldots , x_n),
(y_1, \ldots y_n) )  \tag*{Def of $\iota_e^\bullet$}
\\
& \textstyle = 
\int_{X^n}
\prod_{i=1}^n  u(I, x_i)
\prod_{i=1}^n \iota(dx_i, y_i) \\
& \textstyle 
=  \int_{X^n}
\prod_{i=1}^n  u(I, x_i) \, \iota(dx_i, y_i)
 \\
& \textstyle 
= \prod_{i=1}^n 
\int_{X}
 u(I, x) \, \iota(dx, y_i)
\tag*{Fubini-Tonelli}
 \end{align*}

\smallskip

\begin{lem} \label{lemortex}
For any  $u : \mc{I} \longrightarrow \mc{X}$, 
$$\begin{aligned}
(\di{\mc{X}})^* (\natkappa{\mc{X}}{u})= u
\end{aligned}$$
That is, the natural transformation
$
\opTKersfin (\mc{I}, {\sf d}):
\opTKersfin (\mc{I}, -) \longrightarrow 
\opTKersfin (\mc{I}, (-)_e)$
induced by ${\sf d}$ of Definition \ref{natder}
is a left inverse of $\mathsf{k}$ so that 
$\opTKersfin (\mc{I}, {\sf d}) \Comp
\mathsf{k} = \operatorname{Id}_{\opTKersfin (\mc{I}, -)}$.
\end{lem}
\begin{proof} 
\begin{align*}
\textstyle (\di{\mc{X}} \Comp \natkappa{\mc{X}}{u}) (I, x)
= 
\int_{X_e} \natkappa{\mc{X}}{u}(I,  \msbf{y}) \di{\mc{X}}(d
 \msbf{y},x)
= 
\int_{X_e \cap X^{(1)}} \natkappa{\mc{X}}{u}(I,  y) \di{\mc{X}}(d
 y,x) 
= 
\natkappa{\mc{X}}{u}(I,  x) 
= 
u(I,x)
\end{align*}
The third equation is because  $\di{\mc{X}}(d y, x)= \delta(dy, x)$ as
$y \in X_{e} \cap X^{(1)} $.
\end{proof}

\smallskip

In order to define certain exponential comonad in $\GopTKersfin$,
the following linear distributivity is crucial, guaranteeing 
to respect the comonoid structure of $\opTKersfin$.

\smallskip

\begin{lem}[linear distributity of
$ \natkappa{\mc{X}}{-}$] \label{ldnatkap}
The natural transformation $ \natkappa{\mc{X}}{-} :
\opTKersfin(\mc{I},-) \longrightarrow  \opTKersfin(\mc{I},(-)_e)$
meets the following criteria (i), (ii) and (iii) of Hyland-Schalk
 (cf. pg.209 \cite{HSha}) in order to make $\opTKersfin(\mc{I},
 -)$ linear distributive: 

{\em \noindent (Remind the notation below that 
$\mc{C}(\mc{I}, -)$ is a functor so instantiated both by object and by morphism.)}

\noindent(i) well-behavior wrt the comonad structure
$$\xymatrix%@R=14pt
{ & 
\ar[dl]_{\operatorname{Id}_{\mc{C}(\mc{I}, \mc{X})}}
\mc{C}(\mc{I}, \mc{X}) \ar[r]^{
\mathsf{k}_{\mc{X}}}
\ar[d]^{\mathsf{k}_{\mc{X}}}
  &  \mc{C}(\mc{I}, ! \mc{X})
 \ar[d]^{\mathsf{k}_{! \mc{X}}} \\
\mc{C}(\mc{I}, \mc{X}) & 
\ar[l]^{ \mc{C}(\mc{I}, \di{\mc{X}})}
\mc{C}(\mc{I}, ! \mc{X}) \ar[r]_{
\mc{C}(\mc{I}, \stor{\mc{X}})}  &  \mc{C}(\mc{I}, ! ! \mc{X})
}$$
\noindent (ii) respecting the comonoid structure
$$\xymatrix%@R=14pt
{ \emptyset
\ar[d]_{\ulcorner \operatorname{Id}_{\emptyset} \urcorner}
& 
\ar[l]%_{\exists_\emptyset} 
\ar[d]^{}
\mc{C}(\mc{I}, \mc{X}) \ar[r]^(.4){\Delta_{\mc{C}(\mc{I}, \mc{X})}}
\ar[d]^{\mathsf{k}_{\mc{X}}}
  &  \mc{C}(\mc{I}, \mc{X}) \otimes \mc{C}(\mc{I}, \mc{X}) 
\ar[r]^{\mathsf{k}_{\mc{X}} \times \mathsf{k}_{\mc{X}}}
   &  \mc{C}(\mc{I}, ! \mc{X}) \times \mc{C}(\mc{I}, ! \mc{X})
 \ar[d]^{\otimes} \\
\mc{C}(\mc{I}, \mc{I}) & 
\ar[l]^{ \mc{C}(\mc{I}, \wk{\mc{X}})}
\mc{C}(\mc{I}, ! \mc{X}) \ar[rr]_{
\mc{C}(\mc{I}, \con{\mc{X}})}  &  &  \mc{C}(\mc{I}, ! \mc{X} \otimes ! \mc{X})
}$$
\noindent (iii) monoidal
$$\xymatrix%@R=14pt
{ \mc{C}(\mc{I}, \mc{I})
\ar[r]^{\mathsf{k}_{\mc{I}}} & \mc{C}(\mc{I}, !\mc{I})
& 
\mc{C}(\mc{I}, \mc{X})
\times
\mc{C}(\mc{I}, \mc{Y})
 \ar[d]_{\mathsf{k}_{\mc{X}}  \times \mathsf{k}_{\mc{Y}} }
 \ar[rr]^{\otimes}  &  &
\mc{C}(\mc{I}, \mc{X} \otimes \mc{Y})  
 \ar[d]^{
\mathsf{k}_{\mc{X} \otimes \mc{Y}}} \\
\ar@{}[ur]^{=}
\mc{C}(\mc{I}, \mc{I})
\ar[r]^{
\mc{C}(\mc{I}, \moni)} 
& \mc{C}(\mc{I}, !\mc{I})
& 
 \mc{C}(\mc{I}, ! \mc{X}) \times
 \mc{C}(\mc{I}, ! \mc{Y} )
  \ar[r]_{\otimes}
   &  
 \mc{C}(\mc{I}, ! \mc{X} \otimes  ! \mc{Y})
 \ar[r]_{
\mc{C}(\mc{I}, \mon{\mc{X}}{\mc{Y}})} & 
\mc{C}(\mc{I}, ! (\mc{X} \otimes \mc{Y}))
}$$
\end{lem}

\begin{proof}
(ii) and (iii) are direct as so are the following equations
stipulating the commutativity diagrams: \\
(ii) $\natkappa{\mc{X}}{u} \otimes \natkappa{\mc{X}}{u} =
\con{\mc{X}} \Comp \natkappa{\mc{X}}{u}$
and $\di{\mc{X}} \Comp \natkappa{\mc{X}}{u} =
\ulcorner \operatorname{Id}_{\emptyset} \urcorner \Comp \exists_\emptyset$
where $\exists_{\mc{\emptyset}}$ is the unique morphism to the empty measurable space 
(emptyproduct as terminal object) 
and $\ulcorner \operatorname{Id}_{\emptyset} \urcorner
: \emptyset \longrightarrow \opTKersfin(\mc{I}, \mc{I})$.

\noindent (iii) $\natkappa{\mc{X} \otimes \mc{Y}}{u \otimes v}
= 
\mon{\mc{X}}{\mc{Y}} \Comp
(\natkappa{\mc{X}}{u} \otimes \natkappa{\mc{Y}}{v})
$ and
$\natkappa{\mc{I}}{u}=\moni \Comp u$

\smallskip

Hence, we need to prove (i) having two equalities:

\noindent (i-a) 
$\natkappa{\mc{X}_e}{\natkappa{\mc{X}}{u}}
= s_{\mc{X}} \Comp \natkappa{\mc{X}}{u} $. 

\smallskip
\noindent 
Let $\msbf{x}_1 \cdots \msbf{x}_m
 \in X_{ee}$ so that $\msbf{x}_i = x_{i1} \cdots x_{in_i}$ for certain $n_i$
with $i=1, \ldots
 , m$. \\
Recall Definition \ref{mapabs} that 
 $\abs{\msbf{x}_1 \cdots \msbf{x}_m}=x_{11} \cdots x_{1n_1}
\cdots x_{m1} \cdots x_{mn_m}$.

$\begin{aligned}
 RHS(I, \msbf{x}_1 \cdots \msbf{x}_m)
& \textstyle =
\Int{X_e}{\natkappa{\mc{X}}{u}(I, \msbf{y})}{\stor{\mc{X}}(d
  \msbf{y}, \msbf{x}_1 \cdots \msbf{x}_m)} 
 =
\Int{X_e}{\natkappa{\mc{X}}{u}(I, \msbf{y})}{\Dd{\abs{\msbf{x}_1 \cdots \msbf{x}_m}}}{
d \msbf{y}}
\\
&\textstyle  = \natkappa{\mc{X}}{u}(I, \abs{\msbf{x}_1 \cdots \msbf{x}_m}) 
=  
u(I, x_{11} ) \cdots u(I, x_{1n_1})
\cdots u(I, x_{m1}) \cdots u(I, x_{mn_m})
\end{aligned}$

$\begin{aligned}
\textstyle LHS(I, \msbf{x}_1 \cdots \msbf{x}_m)
= \prod\limits_{i=1}^m \natkappa{\mc{X}}{u} (I,\msbf{x}_i) 
= \prod\limits_{i=1}^m \prod\limits_{j=1}^{n_i} 
u (I,x_{ij}) 
\end{aligned}$

\bigskip

\noindent (i-b) 
$ \di{\mc{X}} \Comp \natkappa{\mc{X}}{u}
= u$ \\
This is by Lemma \ref{lemortex}.
\end{proof}

\smallskip

With Lemma \ref{ldnatkap}, the following is an instance of
$\mc{C}=\opTKersfin$ of Hyland-Schalk's Proposition 36 of \cite{HSha}
on ${\bf G}(\mc{C})$.

\smallskip

\begin{prop}[Hyland-Schalk exponential comonad on glueing \cite{HSha}.] \label{gluex}
There are two kinds (I) and (II) of  linear exponential comonad on $\GopTKersfin$
as follows so that the forgetful functor to $\opTKersfin$
preserves the structure. For an object
$\bm{\mathcal{X}} = (\mathcal{X}, U, R)$ in $\GopTKersfin$,
\end{prop}
%\vspace{-2ex}

\smallskip

$\begin{aligned}
(I) \quad  \bm{\mathcal{X}}_e = (\mathcal{X}_e, \natkappa{\mc{X}}{U}, \opTKersfin(\mathcal{X}_e,
 \mathcal{I}) ),  \quad \mbox{where} \quad \natkappa{\mc{X}}{U} =  
 \{ \natkappa{\mc{X}}{u} : 
\mathcal{I} \longrightarrow 
\mathcal{X}_e  \mid u \in U \}
\end{aligned}$
$$\begin{aligned}
& (II) \quad \bm{\mathcal{X}}_e = (\mathcal{X}_e, \natkappa{\mc{X}}{U}, ? R ), 
\quad \mbox{where $\natkappa{\mc{X}}{U}$ as above,
but $?R$ is the smallest subset of $\opTKersfin(\mathcal{X}_e,
  \mathcal{I})$} \\ 
& (a) \, \mbox{Containing $\xymatrix{ \{ (\di{\mc{X}})_* \mu : \mc{X}_e
 \ar[r]^(.7){\di{\mc{X}}} & \mc{X}
\ar[r]^(.4){\mu} & \mc{I} \mid \mu \in R \} } $ } \\
& (b) \, \mbox{Containing the weakening $\wk{\mathcal{X}}: \mathcal{X}_e
   \longrightarrow \mathcal{I}$} \\
& (c) \,\mbox{Closed under the following %interaction with $U_e$ via
for the contraction $\con{\mathcal{X}}$
:} \\
& \hspace{8ex} \mbox{For any $h : \mc{X}_e \otimes \mc{X}_e \longrightarrow \mc{I},$ 
} \\
& \hspace{8ex} \forall u \in U \, \, \left[ ( \natkappa{\mc{X}}{u}  \otimes
   \operatorname{Id_{\mc{X}_e}})_* h \in ?R \wedge 
 (\operatorname{Id_{\mc{X}_e}} \otimes \natkappa{\mc{X}}{u} )_* h \in
   ?R \, \right] 
 \Longrightarrow \, (\con{\mathcal{X}})_* h \in ?R.
\end{aligned}$$

\smallskip

\subsection{Orthogonality as Relation between
Measures and Measurable Functions
% $\TKersfin(\mc{X},\mc{I})$ 
%and $\TKersfin(\mc{I},\mc{X})$ 
} \label{ort}
The Hyland-Schalk orthogonality relation \cite{HSha},
when applied concretely to the measure theoretic
framework in the present paper, becomes a
relation between measures and
measurable functions over a measurable space.
The relation is shown to satisfy a property ``reciprocity'',
which is derivable from the adjunction of the inner product
in terms of Lebesgue integral of a
measurable function over a measure. 
Our reciprocal orthogonality is strong enough
to guarantee a certain relevant structure maps Hyland-Schalk employed
in \cite{HSha} to obtain the product and exponential structures
for the slack orthogonality category
$\Sla{\mc{C}}$, which forms a subcategory
of the double glueing ${\bf G}(\mc{C})$.
Note in this subsection, we do not assume any closed structure (i.e.,
the linear implication).
%Although the original relation 
%stipulates the conditions only for the tensor
%as well as for the isomorpohism, identity.
The category $\mc{C}$ considered in this subsection is 
either $\opTKersfin$ or $\opTKer$.

\smallskip

\begin{defn}[orthogonality on a monoidal category $\mc{C}$]\label{monoort}{\em
An orthogonality on a monoidal category $\mc{C}$
is a family of relation $\bot_{R}$
between maps $I \longrightarrow R$ and those $R \longrightarrow I$
satisfying the following conditions on isomorphism,
identity and on  tensor for a monoidal category
by Hyland-Schalk (cf. Definition 45 \cite{HSha}).
Note: Although the original orthogonality is for a monoidal closed
 category,
we consider a general monoidal one without the implication.

\noindent (isomorphism) \\
If $\iota: R \longrightarrow S$
is an isomorphism then for any 
$u: I \longrightarrow R$ and
$x: R \longrightarrow I$,
\begin{align*}
\ort{u}{R}{x} \quad \mbox{iff} \quad 
\ort{\iota \Comp u}{S}{x \Comp \iota^{-1}}
\end{align*}

\noindent (identity) \\
For all $u: I \longrightarrow R$ and
$x: R \longrightarrow I$,
\begin{align*}
\ort{u}{R}{x} \quad \mbox{implies} \quad 
\ort{\operatorname{Id}_{I}}{I}{x \Comp u}
\end{align*}

\noindent (tensor) \\
Given $u: I \longrightarrow R$,
$v: I \longrightarrow S$ and
$h: R \otimes S \longrightarrow I$,
$\ort{u}{R}{h \Comp (\operatorname{Id}_{R} \otimes v)}$
and 
$\ort{v}{S}{h \Comp (u \otimes \operatorname{Id_{S}})}$
imply 
$\ort{u \otimes v}{R \otimes S}{h}$.

}\end{defn}

\bigskip

For $U \subseteq \mc{C}(I,R)$, its orthogonal $\crc{U} \subseteq  \mc{C}(R,J)$
is defined by 
\begin{align*}
\crc{U}:= \{ x :  R \longrightarrow J \mid \forall u \in U \, \,  \ort{u}{R}{x}\} 
\end{align*}
This gives a Galois connection so that $U^{\circ \circ \circ } =
\crc{U}$.
The operator $\ccrc{(~)}$ is called the closure operator in the sequel.

\smallskip
\begin{lem}[reciprocal orthogonality] \label{adjort}
If a family of relation satisfies the following
for every $u: I \longrightarrow R$,
$x : S \longrightarrow J$, and $f: R \longrightarrow S$,
\begin{align}
  \ort{u}{R}{x \Comp f} \quad \quad
& \mbox{if and only if} \quad \quad  \ort{f \Comp u}{S}{x} \label{exadiort},
\end{align}
then this becomes an orthogonality relation.
An orthogonality satisfying (\ref{exadiort})
is called {\em reciprocal} 
\end{lem}
\begin{proof}
We show the condition (\ref{exadiort}) entails
the three conditions (isomorphism), (identity) and (tensor)
of Definition \ref{monoort}.
Moreover, the (tensor condition) is strengthened into
the (precise tensor) obtained by replacing ``imply'' with ``iff''.
The derivation is direct 
for the isomorphism and the identity conditions.
For the precise tensor condition, observing $u \otimes v =
(\operatorname{Id}_{R} \otimes v) \Comp u 
= (u \otimes \operatorname{Id}_{R}) \Comp v$,
one antecedent implies the descendant by reciprocity 
either on $\operatorname{Id}_{R} \otimes v$
or on $u \otimes \operatorname{Id}_{S}$, and vice versa.
\end{proof}

\bigskip
In order to construct an orthogonality on $\opTKersfin$,
we define an inner product of $\opTKersfin$,
which has the adjunction property:
\begin{defn}[inner product]\label{inprod}{\em
For a measure $\mu \in \opTKer(\mc{X}, \mc{I})$ and a measurable function $f \in
\opTKer(\mc{I}, \mc{X})$, we define
$$\begin{aligned}
\inpro{f}{\mc{X}}{\mu} := \int_{X} f d\mu 
\end{aligned}$$
}
\end{defn}
\bigskip
\noindent Then the two operators in Definition 
\ref{opekernel} become characterised as follows:
 \begin{lem}[adjunction between $\kappa^*$ and $\kappa_*$] \label{adjn}
In $\opTKersfin$, for any 
measure $\mu: \mc{X} \longrightarrow \mathcal{I}$,
any measurable function
$f: \mc{I} \longrightarrow \mathcal{Y}$ and 
any transition kernel $\kappa: \mc{Y} 
\longrightarrow \mc{X}$,  
\begin{align*}
\inpro{f}{\mc{Y}}{\kappa_* \mu }
   & = \inpro{\kappa^*  f}{\mc{X}}{\, \mu}
\end{align*}
Remind that $\kappa_* \mu =  \mu \Comp \kappa$ and 
$\kappa^* f=\kappa \Comp f$.
\end{lem}
\begin{proof}
The following starts from LHS and ends with RHS of the assertion,
using Fubini-Tonelli:
\begin{align*}
%\int_{Y} f d (\kappa_* \mu)=
\int_{Y} f (y) (\kappa_* \mu) (dy) 
=
\int_Y f(y) \int_X \kappa(dy, x) \mu(dx) 
=
\int_X \mu(dx)  \int_Y f(y) \kappa(dy,x)
=  \int_{X} (\kappa^* f)(x) \mu(dx)
%=\int_{X} (\kappa^* f) d \mu
\end{align*} 
\end{proof}

\bigskip
Using the inner product, an orthogonality relation on $\opTKersfin$
is defined.
\begin{defn}[orthogonality in terms of integral] \label{ortint}
{\em For a measurable function $f \in
\opTKersfin(\mc{I}, \mc{X})$ and 
a measure $\mu \in \opTKersfin(\mc{X}, \mc{I})$, the relation $\bot_{\mc{X}}
\subset \opTKersfin(\mc{I},\mc{X}) \times \opTKersfin(\mc{X},\mc{I})$ is defined
\begin{align*}
\ort{f}{\mathcal{X}}{\mu} \quad \text{if and only if} \quad 
\inpro{f}{\mc{X}}{\mu} \, \, \leq 1
 \end{align*}
}\end{defn}

\bigskip

\begin{prop}[$\bot_{\mc{X}}$ is a reciprocal orthogonality in
$\opTKersfin$] \label{botrecipro}
The relation $\bot_{\mc{X}}$ defined in Definition \ref{ortint} 
is an orthogonality in $\opTKersfin$,
and moreover is reciprocal.
\end{prop}
\begin{proof}
By Lemma \ref{adjort},
it suffices to show that the relation is reciprocal to satisfy the
condition (\ref{exadiort}),
 which is derived by Lemma \ref{adjn}.
\end{proof}

\bigskip

The orthogonality of Definition \ref{ortint}
gives rise to the following full subcategory
of $\GopTKersfin$. 
\begin{defn}[slack category $\Sla{\opTKersfin}$ (cf. \cite{HSha} for
the general definition)] \label{slati}~~~{\em
The {\em slack orthogonality category} $\Sla{\opTKersfin}$
is the full subcategory of $\GopTKersfin$ on those objects 
$(\mc{X},U,R)$ such that $U \subseteq R^\circ$ and
$R \subseteq U^\circ$.
}
\end{defn}

%\smallskip

\begin{exam}[objects of $\Sla{\opTKersfin}$]
\footnote{The example
 is not a prerequisite for the rest.} \\
{\em 
The following independent examples (i) and (ii) guarantee that 
the slack category of Definition \ref{slati}
is not degenerate so that $U$ and $R$ become in general continuous.

\smallskip

\noindent (i)
When $(X, {\cal X})$ is the Borel-field
$(\mathbb{R}, \mc{B}(\mathbb{R}))$,
let $\mu$ be the Lebesgue measure over $\mc{B}(\mathbb{R})$
and $f$ be a density function of a probability distribution on
${\cal X}$. That is, $f$ is Lebesgue integrable such that 
$\int_{- \infty}^a f(x) d \mu = \mu((-\infty, a ])$ 
for a certain probability measure $\mu$ on $\mc{X}$ with any $a$.
A particular example is well known for the
exponential distribution with rate $r$, for which $f$
is given by $f(x)= r e^{- r x}$ for $x \geq 0$,
where it equals zero otherwise.
\begin{align*}
\mbox{Define} \quad \quad \textstyle U_f :=\{ \chi_{(a,b)}  f  \mid a < b \in \mathbb{R}  \}
\quad \mbox{and} \quad 
R_\mu := \{  
%\mu \! \restriction_{(a,b)} : A \mapsto  \mu (A \cap (a,b))
\lambda A \in \mc{X} . \, \mu (A \cap (a,b))
\mid a < b \in \mathbb{R}  \}, 
\end{align*}
in which $\chi_{(a,b)}$ is the characteristic function of an interval $(a,b)$. \\
%and $\mu \! \restriction_{(a,b)}$ is the restriction of $\mu$
%to the sub $\sigma$-field 
%generated by the Borel subsets of $(a,b)$. \\
Then $(\mc{X},U_f,R_\mu)$ is an object of the slack category.
\begin{align*}
\textstyle  \int_{-\infty}^{\infty} (\chi_{(a,b)} f)(x) \mu(dx \cap (a',b'))
=\int_a^b f (x) \mu(dx \cap (a',b'))
= \int_{(a,b) \cap (a',b')} f d \mu \leq \int_{-\infty}^\infty f d \mu
=1 
\end{align*}
The last equality and inequality are
the properties of density functions. 
%$U_f$ and $R_\mu$ are in general continuous  
%so that the slack category of Definition \ref{slati}
%is not degenerate.

\smallskip

The density function for a probabilistic distribution is widely recognised
for providing the {\em importance sampling} procedure in
probabilistic programming. 
This procedure involves drawing a sample from a distribution based on its density function, which describes the likelihood of a scored point. 
For syntax and semantics
on sampling concerning the score, refer to \cite{BLGS, Stat, HKSY}.
In this context, our second component $U_f$ may be seen to comprise sampling procedures whose $\chi_{(a,b)}$
function serves to confine the sampling space within the interval in 
$\mathbb{R}$. 
Each sampling choice is implemented through a double-glueing morphism $\kappa$
from the slack tensor unit 
$(\mc{I}, \{ \operatorname{Id}_{\mc{I}}  \},
\{ \operatorname{Id}_{\mc{I}}  \}^{\circ}
)$ (cf. Lemma \ref{monoprosla} below)
to the  object as dictated by the first
morphism condition $\kappa \Comp \operatorname{Id}_{\mc{I}} \in U_f$
in Definition \ref{CatGop}, while the second condition, becoming the
orthogonality  
$\mu' \Comp \kappa = \int_{a}^b \kappa(x) \mu(dx) \leq 1$ for $\mu' \in R_\mu$,
is automatic by the slack object.
Note that a kernel $\kappa$ of this type in
$\opTKersfin$ is identified with a measurable function
on $\mc{X}$.

\bigskip

\noindent (ii)
Given any measurable space $\ms{X}$ and any s-finite measure $\mu$ (i.e., an element $\opTKersfin (\mathcal{X}, \mathcal{I})$).
$L_1^+ (\mathcal{X}, \mu)$ denotes the subclass of non-negative measurable $h$
whose $L_1$-norm is finite; i.e.,
$\lVert h \rVert := \int_X  h  \, d \mu < \infty$. 
Note that 
$L_1^+ (\mathcal{X}, \mu)$ is contained in $\mf{X}$ 
defined in Proposition \ref{lemchar}.
Recall in the proposition that a transition kernel $\kappa: X \times \mc{X} \rightarrow \zeroinf$
induces the endomap $\kappa^*$ on $\mf{X}$, hence
on $L_1^+ (\mathcal{X}, \mu)$.
In this example, $\kappa$ is called a $\mu$-{\em 
contraction} when $\lVert \kappa^* g \rVert \leq 
\lVert g \rVert$ for any $g \in L_1^+ (\mathcal{X}, \mu)$.
That is,  $\kappa^*$ on $L_1^+ (\mathcal{X}, \mu)$
contracts the norm. 

Fix an arbitrary  
$f \in \opTKersfin (\mathcal{I}, \mathcal{X})$ 
such that $\lVert f \rVert \leq 1$. 
Define
\begin{align}
 U_f:=\{ \kappa^* f \mid \mbox{$\kappa :
\mc{X} \rightarrow \mc{X}$ is a contraction s-finite kernel} \}
\quad \mbox{and}  \nonumber \\ 
R_\mu := \{ \kappa_* \mu \mid
\mbox{$\kappa :
\mc{X} \rightarrow \mc{X}$
is a contraction s-finite kernel} \}  \label{UR}
\end{align}
Note that $f \in U_f$ and $\mu \in R_\mu$ as
we may take $\kappa$ to be the unit kernel which is a $\mu$-contraction. \\
$(\mc{X},U_f,R_\mu)$ is an object of the slack category as 
the two inclusions of Definition \ref{slati} hold by the following
for any $\mu$-contraction kernels $\kappa$ and $\tau$ on $\mc{X}$: 
\begin{align*} 
\inpro{\kappa^* f}{\mc{X}}{\tau_* \mu} =
\inpro{\tau^* \kappa^* f}{\mc{X}}{\mu} =
\inpro{ (\kappa \Comp \tau)^* f}{\mc{X}}{\mu} \leq 
\inpro{f}{\mc{X}}{\mu}  \leq 1 
\end{align*}
The first equality is by reciprocity and third inequality is by preservation of contraction under kernel composition.

\smallskip

\smallskip
\noindent Contraction kernels $\kappa$'s in (\ref{UR})
are exemplified in terms of conditional expectations for a probability measure $\mu$:
Given a probability space $(\ms{X}, \mu)$ and 
a sub $\sigma$-field $\mc{G}$ of $\mc{X}$,
the {\em conditional expectation} 
$E[g \! \mid \! \mc{G}]$ of $g \in L_1^+ (\ms{X}, \mu)$
is the $\mc{G}$-measurable function
such that $\int_{A} E[g \! \mid \! \mc{G}] d \mu = \int_A g  d \mu$ for all measurable $A \in \mc{G}$.  
Then the following inequality is directly derivable 
$$\lVert E[g \! \mid \!  \mc{G}] \rVert \leq \lVert g \rVert $$ 
from the well know property $E[E[g \! \mid \mc{G}]]=E[g]$
together with Jensen's inequality $\abs{E[g \! \mid \mc{G}]} \leq E[\abs{\! g \!} \mid \mc{G}]$  (cf. respectively (15.6) and (15.12) of \cite{Bau}).

The inequality means the conditional expectation {\em contracts} the $L_1$-norm:
\begin{align}
E[\sim \mid \mc{G}]: L_1^+ ((X, \mc{X}), \mu) \longrightarrow 
L_1^+ ((X, \mc{G}), \mu) \quad \quad g \longmapsto E[g \! \mid \!  \mc{G}]
\label{condE}
\end{align}
 The conditional expectation map (\ref{condE})
is known to be linear, positive and preserving monotone convergence
(cf. Section 15 of \cite{Bau} in particular (15.13)
for the  conditional monotone convergence).
These
conditions happen to be the same as those for morphisms in $\Mes$ used in 
Proposition \ref{lemchar}. Hence the same argument applies to conclude that
the map (\ref{condE}) coincides with $\kappa^*$ on $L_1^+ (\ms{X}, \mu)$
for certain kernel $\kappa$.

\smallskip

%The above construction guarantees that 
%$U_f$ and $R_\mu$ of (\ref{UR}) become in general continuous
%so that the slack category of Definition \ref{slati}
%is not degenerate: E.g., consider when
Consider when ${\cal X}={\cal B}(\mathbb{R})$ and 
$\mc{G}$ is generated by the Borel subsets of
an interval $(a,b)$. Then 
the conditional expectation of (\ref{condE})
is 
$$\textstyle 
E[g \! \mid \!   \mc{G}]= \frac{1}{b-a} \chi_{(a,b)} \int_a^b  g(t) dt$$ 
This determines the action $\kappa^*$ on functions,
and correspondingly  
the action $\kappa_*$ on measures by means of  
$\kappa_* \mu (B)=\int_{X} \kappa^* \chi_B (x) \mu (dx)$
(cf. (\ref{kapchi})). In particular, 
$U_f$ and $R_\mu$ of (\ref{UR}) contain  
the following respective uncountable subsets:
%This makes both $U_f$ and $R_\mu$ of (\ref{UR}) continuous  
%as the two contain 
%respective uncountable subsets 
\begin{align*}
\textstyle \{  \frac{1}{b-a} \chi_{(a,b)} \int_a^b  f(t) dt
\mid a < b \in \mathbb{R}  \}
\quad \mbox{and} \quad 
\{  \lambda A \in \mc{X} . \,  \mu (A \cap (a,b))
%\mu \! \restriction_{(a,b)} : A \mapsto  \mu (A \cap (a,b))  
\mid a < b \in \mathbb{R}  \}
\end{align*}

%E[\sim \mid \mc{G}]: L_1^+ ((X, \mc{X}), \mu) \longrightarrow 
%L_1^+ ((X, \mc{G}), \mu) \quad \quad g \longmapsto E[g \! \mid \!  \mc{G}]

\smallskip

\smallskip
The above construction of the conditional expectation
is generalised when
a given measure $\mu$ is s-finite. We may write
$\mu = \Sigma_i \mu_i$ with each $\mu_i$ being a probability measure.
For each probability space $((X, \mc{X}), \mu_i)$, the conditional expectation
$E_i[\sim \mid \mc{G}]$ yields an endomap $(\kappa_i)^*$
on $\mf{X}$ for certain contraction kernel $\kappa_i$.
Then s-finite $\kappa:= \sum_i \kappa_i$ becomes a contraction s-finite kernel
by virtue of the property on the mixture of measures
$\int_X g \, d (\sum_i \mu_i) = 
\sum_i (\int_X g \, d \mu_i$).
}
\end{exam}

\bigskip
\begin{lem}[monoidal product in $\Sla{\opTKersfin}$] \label{monoprosla}
$\Sla{\opTKersfin}$ is a monoidal subcategory of $\GopTKersfin$
with the tensor unit
$(\mc{I}, \{ \operatorname{Id}_{\mc{I}}  \},
\{ \operatorname{Id}_{\mc{I}}  \}^{\circ}
)$.
\end{lem}
\begin{proof}
By virtue of the tensor condition for orthogonality,
the monoidal product of
Proposition \ref{catGTKer} is shown closed in the subcategory
$\Sla{\opTKersfin}$:
It suffices to show  that the third component of
$(\mc{X}, U, R) \,  \otimes \,  (\mc{Y}, V, S)$ is perpendicular to $U
 \otimes V$.
Take $\nu$ from the third component, then
$\forall f \in U \, \nu^* (f \otimes \delta_{\mc{Y}}) \in S \subseteq V^{\circ}$
and
$\forall g \in V \,
\nu^* (\delta_{\mc{X}} \otimes g) \in R \subseteq U^{\circ}$,
but which implies
$\ort{\nu}{\mc{X} \otimes \mc{Y}}{f \otimes g}$
by the tensor condition of the orthogonality
of Definition \ref{monoort}.
\end{proof}

\smallskip

In \cite{HSha},
to obtain an exponential structure as well as an additive one
for the slack category,
Hyland-Schalk employ
certain relevant structure maps for a general category $\mc{C}$.
We remark that $\mc{C}=\opTKersfin$ in this paper,
automatically validates their structure maps:

\begin{rem}[Hyland-Schalk's positive and negative maps
are implicated by reciprocity] \label{roipn}
{\em Hyland-Schalk
(in Definition 51 of \cite{HSha}) call a map $f: R \rightarrow S$
is {\em positive} (resp. {\em negative}) with respect to
$U \subseteq \mc{C}(I,R)$ and $Y \subseteq \mc{C}(S,I)$
when $\ort{f \Comp u}{S}{y}$ implies (resp. is implied by)
$\ort{u}{R}{y \Comp f}$ for all $u \in U, y \in Y$.
When $U$ and $Y$ are the whole homsets,
they say positive and negative outright.
When a map is both positive and negative, it is called {\em focused}.
We remark that our reciprocity (\ref{exadiort}) in $\opTKersfin$ 
ensures these property on maps: That is,
if an orthogonality is reciprocal, then any map is
automatically both positive and negative hence focused
with respect to any $U$ and $Y$.
}\end{rem}

\smallskip

By the remark, the slack category $\Sla{\opTKersfin}$ over $\opTKersfin$,
has the product and coproduct,
and moreover the exponential comonad as follows.

\begin{prop}[product and coproduct in $\Sla{\opTKersfin}$] \label{slpcp}~
The slack category $\Sla{\opTKersfin}$ is closed under the product
and the coproduct in
$\GopTKersfin$ of Proposition \ref{catGTKer}.
\end{prop}
\begin{proof}
By Hyland-Schalk's Propositions 52 in \cite{HSha}
because their presupposition  
positivity (res. negativity) of the projection (resp. injection)
of product (resp. coproduct) is entailed from
our reciprocal orthogonality condition
by Remark \ref{roipn}. 
\end{proof}

\begin{prop}[exponential comonad on the slack category] \label{propslack}
$\Sla{\opTKersfin}$ has the following exponential comonad:
$$
!(\mc{X}, U, R) =
(!\mc{X}, \natkappa{\mc{X}}{U}, ?R)
$$
where $?R$ is defined as in Proposition \ref{gluex},
but the clause (b) is replaced by;
$$
\{ \chi \cdot \wk{\mc{X}} \mid \ort{\operatorname{Id}_{\mc{I}}}{\mc{I}}{\chi} \}
\subseteq ?R
$$
\end{prop}
\begin{proof}
By Proposition 53 of \cite{HSha}, which supposition on the
positivity of the three structure maps $\di{}$,
$\wk{}$ and $\con{}$ (for $\mathsf{k}$ ) is a direct consequence in
$\opTKersfin$ by Remark \ref{roipn}
%By Lemma \ref{lempropslack}, which assures that
%all elements of $\natkappa{\mc{X}}{U}$ are perpendicular to
%those of $?R$ inductively on the construction (a), (b), (c).
%Then the reminder of the proof is by Proposition 53 of \cite{HSha}.
\end{proof}

\section{Discretisation $\TKersfinomg$
and Probabilistic Coherent Spaces} \label{closedTK}
This section starts with considering a discrete (i.e., countable) restriction
of the transition kernels within the transition matrices.
The restriction makes the integral for the categorical composition
into simpler algebraic sum, and turns out to
give an involution in the full subcategory $\TKeromg$
of the countable measurable spaces. The involution is directly
shown to imply dagger compact closedness of the subcategory $\TKersfinomg$.
Second, the double glueing is constructed over the dagger compact
closed category, so that a $*$-autonomous structure is obtained.
Finally, the orthogonality of the previous section is extended
over the involution, and the tight orthogonality subcategory 
$\Ti{\TKersfinomg}$ of the
double glueing is shown to coincide with Danos-Ehrhard's category of
probabilistic coherent spaces \cite{DE}.

\subsection{Involution in $\TKeromg$ and Dager Compact Closed $\TKersfinomg$}
When the set $Y$ of $\ms{Y}$ is countable, the integral of the composition 
(\ref{compTKer}) is replaced by the cruder sum:
\begin{align}
\iota \Comp \kappa (x,C) = &
\sum_{y \in Y} \kappa(x, \{ y \}) \iota(y,C) \label{compTKeromg}
\end{align}
In the countable case, $\kappa(x, \{ y \})$ is written
simply by $\kappa(x, y)$,
and the collection $\left( \kappa (x,y) \right)_{x \in X, y \in Y}$ is
called {\em a transition matrix},
as the composition (\ref{compTKeromg}) becomes the matrix
multiplication,  under the same countable condition making $\iota(y,C)$
into $\iota(y,\{ c \})$. This yields the full subcategory $\TKeromg$
consisting of
the countable measurable spaces in $\TKer$.

\begin{defn}[$\TKeromg$]{\em
A measurable space $\ms{X}$ is {\em countable} when
the set $X$ is countable.
$\TKeromg$ is the full subcategory whose objects are the countable measurable
 spaces in $\TKer$. Then the morphisms of $\TKeromg$ are characterised as
the transition matrices between two countable measurable spaces.
}\end{defn}

\smallskip

\begin{prop}[involution $(~)^*$] \label{inv}~\\
$\TKeromg$ is a dagger category \cite{Seli}
with the following self involutive
functor $(~)^*$, which is contravariant and
the identity on the objects. 
\end{prop}
% \noindent (On objects) $\ms{X}^* = (X^*, \mc{X}^*)$,
% $$\begin{aligned}
% & \text{where} \quad    X^* = \{
%  \delta_{a}(x):= \delta_{a,x} : X \longrightarrow \{ 0, 1 \}
% \mid a \in X  \} \quad \text{and} \quad
%  \{ \delta_{a}  \}_{a \in A} \in \mc{X}^* \, \text{if and only if} \,
% A \in \mc{X} 
% \end{aligned}$$
% Note the set $X^*$ is a subset of the homset $\TKer(\mc{I},\mc{X})$,
%  and is isomorphic to the set $X$.

% \smallskip

\noindent (On morphisms)
For a transition matrix $\kappa : \ms{X} \longrightarrow \ms{Y}$,
$\kappa^* : \ms{Y} \longrightarrow \ms{X}$ is given by the 
{\em transpose} of the matrix:
$$\begin{aligned}
\left( \, \kappa^*(y,x) \, \right)_{y \in Y, x \in X}
:=\left( \, {}^t \! \kappa (x,y) \, \right)_{x \in X, y \in Y}
\end{aligned}$$
% $\kappa^* : \ms{Y^*} \longrightarrow \ms{X^*}$ \, \text{is defined
% for given}\, 
% $\kappa : \ms{X} \longrightarrow \ms{Y}$,
% \begin{align*}
% &  \kappa^* (\delta_y,  \delta_x) := (\kappa^* \delta_y)(x)
% = \sum_{y \in Y} \kappa(x,y) \, \delta_y = \kappa(x,y)
% \end{align*} 
% Note $\left( \kappa^* \right)_{y \in Y, x \in X}$
% is seen the transpose $\left( {}^t \! \kappa \right)_{x \in X, y \in Y}$
% of the matrix $\kappa$.

%\smallskip

\begin{rem}{\em 
The involution $(~)^*$  is an {\em internalisation} of the contravariant
 equivalence of
Proposition \ref{contraeq} restricting the subcategory $\TKeromg$.
%$Note that $\mc{I}^*=\mc{I}$.
}\end{rem}

The involution of Proposition \ref{inv} with the monoidal product
directly yields the compact closed structure of the subcategory 
$\TKersfinomg$ of $\TKeromg$.
\begin{prop}[dagger compact closed category $\TKersfinomg$] \label{dccc}
Let $\TKersfinomg$ be the full subcategory of $\TKersfin$
consisting of the countable measurable spaces.
Then the dagger $\TKersfinomg$
of Proposition \ref{inv} becomes compact closed whose 
dual object (and its extension to the contravariant functor)
is given by the involution $(~)^*$. In particular, the monoidal closed
structure
is by one to one correspondence between the transition matrices
$\kappa((x,y),z)$ and $\kappa(x,(y,z))$;
$$\begin{aligned}
\TKersfinomg (\mc{X} \otimes \mc{Y},  \mc{Z} ) \, \cong \,  
\TKersfinomg (\mc{X}
, \mc{Y}^* \otimes \mc{Z} )
\end{aligned}$$
\end{prop}
\begin{proof}
The unit $\phi_{\mc{X}}: \mc{I} \rightarrow \mc{X}^* \otimes \mc{X}$
for the compact closedness
is given by the matrix whose element for each $(*, (x_1, x_2)) \in
I \times (X \times X)$ is
$\phi_{\mc{X}}(*, (x_1,x_2))= \delta_{x_1, x_2}$. The 
counit $\psi_{\mc{X}}:
\mc{X} \otimes \mc{X}^*   \rightarrow \mc{I}$
is the transpose matrix of the unit.
This directly yields 
the dagger compact closedness
$\phi_{\mc{X}} = \sigma_{\mc{X}, \mc{X}^*} \Comp (\psi_{\mc{X}})^*$.
\end{proof}

\smallskip

\begin{rem}[$\TKersfinomg = \TKeromg $]{\em
$\TKersfinomg$ coincides with $\TKeromg$.
As the former is a wide subcategory of the latter,
we need to check the fullness for the subcategory:
Every morphism in $\TKeromg$ is a transition matrix
$\left( \kappa (x,y) \right)_{x \in X, y \in Y}$,
whose each element $\kappa (x,y)$ is approximated as a countable sum
$\sum_{i \in \mathbb{N}} \kappa_i (x,y)$
with finite $\kappa_i (x,y)$'s.
Thus $\kappa=
\sum_{i \in \mathbb{N}} \sum_{(x, y) \in X \times Y}
 \kappa_i (x,y)$,
where each $\kappa_i (x,y)$ determines the transition matrix
$(\delta_{z, (x,y)} \kappa_i (x,y))_{z \in X \times Y} $.
The coincidence of the two categories
means that the discretisation makes the s-finiteness
redundant so that the well behaved monoidal composition is freely 
obtained in $\TKeromg$ with respect to Fubini-Tonelli.
The well behavior is a direct consequence
that the monoidal product and composition become algebraic
when the morphisms of continuous kernels collapse
into transition matrices in $\TKeromg$. 
} \end{rem}
In spite of the remark, in what follows in Subsection \ref{DGPCS},
we continue to use $\TKersfinomg$, whereby the connection
to the continuous case studied in the previous sections is seen
direct.

\subsection{Double Glueing and Probabilistic Coherent Spaces}
\label{DGPCS}

\subsubsection{*-autonomy of  $\GopTKersfinomg$}
The base category considered in this subsection is
$(\opTKersfin)_\omega$ denoting 
the full subcategory of $\opTKersfin$ consisting of the countable measurable spaces, which is equal to $(\TKersfinomg)^{\opsymbol}$. Hence the category is
denoted simply by $\opTKersfinomg$.
When Proposition \ref{catGTKer} takes
$\opTKersfinomg$ as the base category,
the double glueing category $\GopTKersfinomg$  
has a stronger form, inheriting the dagger compact closed
structure of $\opTKersfinomg$ of Proposition \ref{dccc}.

\begin{prop}[*-autonomy of $\GopTKersfinomg$] \label{catGTKeromg}
$\GopTKersfinomg$ is self involutive. 
$$\begin{aligned}
(\mc{X},U,R)^\perp =   (\mc{X}^*,R,U)
 \end{aligned}$$
modulo the equivalence $(~)^*$:
$\opTKersfinomg(\mc{I},\mc{X}) \cong  \opTKersfinomg(\mc{X},\mc{I})
$.
Moreover,
$\GopTKersfinomg$ becomes *-autonomous, whose 
monoidal closedness is given by the following implication 
in terms of the involution and the cotensor $\parr$:
\end{prop}
\noindent (implication)  
\vspace{-1ex}
\begin{align*}
& \bm{\mc{X}} \multimap \bm{\mc{Y}}:=
 \bm{\mc{X}}^\bot \parr \bm{\mc{Y}} = 
(\mc{X}^* \otimes \mc{Y}, W, U^* \otimes S),
\text{where} \\
& W=
\{ \nu : \mc{I} \longrightarrow \mc{X}^* \otimes \mc{Y}   \mid
\xymatrix{\mc{I} \ar[r]^(.4){\forall  \kappa \in U} &  \mc{X}}, 
\quad
\nu^* (\kappa^* \otimes \delta_{\mc{Y}}) \in V \quad \text{and} \quad 
\xymatrix{\mc{Y} \ar[r]^(.4){\forall g \in S} & \mc{I}}, \, \,
\quad 
\nu^* (\delta_{\mc{X}} \otimes g) \in R
\} 
\\
& U^* \otimes S = \{ 
 f^* \otimes
g : \mc{X}^* \otimes \mc{Y} \longrightarrow
\mc{I} \otimes \mc{I}   \cong  \mc{I}
   \mid f \in U \, \,  g \in S \}  
\end{align*}
Note that $W$ represents the homset $\GopTKersfinomg( \bm{\mc{X}}, \bm{\mc{Y}})$.

\bigskip
A direct corollary is obtained when the slack category studied in
Section \ref{ort} is restricted to the discrete measurable spaces.
The corollary states that 
the slack subcategory is a model of classical
linear logic. 

\begin{cor}
The slack orthogonality subcategory $\Sla{\opTKersfinomg}$ is 
defined to be the full subcategory of
$\Sla{\opTKersfin}$ whose $\mc{X}$ of Definition \ref{slati} 
is an object from $\opTKersfinomg$. Then $\Sla{\opTKersfinomg}$ 
is *-autonomous with product (hence coproduct), and equipped with
linear exponential comonad.
\end{cor}
\begin{proof}
It suffices to show the following (i),(ii) and (iii),
but whose latter two are direct and
(i) is by Theorem 54 of \cite{HSha},
whose supposition on the
positive (resp. negative) projections (resp. injections)
and the three positive  structure maps $\di{}$, $\wk{}$ and $\con{}$ (for $\mathsf{k}$ ) is by Remark \ref{roipn}: 
(i) The *-autonomy of the slack category $\Sla{\opTKersfinomg}$ 
is inherited from that of $\GopTKersfinomg$
of Proposition \ref{catGTKeromg} by Lemma \ref{monoprosla}. 
(ii) The product of Proposition \ref{slpcp}
is closed in the discrete subcategory. 
(iii) The linear exponential comonad in
Proposition \ref{propslack} 
is closed in the discrete subcategory.
\end{proof}

\subsubsection{Exponential Comonad for Tight Category $\Ti{\opTKersfinomg}$}
This parts starts with introducing the tight orthogonality
subcategory $\Ti{\opTKersfinomg}$ of $\GopTKersfinomg$,
and shows that the subcategory  is also a 
categorical model of linear logic.

% Remind that in the above definition
% $$\begin{aligned}
% & \mbox{} R^\circ = \{ x : \mc{I} \longrightarrow \mc{X} \mid
% \forall r \in R \, \, \ort{x}{\mc{X}}{r} \} 
%  \quad \text{and} \quad 
% U^\circ = \{ y : \mc{X} \longrightarrow \mc{I} \mid
% \forall u \in U \, \, \ort{u}{\mc{X}}{y} \}
% \end{aligned}$$

\smallskip

\begin{defn}[tight
 category $\Ti{\opTKersfinomg}$ (cf. Definition 47 \cite{HSha} for
the general definition)] ~~~{\em
The {\em tight orthogonality category} $\Ti{\opTKersfinomg}$
is the full subcategory of $\GopTKersfinomg$ on those objects 
$(\mc{X},U,R)$ such that $U = R^\circ$ and
$R = U^\circ$.
%That is, $\Ti{\opTKersfinomg}$
%is the full subcategory of $\Ti{\opTKersfin}$ consisting of the objects
%whose first components are countable measurable spaces.
}
\end{defn}

\smallskip
\begin{exam}[objects of $\Ti{\opTKersfinomg}$]
{\em 
Let $\mc{C}=\opTKersfinomg$:
For any subset $U$ of
$\mc{C} (I, X)$, 
$(X, \ccrc{U}, \crc{U})$ becomes an object of the tight
$\Ti{\mc{C}}$.
Dually for any $R \subset \mc{C} (X,I) $,
$(X, \crc{R}, \ccrc{R})$ becomes an object of the tight
$\Ti{\mc{C}}$.}
\end{exam}

A general lemma is prepared 
on the orthogonality broadly for the continuous $\opTKersfin$.

\begin{lem}[stable tensor] \label{sten}
Any reciprocal orthogonality on a monoidal category $\mc{C}$ 
{\em stabilises} the monoidal product: That is, 
For all $U \subseteq \mc{C}(I,R)$ and $V \subseteq \mc{C}(I,S)$, \\
\noindent (stable tensor)
\begin{align*} 
\crc{(\ccrc{U} \otimes \ccrc{V})} =
\crc{(\ccrc{U} \otimes V)} = \crc{(U \otimes \ccrc{V})}
\end{align*}
Hence in particular, the orthogonality in $\opTKersfin$
stabilises the monoidal product.
\end{lem} 
\begin{proof}
We prove ($\supset$) of the stable tensor condition as the converse is
 tautological. Take any $\nu \in RHS$, which means
$\forall f \in \ccrc{U} \, \forall g \in V $ $\ort{f \otimes g
= (f \otimes S) \Comp (I \otimes g)}{R \otimes S}
{R \otimes S \stackrel{\nu}{\longrightarrow} J}$
iff by reciprocity $\ort{g}{S}
{S \cong I \otimes S \stackrel{f \otimes S}{\longrightarrow} R \otimes S
\stackrel{\nu}{\longrightarrow} J}$. But this means 
$\forall h \in \ccrc{V} \, \, \ort{h}{S}{\nu \Comp (f \otimes S)}$
iff by reciprocity
$\ort{f \otimes h=(f \otimes S) \Comp (I \otimes h)}{R \otimes S}{\nu}$, which means $\nu \in LHS$.
\end{proof}

In the presence of the involution in $\opTKersfinomg$, it is necessary to
enhance the coherence of the orthogonality relation in
 Definition \ref{monoort} to account for the involution $(~)^*$. This augmented condition is called as  "symmetry orthogonality" within the general context in \cite{HSha}.

\noindent 
(symmetry)
Given $u: \mc{I} \longrightarrow \mc{X}$ and 
$v: \mc{X} \longrightarrow \mc{I}$,
$$\begin{aligned}
 \ort{u}{\mc{X}}{v} \quad \text{iff} \quad  \ort{v^*}{\mc{X}^*}{u^*}
\end{aligned}$$

Then this additional condition is satisfied in $\opTKersfinomg$,
to yield the following proposition, corresponding to
Proposition \ref{botrecipro}.
\begin{prop} \label{symmort}
$\bot$ is a symmetric orthogonal relation in $\opTKersfinomg$.
\end{prop}
\begin{proof}
The additional condition of symmetry is checked:  
$$\begin{aligned}
\inpro{f}{\mc{X}}{\mu}
= f \Comp \mu
= (f \Comp \mu)^*
= \mu^* \Comp f^*
= \inpro{\mu^*}{\mc{X^*}}{f^*}
\end{aligned}$$
The second equality is because the transpose of 
the scalars (i.e., of the homset $\opTKersfinomg(\mc{I},\mc{I})$)
is the identity.
\end{proof}

Moreover, $\opTKersfinomg$ has a monoidal closed  structure,
so the the coherence of the orthogonality on
implication of \cite{HSha} needs
to be augmented in Definition \ref{monoort}:

\noindent (implication) \\
Given $u: I \longrightarrow R$,
$y: S \longrightarrow I$ and $f: R \longrightarrow S$,
\begin{align}
\ort{u}{R}{y \Comp f} \quad \mbox{and} \quad 
\ort{f \Comp u}{S}{y}
\quad \mbox{imply} \quad 
\ort{\hat{f}}{R \multimap S}{u \multimap y} \label{ortimp}
\end{align}
where $\hat{f}$ is the transpose of $f$.

\smallskip

The strengthened condition is called
{\em precise implication} when ``imply'' is replaced by ``iff''
(same as the tensor condition was called).

\smallskip

\begin{lem} \label{impder}
The implication orthogonality is derivable from the symmetry orthogonality
together with the tensor one. 
Moreover the precise implication is derivable when the tensor condition is precise. Hence, in particular $\opTKersfinomg$
validates the precise implication condition.
\end{lem} 
\begin{proof}
% In a compact closed category the descendant of (\ref{ortimp})
% becomes  \\
% $\ort{\hat{f} = (R^* \otimes f) \Comp \eta_R  :
% \xymatrix{I \ar[r]^(.4){\eta_R} & R^* \otimes R
% \ar[r]^{R^* \otimes f}
% & R^* \otimes S}}{R^* \otimes S}{u^* \otimes y}$
% where $\eta_R$ is the unit of the compact closedness. 
% %\label{desortimp} Now (\ref{desortimp}) 
% This is 
% iff by reciprocity 
% $\ort{\eta_R}{R^* \otimes S}{ (u^* \otimes y) \Comp (R^* \otimes f)=
% u^* \otimes (y \Comp f)}$ iff by symmetry
% $\ort{(u^* \otimes (y \Comp f))^*=
% (f^* \Comp y*) \otimes u}{R^* \otimes R}{(\eta_R)^*= \epsilon_R
% : R^* \otimes R \rightarrow I} \quad \mbox{with $\epsilon$ is the counit.}$
% %\label{desortten} (\ref{desortten}) 
% This happens to be a descendant of (tensor) whose two antesedants are
% $\ort{f^* \Comp y^*}{R^*} \epsilon_R \Comp (R^* \otimes u)
% \quad \mbox{and} \quad 
% \ort{u}{R}{\epsilon_R \Comp ((f^* \Comp y^*) \otimes R)}$, 
% %\label{antortten}
% of which the right hand sides of the two orthogonality are
% respectively $u^*$ and
% $(f^* \Comp y^*)^* =y \Comp f$.
% But %(\ref{antortten}) 
% the both two orthogonalities are derivable from 
% the first antecedent of (\ref{ortimp}).

Apply the symmetry condition to the descendant of (\ref{ortimp}),
then $\ort{y^* \otimes u}{S^* \otimes R}{(\hat{f})^*}$.
This happens to be a descendant of (tensor)  whose two antecedents are
$\ort{y^*}{S^*}{(\hat{f})^* \Comp (S^* \otimes u)}$ and 
$\ort{u}{R}{(\hat{f})^* \Comp (y^* \otimes R)}$.
But (RHS)$^*$ of the first orthogonality is
$(u \multimap S) \Comp \hat{f}=
\widehat{(f \Comp u)}=f \Comp u$ and (RHS)$^{**}$ of the second orthogonality
is $((R \multimap y) \Comp \hat{y})^*=(\widehat{y \Comp f})^*=y \Comp f$.
Hence the two orthogonality are the antecedent of the tensor condition.
The second assertion is by the reversibility of the above argument
using precise tensor.
The third assertion is by Lemma \ref{sten}.
\end{proof}

\smallskip

\begin{lem}
If an orthogonality is symmetric,
the stable tensor implies the stable implication;
That is, for all $U \subseteq \mc{C}(I,R)$ and $Y \subseteq \mc{C}(S,I)$, \\
\noindent (stable implication)
\begin{align*} 
\crc{(\ccrc{U}  \multimap  \ccrc{Y})} =
\crc{(U \multimap \ccrc{Y})} = \crc{(\ccrc{U} \multimap Y)}
\end{align*}
Hence the orthogonality in $\opTKersfinomg$ stabilises the implication.
\end{lem}
\begin{proof}
The *-autonomy 
$X \multimap Z = (X \otimes Z^*)^*$
of $\opTKersfinomg$ 
makes the stable implication into a stable tensor via 
$(U^\circ)^*=(U^*)^\circ$, which equality is obtained directly
by %the duality $()^*$ and 
the symmetry orthogonality.
The second assertion is by Lemma \ref{sten} and Proposition \ref{symmort}.
\end{proof}

\smallskip
An orthogonality is called {\em stable} when it satisfies both stable tensor and
stable implication.

\smallskip
The rest of this part is devoted to observing that
the stable orthogonality of $\opTKersfinomg$
ensures, accordingly to Hyland-Schalk \cite{HSha},
an exponential comonad on $\Ti{\opTKersfinomg}$
as well as both monoidal product and product and coproduct.
That is,
\begin{thm}
The tight category
$\Ti{\opTKersfinomg}$
is a model of classical linear logic.
\end{thm}
\begin{proof}
By the following three Propositions
 \ref{monoproti}, \ref{tipcp}, and \ref{expcmti}.
\end{proof}

\begin{prop}[monoidal product in the tight orthogonality category]
\label{monoproti}
The tight category
$\Ti{\opTKersfinomg}$ has the following monoidal product so that
the forgetful to $\opTKersfinomg$ preserves the structure: 
$$(\mc{X}, U, R) \otimes (\mc{Y}, V, S)
= (\mc{X} \otimes \mc{Y}, (U \otimes V)^{\circ \circ},
(U \otimes V)^{\circ})$$
with the tensor unit $(\mc{I}, \{ \operatorname{Id}_{\mc{I}}  \}^{\circ \circ},
\{ \operatorname{Id}_{\mc{I}}  \}^{\circ}
)$.
Moreover, the monoidal structure is closed with the following
linear implication space so that
the forgetful to $\opTKersfinomg$ preserves the structure: 
$$(\mc{X}, U, R) \multimap (\mc{Y}, V, S)
= (\mc{X} \multimap \mc{Y}, (U \multimap S)^{\circ},
(U \multimap S)^{\circ \circ})$$
\end{prop}
\begin{proof}
By Proposition 61 of Hyland-Schalk \cite{HSha} for 
a general $\mc{C}$ with stable orthogonality,
which proposition is
applicable to $\mc{C}=\opTKersfinomg$ thanks to Lemma \ref{impder}.
\end{proof}

\smallskip

\begin{prop}[product and coproduct in $\Ti{\opTKersfinomg}$] \label{tipcp}~
The tight category $\Ti{\opTKersfinomg}$ has the following product and coproduct so that the forgetful functor to $\opTKersfinomg$ preserves the
structures:
$$\begin{aligned}
(\mc{X}, U, R) \,  \& \,  (\mc{Y}, V, S)
& = 
(\mc{X} \amalg \mc{Y}, U \& V, (U \& V)^\circ) \\
(\mc{X}, U, R) \,  \oplus \,  (\mc{Y}, V, S)
& = 
(\mc{X} \amalg \mc{Y}, (R \oplus S)^\circ, R \oplus S)
 \end{aligned}$$
Note by the self involution of the glueing (acting *, hence identity
on the first component,
and flipping second and third components), $\&$ and $\oplus$
are mutually definable each other.
\end{prop}
\begin{proof}
By Hyland-Schalk's Proposition 63 in \cite{HSha}
for a general ${\cal C}$ with stable orthogonality
because their presupposition  
positivity (res. negativity) of the projection (resp. injection)
of product (resp. coproduct) is entailed from
our reciprocal orthogonality condition
by Remark \ref{roipn}. 
\end{proof}

\smallskip

% $\ort{f^* \Comp y^*}{R^*} \epsilon_R \Comp (R^* \otimes u)
% \quad \mbox{and} \quad 
% \ort{u}{R}{\epsilon_R \Comp ((f^* \Comp y^*) \otimes R)}$, 
% 

% $\ort{f^* \Comp y^*}{R^*} \epsilon_R \Comp (R^* \otimes u)
% \quad \mbox{and} \quad 
% \ort{u}{R}{\epsilon_R \Comp ((f^* \Comp y^*) \otimes R)}$, 
% 

Finally, exponential structure for the tight category $\Ti{\opTKersfinomg}$
is obtained as an instance of Hyland-Schalk general construction for
$\Ti{\mc{C}}$ in \cite{HSha}.
\begin{prop}[exponential comonad on $\Ti{\opTKersfinomg}$] \label{expcmti}
$\Ti{\opTKersfinomg}$ has the following exponential comonad
so that the forgetful to $\opTKersfinomg$ preserves the structure:
$$
!(\mc{X}, U, R) =
(!\mc{X}, \natkappa{\mc{X}}{U}^{\circ \circ}, \natkappa{\mc{X}}{U}^{\circ})
$$
\end{prop}
\begin{proof}
By Theorem 65 of \cite{HSha}
for a general monoidal category with a stable orthogonality
because their presupposition for the theorem (described below)
is automatically derived by Remark \ref{roipn}:
All the structure
maps $\di{}, \stor{}, \wk{}, \con{}$ for linear exponential comonad
and all maps of
 the form ?
are positive for $\mathsf{k}$ (i.e., with respect to
$\forall U \, \,  \natkappa{\mc{X}}{U}$ and $\mc{C}(I,I)$)
and the product projections are focused.
\end{proof}

\subsubsection{$\Pcoh$ and its Equivalence 
to $\Ti{\opTKersfinomg}$}

% \noindent (implication) \\
% - Given $u: I \longrightarrow R$,
% $v: S \longrightarrow I$ and
% $f: R \longrightarrow S$,

% $\ort{u}{R}{v \Comp f}$
% and 
% $\ort{f \Comp u }{S}{v}$
% imply 
% $\ort{\hat{f}}{R^* \otimes S}{u^* \otimes v}$,
% where $\hat{f}: I \longrightarrow R^* \otimes S $
% is a transpose of $f: I \otimes  R  \longrightarrow S$.

\begin{defn}[$\Pcoh$ \cite{DE, Crubille}]{\em
The definition of the 
Danos-Ehrhard's category $\Pcoh$
of {\em probabilistic coherent spaces}
starts with the {\em inner product} and the {\em polar}:
}\end{defn}

\noindent (inner product)
$\inprocoh{x}{x'} := \sum_{a \in A} x_a x'_a$, for $x, x' \subseteq
\mathbb{R}_{+}^{A}$ with a countable set $A$.

\noindent (polar)
$P^\bot \! := \{ x' \in \mathbb{R}_+^A \mid \forall x \in P \, 
\inprocoh{x}{x'} \leq 1 \}$ for $P \subseteq \mathbb{R}_+^A$. 

\smallskip
Then $\Pcoh$ is defined as follows:

\noindent (object)
$X=(\abs{\! X \!},\Po{X})$, 
where $\abs{\! X \!}$ is a countable set, $\Po{X} \subseteq
 \mathbb{R}_+^{\abs{\, X \,}}$ such that $\Po{X}^{\bot \bot} \subseteq
 \Po{X}$, and $0 <  \sup \{ x_a  \mid  x \in \Po{X} \}
< \infty$ for all $a \in \abs{\! X \!}$.

\smallskip

\noindent (morphism)
A morphism from $X$ to $Y$ is an element $u \in \Po{(X \otimes
Y^\bot)^\bot}$,
which can be seen as a matrix $(u)_{a \in \abs{X}, b \in \abs{Y}}$
of columns from $\abs{\! X \!}$ and of rows from $\abs{\! Y \!}$.
Composition is the product of two matrices such that
$(uv)_{a, c}= \sum_{b \in \abs{Y}} u_{a,b} v_{b,c}$
for $u: X \longrightarrow Y$ and $v: Y \longrightarrow Z$.

\smallskip

\noindent(dual) 
$X^\bot =(\abs{\! X \!},\Po{X}^\bot)$ and
$u^\bot\in \Pcoh(Y^\bot, X^\bot)$ is the transpose of
a matrix $u \in \Pcoh(X,Y)$.

\smallskip

\noindent(tensor $\otimes$) \\
$X \otimes Y = (\abs{\! X \!} \times \abs{\! Y \!}, \{ x \otimes y \mid x \in
   \Po{X} \, \,  y \in \Po{Y}   \}^{\bot \bot})$. \\
For $u \in \Pcoh(X_1,Y_1)$ and $v \in \Pcoh(X_2,Y_2)$,
$u \otimes v \in \Pcoh(X_1 \otimes X_2,Y_1 \otimes Y_2)$ is 
$(u \otimes v)_{(a_1,a_2), (b_1,b_2)}= u_{a_1,b_1} v_{a_2,b_2}$.

\smallskip

\noindent (product $\&$)\\
$X_1 \& X_2 = (\abs{\! X_1 \!} \biguplus \abs{\! X_2 \!}, \{ x \in \mathbb{R}_+^{\biguplus_i
   \abs{X_i}} \mid  \forall i \, \, \pi_i(x) \in \Po{(X_i)}  \} )$, \\
where $\pi_i(x)_a$ is $x_{(i,a)}$.
$\Po{(X_1 \& X_2)}$ becomes automatically closed under the bipolar.

\smallskip

\noindent (exponential $\!$) 
The original definition using finite multisets is
rewritten by the exponential monoid and the counting function
in Section \ref{subsecems} of this paper. \\
$\! X = ( \, \abs{\! X \!}_e, \{ x^! \mid x \in \Po{X} \}^{\bot \bot} \, ) $,
where 
\begin{align} 
\textstyle x^! (\msbf{a}):= \prod \limits_{a \in \abs{X}}
 x_a^{\msbf{a}(a)} \label{defbangpcoh}
\end{align}
This is well defined as the $\msbf{a}$'s support
set $\{ a \in \abs{\! X \!} \, \, \mid \, \, \msbf{a}(a) \not =0 \}$
is finite. When $\msbf{a}$ is explicitly written by
$\msbf{a}=a_1 \cdots a_k$
with $a_i \in \abs{\! X \!}$,  
$x^! (\msbf{a})= \prod_{i=1}^k x_{a_i}$ since
each $a_i$ has $\msbf{a}(a_i)$-times
multiplicity in $\msbf{a}$.

\bigskip

\noindent
For $t \in \Pcoh (X,Y)$, $! t \in \Pcoh (!X,!Y)$
is defined by 
\begin{align}
\textstyle (!t)_{\msbf{a},\msbf{b}}:=
\sum \limits_{\msbf{c} \in L(\msbf{a}, \msbf{b})}
\frac{\msbf{b}!}{\msbf{c}!}
t^{\msbf{c}} \label{impexphcst}
\end{align}
in which for $\msbf{a} \in \abs{X}_e$ and  $\msbf{b} \in
\abs{Y}_e$: \\
-- $L(\msbf{a},\msbf{b}):=
\left\{ \msbf{c} \in (\abs{\! X \!} \times \abs{\! Y \!})_e \mid
\begin{array}{l}
\forall \, a \in \abs{X} \sum_{b \in \abs{Y}} \msbf{c}((a,b)) = \msbf{a}
(a)  \\
\forall \, b \in \abs{Y} \sum_{a \in \abs{X}} \msbf{c}((a,b)) = \msbf{b}(b)
\end{array}  
 \right\}$ \\
--
$\msbf{b} ! :=
\prod_{b \in \abs{Y}} \msbf{b}(b)!$ and 
$\msbf{c} ! :=\prod_{(a, b) \in \abs{X} \times \abs{Y}}
\msbf{c}((a,b))!$

\smallskip

\noindent 
To be explicit , when $\msbf{a}$ and 
$\msbf{b}$ are  given explicitly by
$\msbf{a}=a_1 \cdots a_n$
and $\msbf{b}=b_1 \cdots b_n$,
(\ref{impexphcst}) is written 
\begin{align}
\textstyle (!t)_{\msbf{a}, \msbf{b}} & 
 = \textstyle \sum \limits_{ \sigma \in \mathfrak{S}_n / S_{\vec{\msbf{a}}}}
\, \, \prod_{i=1}^n t_{a_{\sigma(i)},b_{i}} \label{exphcst}
 \end{align}
in which $S_{\vec{\msbf{a}}}$ denotes the {\em stabiliser subgroup}  of
$\mathfrak{S}_n$ at $\vec{\msbf{{a}}}:=(a_1,
 \ldots , a_n)$
defined by 
\begin{align*}
S_{\vec{\msbf{{a}}}} :=  \{ \sigma \in \mathfrak{S}_n \mid a_i = a_{\sigma(i)}
  \, \,   \forall i =1, \ldots, n \}  
\end{align*}
Note that the definition (\ref{exphcst}) does not depend on the ordering
$\vec{\msbf{{a}}}$ of $\msbf{{a}}$ for the stabiliser subgroup
$S_{\vec{\msbf{a}}}$ as 
$\mathfrak{S}_n / S_{\sigma(\vec{\msbf{a}})} =
\sigma(\mathfrak{S}_n / S_{\vec{\msbf{a}}})$ for any permutation 
$\sigma \in \mathfrak{S}_n$, hence its action on $\msbf{a}$
is well defined.

\bigskip

The category $\Pcoh$ turns out to be equivalent to the 
*-autonomous $\Ti{\opTKersfinomg}$ with the linear exponential comonad.
The equivalence is by the following Theorem 
\ref{eqpcs} and Proposition \ref{comoeqpcs}.

\begin{thm}[equivalence of $\Pcoh$] \label{eqpcs}
The tight orthogonality category $\Ti{\opTKersfinomg}$
is equivalent to the category $\Pcoh$.
\end{thm} 
\begin{proof}
 The key property for the equivalence is
that the measures (i.e., the homset $\opTKeromg(\mc{X},\mc{I})$)
and the measurable functions (i.e., the homset
 $\opTKeromg(\mc{I},\mc{X})$) become
isomorphic in $\opTKeromg$ by virtue of the involution $(~)^*$,
and furthermore in $\opTKersfinomg$
they both collapse to bounded functions from $X$
to $\mathbb{R}_{+}$, hence residing in $\mathbb{R}_{+}^X$
in $\Pcoh$.\\
The orthogonality $<\,,\,>$ in $\Pcoh$ coincides with $< \, \mid \, >$
in $\opTKersfinomg$, as the integral of Definition \ref{ortint}
collapses to the sum in the subcategory of discrete measurable spaces. \\
An object $X=(\abs{\! X \!},\Po{X})$ in $\Pcoh$ corresponds one to one
to the object
$\bm{\mc{X}}=(X, \Po{X}, (\Po{X})^\circ )$ in $\Ti{\opTKersfinomg}$,
preserving the involution $(~)^\bot$.
Every morphism from $X$ to $Y$ in $\Pcoh$ is by definition
an element $\Po{(X^\bot \parr Y)}$, which is
the second component of $\bm{\mc{X}} \multimap \bm{\mc{Y}}$
in $\Ti{\opTKersfinomg}$.
Composition of $\Pcoh$ is the product of matrices, same as  $\opTKersfinomg$.
E.g., in particular their map ${\sf fun}(u): \Po{X} \longrightarrow \Po{Y}$ for
$u \in \Pcoh (X,Y)$ (cf. Section 1.2.2 \cite{DE})
is written in $\opTKersfinomg$ simply by  
${\sf fun} (x) = u^* x \in \Po{Y}$.
Since the tensor and the additive structures are direct,
only the exponential structure is checked on (i)
objects and on (ii) morphisms. In the both levels,
Danos-Ehrhard's exponential construction in $\Pcoh$ turns out to coincide with that
of Hyland-Schalk for double glueing applied to our $\opTKersfinomg$:  \\
(i) It is shown that
$(-)^{!}: \mathbb{R}^{\abs{X}}
\longrightarrow \mathbb{R}^{\abs{X}_e}$ of (\ref{defbangpcoh})
in $\Pcoh$
defines the same map as $\natkappa{\mc{X}}{-}: 
\opTKersfinomg(\mc{I}, \mc{X})
\longrightarrow 
\opTKersfinomg(\mc{I}, \mc{X}_e)$
of Definition \ref{defnatkappa}:
For $f \in \mathbb{R}_+^{\abs{X}}$,
$f^{!}(\msbf{a})=
\prod_{a \in \abs{X}} f(a)^{\msbf{a}(a)}$. On the other hand in $\opTKersfinomg$,
for $\msbf{a}=a_1 \cdots a_k \in X_e$,
$\natkappa{X}{f} ((*, \msbf{a}))= \prod_{i=1}^k f(*, a_i)$.
Thus $f^!= \natkappa{\mc{X}}{f}$, for which the right $f: \mc{I} \rightarrow \mc{X}$ is identified with a measurable function on $\mc{X}$.

\noindent 
(ii) First note  $\forall \msbf{c} \in L(\msbf{a},\msbf{b}) \, \,
\cunt{\abs{X} \times \abs{Y}}(\msbf{c})=
 \cunt{\abs{X}}(\msbf{a})=
\cunt{\abs{Y}}(\msbf{b})$, whose number is denoted by $r$ such that
$\msbf{a}= a_1 \cdots a_r$ and $\msbf{b} =b_1 \cdots b_r$.
Then in $\opTKersfinomg$, 
\begin{align*}
t_e (\msbf{a},\msbf{b}) &  =t_e (a_1 \cdots a_r, \, b_1 \cdots b_r)
\\ & =
t_e^\bullet ( \FI(a_1 \cdots a_r) , \, (b_1, \ldots, b_r)) 
\tag*{by $(~)_e$ in $\opTKersfinomg$}\\ 
& 
\textstyle =
t^r ( \bigcup \limits_{\sigma \in \mathfrak{S}_r} 
\{ (a_{\sigma(1)}, \ldots, a_{\sigma(r)}) \}, \, (b_1, \ldots, b_r))
\tag*{by the def of $F$} \\ 
& \textstyle
=
t^r ( \! \! \biguplus \limits_{\sigma \in \mathfrak{S}_r/S_{\vec{\msbf{a}}}}
\! \!
\{ (a_{\sigma(1)}, \ldots, a_{\sigma(r)}) \}, \, (b_1, \ldots, b_r)) 
\tag*{thanks to the quotient by $S_{\vec{\msbf{a}}}$}\\
& 
=
\textstyle \sum \limits_{\sigma \in \mathfrak{S}_r/S_{\vec{\msbf{a}}}}
 \! t^r ((a_{\sigma(1)}, \ldots, a_{\sigma(r)}), \, (b_1, \ldots, b_r)) 
\tag*{by $\sigma$-additivity} \\
& 
=
\textstyle \sum \limits_{\sigma \in \mathfrak{S}_r/S_{\vec{\msbf{a}}}}
\prod \limits_{i=1}^r 
 t (a_{\sigma(i)}, \, b_i)
\tag*{by the def of $t^r$ }
\end{align*}
This has shown that $t_e (\msbf{a},\msbf{b})=(! t)_{\msbf{a},\msbf{b}}$.
\end{proof}

\bigskip

The comonad structure of the exponential of $\Pcoh$
is given in \cite{DE} as follows: 

\noindent (dereliction)
$(\di{X})_{\msbf{a}, a} :=
\delta_{\msbf{a}, [a]}$ where $\di{X} \in \Pcoh(!X, X)$

\smallskip
\noindent (storage)
$(\stor{X})_{\msbf{a}, M} := 
\delta_{\msbf{a}, \,  \sum M }$
where $\stor{X} \in \Pcoh(!X,
!!X)$

\bigskip

\begin{prop} \label{comoeqpcs}
The comonad structure of $\Pcoh$
is the discretisation of that of
$\Ti{\opTKersfinomg}$ 
under the equivalence of Theorem \ref{eqpcs}.
\end{prop}
\begin{proof}
The dereliction and the storage
become directly the discretisation of the corresponding
maps in $\opTKersfin$ (hence of $\Ti{\opTKersfinomg}$),
defined respectively in
Propositions \ref{natder} and \ref{stor}.
\end{proof}

\bigskip

\begin{rem}[The opposite $\opTKersfin$ coincides with
$\Pcoh$'s left enumeration] \label{whyop}
{\em Our choice taking the {\em opposite} of $\TKersfin$ starting from 
Section \ref{comonadTKer} turns out to yield
Danos-Ehrhard's choice of {\em left} enumeration
in formalising the exponential in (\ref{exphcst}).
When the right enumeration is chosen oppositely,
there arises another exponential, say $\natural$, 
\begin{align*}
\textstyle
(\nat{t})_{\msbf{a}, \msbf{b}} := \sum_{\msbf{c} \in L(\msbf{a}, \msbf{b})}
\frac{\msbf{a}!}{\msbf{c}!}
\, t^{\rho} = \sum_{ \sigma \in \mathfrak{S}_n / S_{\vec{\msbf{b}}}}
\, \, \prod_{i=1}^n t_{a_i,b_{\sigma(i)}}
\end{align*}
Compare the first and second formulas respectively with (\ref{impexphcst})
and (\ref{exphcst}) to see the opposite enumeration.
It holds 
$(\nat{t})_{\msbf{a}, \msbf{b}}
=\frac{\msbf{a}!}{\msbf{b}!}  \, \, \,
(\bang{t})_{\msbf{a}, \msbf{b}}$.
The two exponentials become isomorphic \cite{ThomasMail},
by the following natural isomorphism in terms of
the multinomial coefficient $\mn:
\, ! \longrightarrow \natural$, defined by   
\begin{align*}
\quad \quad 
(\mn_X)_{\msbf{a}, \msbf{b}} = \mn (\msbf{a}) \,  \delta_{\msbf{a}, \msbf{b}}
\end{align*}
Remind the multinomial coefficient of $\msbf{a} \in
X_e \cap X^{(n)}$
is defined by $\mn (\msbf{a}) := \frac{n !}{\prod_{a \in X^{(n)}}
 \msbf{a} (a) !}$, 
which number is equal to the cardinality
 $\abs{\mathfrak{S}_{\vec{\msbf{a}}} / S_{\vec{\msbf{a}}}
}$ of the quotient independently of the ordering $\vec{\msbf{a}}$ of $\msbf{a}$.

}\end{rem}

\section{Conclusion}
This paper offers four main contributions: 
\begin{enumerate}
\item[(i)] Presenting a monoidal category $\TKersfin$ of
s-finite transition kernels
between measurable spaces after Staton \cite{Stat}, with countable biproducts.
Showing a construction of exponential kernels in
$\TKersfin$
 by accommodating
the exponential measurable spaces for counting process
into the category.

\item[(ii)]
Constructing a linear exponential comonad over $\opTKersfin$,
modelling the exponential modality
in linear logic.
Although this initiates a continuous linear exponential comonad
employing a general measure theory, 
though we leave it a future work 
on any monoidal closed structure
inside $\TKersfin$ required for modelling the multiplicative fragment
of the logic.

\item[(iii)]
Giving a measure theoretic instance of Hyland-Schalk orthogonality
in terms of an integral between measures and measurable functions.
The instance is inspired by the contravariant equivalence between
	    $\TKer$ of the transition kernels and $\Mes$ of measurable
functions, and realised by adjunction of a kernel acting on both
	   sides.
We examine a monoidal comonad 
in the double glueing $\GopTKersfin$
inheriting from $\opTKersfin$ of (i)
and also that in the slack orthogonality subcategory 
$\Sla{\opTKersfin}$.

\item [(iv)]
(Discretisation of (i), (ii) and (iii)): \\
Obtaining a dagger compact closed category $\TKersfinomg$
when restricting $\TKersfin$ of (i) to the countable measurable spaces.
We show an equivalence of the tight orthogonality category $\Ti{\opTKersfinomg}$
to $\Pcoh$ of probabilistic
coherent spaces by virtue of the discrete collapse of
the orthogonality of (ii) into the linear duality of $\Pcoh$.
\end{enumerate}

%\smallskip

We now discuss some future directions.
Our categories $\TKer$ with countable biproducts and 
$\TKersfin$ with tensor
are inspired from the standard measure-theoretic formalisation of
probability theory,
and similarly the linear exponential comonad over $\opTKersfin$
from the counting process for exponential measurable spaces.
We believe our semantics of transition kernels
will provide a general tool for semantics of 
 higher order probabilistic programming languages
such as probabilistic PCF \cite{DE, EPT},
of which $\Pcoh$ is a denotational semantics.
We need to examine a concrete example making
continuous Markov kernels indispensable (rather than
discrete Markov matrices) for interpreting probabilistic
computational reductions
as a stochastic process. 
For this, any monoidal
closed structure fundamental to denotational semantics
needs to be explored in continuous measure spaces.
Recent development \cite{EPT, CruLICS} on CCC extension induced by $\Pcoh$
for continuous probabilities may be seen as a mutual construction
of our construction because ours starts with the continuity 
to obtain $\Pcoh$ as a discretisation.
% On the other hand as a related work to our Section \ref{DGPCS},  
% a sketch is announced in the talk \cite{Pagslide}
% on constructing $\Pcoh$ as a double glueing
% in terms of $\zeroinf$-weighted $\Rel$ for the analytic exponential.

An important future work %for another direction 
is to connection
to Staton' s denotational semantics  \cite{Stat} for
commutativity of first-order probabilistic functional programming,
in which s-finiteness of kernels characterises commutativity of
programming languages.
We are interested in how our trace structure for feedback and
probabilistic iteration
(cf. Remark \ref{UDC}) may play any role in his probabilistic data flow
analysis using categorical arrows.
That is, a direction towards a probabilistic Geometry of Interaction
employing the
continuous categories of the present paper. \\
Another promising further study is association
with point process monad in \cite{DaSta} using
distribution between Giry monad and multisets.
This may provide a general categorical understanding  
how our measure theoretic commutative monoid, seen as 
counting process, yields
the exponential comonad for linear logic.

% preserved under categorical composition in a monoidal category.
% This may shed a certain light to data flow analysis of
% probabilistic programming languages,
% using categorical arrows with certain feed back structure
% inherent our categorical structure of Remark \ref{UDC}.

%Simulating samples from arbitrary probability distributions
%because Poisson process () is interoreted as counting process
%in queing models.
%Fundationally probabilstic monad using exponentials
%practically monte carrlo simulation doing with counting process.


%\section*{Acknowledgment}

Relating our model to Girard's coherent Banach spaces \cite{GirBan}
on one hand involves analysing the contravariant equivalence
of Proposition \ref{contraeq} under
certain constraints required from logical and type systems.
On the other hand, the double glueing in Section \ref{dgort} 
will give a direct bridge to 
the duality of coherent Banach spaces,
employing a (variant of) a Chu construction,
which is known as another instance of Hyland-Schalk orthogonality.

After the submission of the earlier version of the paper, series of
pioneering works are published \cite{EhrCone, Geoff, Paquet}
on continuous exponentials for higher ordered programming.
Their approach to measure-theoretic 
continuity would warrant our future work mentioned above. 
Especially, Ehrhard's measurable exponential \cite{EhrCone}
in cones, giving a continuous extension of discrete probability,
could be quite beneficial.

\section*{Acknowledgment}
The author is deeply grateful to the anonymous reviewers
for their carefully reading and constructive feedback,
which have improved the paper in 
innumerable ways. 
%and saved the paper from certain technical defects,
%those that inevitably remain are entirely the author's own. %responsibility.

%\smallskip





\end{document}